\RequirePackage{silence}
\documentclass[twocolumn,prx,aps,superscriptaddress,showpacs,10pt]{revtex4-2}

\pdfoutput=1

\usepackage[T1]{fontenc} 
\usepackage[english]{babel} 
\usepackage[utf8]{inputenc}

\PassOptionsToPackage{numbers,sort&compress}{natbib}

\providecommand{\njp}{New J.\ Phys.}
\providecommand{\pra}{Phys.\ Rev.\ A}

\usepackage{booktabs}
\usepackage{multirow}
\usepackage{hhline}
\usepackage[table,dvipsnames]{xcolor}
\usepackage[colorlinks=true,
			hypertexnames=false
			]{hyperref}
\usepackage{diagbox}
\usepackage{graphicx}
\hypersetup{
	linkcolor={black!30!blue},
	citecolor={black!30!blue},
	urlcolor={black!30!blue}
}
\PassOptionsToPackage{linktocpage}{hyperref} 

\usepackage{amsmath,amssymb,mathtools,amsthm,dsfont}
\usepackage{stmaryrd}
\SetSymbolFont{stmry}{bold}{U}{stmry}{m}{n}

\usepackage{thm-restate}

\usepackage{enumitem}
\usepackage{subfigure}
\usepackage[normalem]{ulem}

\usepackage{optidef}
\usepackage{algorithm,algorithmicx}
\usepackage[noend]{algpseudocode}
\usepackage{tikz}
\usetikzlibrary{positioning}
\usetikzlibrary{bbox}

\usepackage{cleveref}
\crefname{figure}{figure}{figures}

\usepackage[withpage,nolist]{acronym}
\makeatletter
\AtBeginDocument{%
 \renewcommand*{\AC@hyperlink}[2]{%
   \begingroup
     \hypersetup{hidelinks}%
     \hyperlink{#1}{#2}%
   \endgroup
 }%
}
\makeatother

\newtheorem{theorem}{Theorem}
\numberwithin{theorem}{section}
\newtheorem{corollary}[theorem]{Corollary}

\newtheorem{lemma}[theorem]{Lemma}

\newtheorem{proposition}[theorem]{Proposition}

\newtheorem{observation}[theorem]{Numerical Observation}

\newtheorem{definition}[theorem]{Definition}

\theoremstyle{definition}
\newtheorem{remark}{Remark}

\newcommand{\e}{\ensuremath\mathrm{e}}
\renewcommand{\i}{\ensuremath\mathrm{i}}
\newcommand{\rmd}{\ensuremath\mathrm{d}}

\DeclareMathOperator{\LandauO}{\mathrm{O}}

\DeclareMathOperator{\Tr}{Tr}

\DeclareMathOperator{\Id}{Id}
\DeclareMathOperator{\supp}{supp}

\newcommand{\fro}{\mathrm{F}}

\DeclareMathOperator{\conv}{conv}

\DeclareMathOperator{\Herm}{Herm} 

\NewDocumentCommand\Cl{mg}{
    \ensuremath{\mathrm{Cl}_{#1}\IfNoValueTF{#2}{}{(#2)}}%
}
\NewDocumentCommand\HW{mg}{
    \ensuremath{\mathrm{HW}_{#1}\IfNoValueTF{#2}{}{(#2)}}%
}

\newcommand{\CC}{\mathbb{C}}
\newcommand{\RR}{\mathbb{R}}

\newcommand{\1}{\mathds{1}}

\newcommand{\PP}{\mathbb{P}}
\newcommand{\FF}{\mathbb{F}}

\newcommand{\R}{\mathbb{R}}

\newcommand{\mc}[1]{\mathcal{#1}}

\newcommand{\argdot}{{\,\cdot\,}}
\renewcommand{\vec}[1]{\boldsymbol{#1}}

\DeclarePairedDelimiterX{\abs}[1]{\lvert}{\rvert}{%
  \ifblank{#1}{\,\cdot\,}{#1}
}   

\DeclarePairedDelimiterX\norm[1]\lVert\rVert{%
  \ifblank{#1}{\,\cdot\,}{#1}
}   

\newcommand{\lonenorm}[2][1]{\norm{#2}_{\ell_{#1}}}
\newcommand{\linfnorm}[2][\infty]{\norm{#2}_{\ell_{#1}}}
\DeclarePairedDelimiterX{\iiiNorm}[1]{\lvert}{\rvert}{%
  \delimsize\lvert\delimsize\lvert#1\delimsize\rvert\delimsize\rvert%
}

\DeclarePairedDelimiterXPP\snorm[1]{}\lVert\rVert{_\infty}{\ifblank{#1}{\,\cdot\,}{#1}}   

\DeclarePairedDelimiterXPP\twonorm[1]{}\lVert\rVert{_2}{\ifblank{#1}{\,\cdot\,}{#1}}   

\DeclarePairedDelimiterXPP\trnorm[1]{}\lVert\rVert{_1}{\ifblank{#1}{\,\cdot\,}{#1}}   

\DeclarePairedDelimiterXPP\fnorm[1]{}\lVert\rVert{_{\fro}}{\ifblank{#1}{\,\cdot\,}{#1}}   

\DeclarePairedDelimiterXPP\dnorm[1]{}\lVert\rVert{_\diamond}{\ifblank{#1}{\,\cdot\,}{#1}}   

\DeclarePairedDelimiterXPP\cbnorm[1]{}\lVert\rVert{_\mathrm{cb}}{\ifblank{#1}{\,\cdot\,}{#1}}   
\DeclarePairedDelimiterXPP\onenorm[1]{}\lVert\rVert{_{1\rightarrow 1}}{\ifblank{#1}{\,\cdot\,}{#1}}   
\DeclarePairedDelimiterXPP\ddnorm[1]{}\lVert\rVert{_{\diamond\rightarrow \diamond}}{\ifblank{#1}{\,\cdot\,}{#1}}   
\DeclarePairedDelimiterXPP\ssnorm[1]{}\lVert\rVert{_{\infty\rightarrow\infty}}{\ifblank{#1}{\,\cdot\,}{#1}}   

\providecommand\given{}
\newcommand\SetSymbol[1][]{%
  \nonscript\:#1\vert
  \allowbreak
  \nonscript\:
  \mathopen{}}
\DeclarePairedDelimiterX\Set[1]\{\}{%
  \renewcommand\given{\SetSymbol[\delimsize]}
  #1
}

\DeclarePairedDelimiterX\innerp[2]{\langle}{\rangle}{%
  \ifblank{#1}{\,\cdot\,}{#1} , \ifblank{#2}{\,\cdot\,}{#2}%
}

\DeclarePairedDelimiterX\braket[2]{\langle}{\rangle}%
  {#1\kern0.15ex\delimsize\vert\kern0.15ex\mathopen{}#2}

\DeclarePairedDelimiterX\ketbra[2]{\vert}{\vert}%
  {#1\kern0.15ex\delimsize\rangle\delimsize\langle\kern0.15ex\mathopen{}#2}

\DeclarePairedDelimiterX\sandwich[3]{\langle}{\rangle}%
  {#1\,\delimsize\vert\kern0.15ex\mathopen{}#2\kern0.15ex\delimsize\vert\kern0.15ex\mathopen{}#3}

\DeclarePairedDelimiterX\obraket[2]{(}{)}%
  {#1\kern0.15ex\delimsize\vert\kern0.15ex\mathopen{}#2}

\DeclarePairedDelimiterX\oketbra[2]{\vert}{\vert}%
  {#1\kern0.15ex\delimsize)\delimsize(\kern0.15ex\mathopen{}#2}

\DeclarePairedDelimiterX\osandwich[3]{(}{)}%
  {#1\,\delimsize\vert\kern0.15ex\mathopen{}#2\kern0.15ex\delimsize\vert\kern0.15ex\mathopen{}#3}

\newcommand{\myleft}{\mathopen{}\mathclose\bgroup\left}
\newcommand{\myright}{\aftergroup\egroup\right}

\newcommand{\func}[1]{{\ensuremath{\mathsf{#1}}}}

\newcommand{\poly}{\func{poly}}

\newcommand{\Hz}{\mathrm{Hz}}
\newcommand{\second}{\mathrm{s}}

\definecolor{green}{HTML}{009E73}
\definecolor{darkBlue}{HTML}{0072B2}
\definecolor{lightBlue}{HTML}{56B4E0}
\definecolor{myRed}{HTML}{CA181D}
\definecolor{myOrange}{HTML}{E69F00}
\definecolor{myYellow}{HTML}{F0E442}

\newcommand{\tikzLabel}[1]{{\color{#1}\begin{tikzpicture}
	 \path [use as bounding box] (-0.4cm,-0.06cm) rectangle (0.4cm,0.05cm);
	 \node (a) at (0,0){};
	 \node[right = 0.25cm of a] (o){};
	 \node[left = 0.25cm of a] (w){};
	 \node[above = 0.05cm of a] (n){};
	 \node[below = 0.05cm of a] (s){};
     \draw[line width=0.5mm] (n) -- (s);
	 \draw[line width=0.5mm] (w) -- (o);
\end{tikzpicture}}}

\newcommand{\dd}{r}
\newcommand{\rr}{s}
\newcommand{\nTro}{n_{\mathrm{Tro}}}
\newcommand{\pp}{\kappa}
\newcommand{\ns}{n_L}
\newcommand{\kk}{\ell}

\newcommand{\WG}{\mathcal{W}(\vec J)}
\newcommand{\WGG}[2]{\mathcal{W}(\vec J)^{(#1 \times #2)}}
\newcommand{\WPfull}{W}
\newcommand{\WP}[2]{W^{(#1 \times #2)}}
\newcommand{\WC}[2]{W(\vec J)^{(#1 \times #2)}}
\newcommand{\WCfull}{W(\vec J)}

\newcommand{\EP}[2]{E^{(#1 \times #2)}}
\newcommand{\EC}[2]{E(\vec J)^{(#1 \times #2)}}

\DeclareMathOperator*{\prodl}{\overleftarrow{\prod}}
\DeclareMathOperator*{\prodr}{\overrightarrow{\prod}}
\newcommand{\av}{\mathrm{av}}
\newcommand{\AV}{\mathrm{AV}}
 \newcommand{\err}{\mathrm{err}}

\newcommand{\hhu}{Heinrich Heine University Düsseldorf, Faculty of Mathematics and Natural Sciences, Düsseldorf, Germany
}
\newcommand{\tuhh}{Hamburg University of Technology, Institute for Quantum Inspired and Quantum Optimization, Hamburg, Germany
}

\begin{document}

\title{General, efficient, and robust Hamiltonian engineering}

\author{Pascal Baßler}
\email{pascal.bassler@tuhh.de}
\affiliation{\tuhh}
\affiliation{\hhu}
\author{Markus Heinrich}
\email{markus.heinrich@uni-koeln.de}
\affiliation{\hhu}
\affiliation{Institute for Theoretical Physics, University of Cologne, Cologne, Germany}
\author{Martin Kliesch}
\email{martin.kliesch@tuhh.de}
\affiliation{\tuhh}

\begin{abstract}
Implementing the time evolution under a desired target Hamiltonian is critical for various applications in quantum science. 
Due to the exponential increase in the number of parameters with system size and experimental imperfections, this task can be challenging in quantum many-body settings. 

We introduce an efficient and robust scheme to engineer arbitrary local many-body Hamiltonians. To this end, our scheme applies single-qubit $\pi$ or $\pi/2$ pulses to an always-on system Hamiltonian, which we assume to be native to a given platform. These sequences are constructed by efficiently solving a linear program (LP) which minimizes the total evolution time. In this way, we can engineer target Hamiltonians that are only limited by the locality of the interactions in the system Hamiltonian. Based on average Hamiltonian theory and using robust composite pulses, we make our schemes robust against errors, including finite pulse time errors and various control errors.

To demonstrate the performance of our scheme, we provide numerical simulations. In particular, we solve the Hamiltonian engineering problem on a laptop for arbitrary two-local Hamiltonians on a 2D square lattice with $196$ qubits in only $60$ seconds. Moreover, we simulate the engineering of general Heisenberg Hamiltonians from Ising Hamiltonians using imperfect single-qubit pulses for smaller system sizes and achieve a fidelity exceeding $99.9\%$, which is orders of magnitude better than non-robust implementations.
\end{abstract}
\maketitle

 \hypersetup{
	     pdfsubject = {Quantum computing},
	     pdfkeywords = {quantum, compiling, negative time, dynamical, decoupling, NMR, nuclear magnetic resonance, analog, analogue, engineering,
	     	 global, interaction, entangling, multi-qubit, many-body, gate, synthesis, unitary, all-to-all, restricted, connectivity, always-on, drift, product formula,
		     ion trap, Ising, Hamiltonian, superconducting, neutral atoms, Rydberg, simulation,
		     convex, optimization, linear program, LP, phase transition,
		     digital-analog, DAQC, SCROFULOUS, SCROBUTUS, conjugation, interaction picture, toggling-frame, twirl,
		     qubit, efficient, Wendel, convex analysis, stochastic geometry, polytope,
		     robust, pulses, average Hamiltonian theory, robustness, composite, finite pulse time, finite pulse duration
	     }
	    }

\section{Introduction}
\label{sec:intro}
Simulating a target Hamiltonian on a quantum system is a central problem in quantum computing, with applications in general gate-based and digital-analog quantum computing, as well as quantum simulations \cite{Somma2002,Blatt12,Wecker2015,Bernien2017,Zhang2017,Defenu2024} and quantum chemistry \cite{Kassal2011,Wecker2014,Olson2017}.
It is commonly believed that simulating the dynamics of a quantum system is one of the most promising tasks for showing a practical advantage of quantum computers over classical computers \cite{Kassal2011,Childs2018,Daley22PracticalQuantumAdvantage}, perhaps already on \ac{NISQ} devices \cite{Preskill18,Trivedi2022}.
Especially for the latter, it is essential to have an efficient and fast implementation of the target Hamiltonian due to short coherence times.

The idea to engineer a target Hamiltonian by interleaving the evolution under a fixed system Hamiltonian with single-qubit pulses originated in the \ac{NMR} community.
However, the first approaches require additional single-qubit pulses to decouple interactions, rescale them, and couple them again, resulting in long pulse sequences \cite{Leung2000,Leung2002,Dodd2002,Nielsen2002,Votto23UniversalQuantumProcessors}.
Recently, there has been a renewed interest in designing pulse sequences to change the effective dynamics governed by a given Hamiltonian.
In particular, there has been impressive progress in the design of robust global pulses, identical pulses on each qubit, to change the global properties of a given Hamiltonian \cite{Choi2020}. 
This approach has already been generalized to qudit systems \cite{Choi2017,zhou2024} and implemented in experiments with ultracold atomic Rydberg gas and NV centers in diamond \cite{Geier2021,Hengyun2020}.
However, Hamiltonian engineering schemes utilizing global pulses, such as Floquet Hamiltonian engineering methods \cite{Geier2021,Scholl2022}, cannot modify individual interaction terms.
In our work, we design robust sequences of local pulses to change any interaction coupling in a given Hamiltonian, with similar robustness properties as for global pulses \cite{Choi2020}.
Moreover, due to the local pulses in our method we are able to implement the dynamics of a much larger family of possible target Hamiltonians.
We believe, that neutral atom architectures based on different operating zones or in the weak coupling regime are well-suited for digital-analog methods as presented in our work \cite{Barredo2015,Bluvstein2022,Bluvstein_2023}.

So far, other approaches utilizing local pulses have certain limitations.
They either require NP-hard classical pre-processing to find the pulse sequence \cite{garcia2024}, rely on specific structures in the system Hamiltonian \cite{Hayes2014}, or require an infinite single-qubit gate set which might be a problem for the fast control electronics in an experiment \cite{figgatt_parallel_2019,lu_global_2019,grzesiak_efficient_2020}.
We provide a comparison of previous methods in \cref{tab:compare}.
These works are also limited to two-body interactions, and to the best of our knowledge, no general scheme has been developed for efficiently engineering many-body interactions individually.
Since there has been an increasing effort in realizing the latter in experiments \cite{Pachos2004,Buechler2007,Peng2009,Bernien2017,Zhang2022}, such a general scheme would allow to simulate quantum chemistry Hamiltonians with genuine many-body interactions.
Moreover, some recent Hamiltonian learning schemes rely on ``reshaping'' unknown many-body Hamiltonians to diagonal Hamiltonians which can be done efficiently and robust to errors with our proposed method~\cite{huang2023learning,ma2024learningkbodyhamiltonianscompressed,hu2025}.

\begin{table*}
	\begin{tabular}{|l|c|c|c|c|c|}		 \hline
		\textbf{Method}                     & \textbf{Target Hamiltonians} & \textbf{Runtime} & \textbf{Memory} & \textbf{Depth}   & \textbf{Robust} 	 \\\hline
		\citet{Choi2020}       & cannot modify individual interactions & $\LandauO (1)$ & $\LandauO (1)$             & $\LandauO (1)$   & yes \\
		DAQC \cite{garcia2024}              & 2-local, arbitrary     & $\LandauO (4^n)$ & $\LandauO (4^n)$           & $\LandauO (n^2)$ & no \\
		\citet{Hayes2014}       & 2-local, 1D spin chain & $\LandauO (n^3)$**  & $\LandauO (n^2)$    & $\LandauO (n^2)$ & no \\
		EASE gate \cite{grzesiak_efficient_2020} & 2-local, Ising type       & $\LandauO (n^3)$ &  $\LandauO (n^2)$     & $\LandauO (n)$*  & no \\
		\citet{Votto23UniversalQuantumProcessors} & 2-local, XY model & $\LandauO (n^3)$ & $\LandauO (n^3)$ & $\LandauO (n^3)$ & yes \\
		This work                 & arbitrary (locality preserving) & $\LandauO (\dd^3)$** & $\LandauO (\dd^2)$   & $\dd$            & yes \\ \hline
	\end{tabular}
	\\[.2em] \qquad
	\begin{minipage}{.68\linewidth}
	\flushleft
	* requires pulses of arbitrary angle \\
	** We have assumed a complexity of $\LandauO(m^3)$ for solving a linear program with $m$ variables \cite{Bach2025}.
	\end{minipage}
	\caption{
			 Comparison of Hamiltonian engineering methods for implementing a target Hamiltonian with $\dd$ interactions on $n$ qubits.
			 We compare approaches similar to ours with respect to classical runtime \& memory usage, the circuit depth (number of pulses) and the robustness.
			 For general $k$-local Hamiltonians we have $\dd = \LandauO (n^k)$. Thus, for a 2-local Hamiltonian we have $\dd = \LandauO (n^2)$ and for a 2-local Hamiltonian on a 1D spin chain we have $\dd = \LandauO (n)$.
			 Moreover, for a better comparison of methods we omit the overhead introduced by implementing non-commuting interactions, i.e.\ via product formulas, in the depth consideration.\newline
			 }
	\label{tab:compare}
\end{table*}

In previous works, we have proposed a Hamiltonian engineering method for Ising Hamiltonians based on linear programming and applied it to multi-qubit gate synthesis and compiling problems \cite{basler_synthesis_2023,basler_time-optimal_2024}.
A related linear program approach for commuting Ising and next-neighbour interactions with an efficient relaxation was introduced in Refs.~\cite{Bhole20RescalingInteractionsFor,Tsunoda20EfficientHamiltonianProgramming}.
In this work, we generalize the linear programming method to a significantly larger family of local Hamiltonians, make it substantially more efficient, and render it robust against dominant error sources.
More concretely, our method allows to efficiently engineer arbitrary target Hamiltonians, only limited by the locality of the system Hamiltonian.
To this end, free evolutions under the system Hamiltonian are interleaved with $\pi$ or $\pi/2$ pulses (i.e.\ Pauli or single-qubit Clifford gates).
With the \ac{LP} formulation, we find an \emph{exact} decomposition of the target Hamiltonian as a sum, where each term corresponds to single-qubit pulses and the corresponding evolution time under the system Hamiltonian.
With this exact decomposition, one can implement the target evolution utilizing standard Hamiltonian simulation methods such as Trotterization.
A crucial feature of our method is that it minimizes the total evolution time, leading to a fast implementation of the target Hamiltonian.
For larger systems, however, optimally solving the minimizing linear program is no longer efficiently possible.
We introduce a relaxation that only scales with the number of interaction terms and not directly with the system size, and thus, provides an efficient method to engineer Hamiltonians.
This relaxation still yields an exact decomposition of the target Hamiltonian, but the total evolution time may no longer be minimal.
The latter can, however, be decreased by expanding the search space for the relaxed problem, providing a trade-off between the runtime of the classical pre-processing and the evolution time of the implementation.

Concretely, for any target Hamiltonian with $\dd$ interaction terms and the same locality as the system Hamiltonian, our method efficiently finds an implementing single-qubit pulse sequence using a linear-in-$\dd$ number of layers.
To implement the target evolution with a product formula, we therefore also require $\LandauO(\dd)$ single-qubit pulse layers, where the implicit constant only depends on the chosen product formula.
Note, that previous approaches require $\LandauO(\dd^2)$ single-qubit pulse layers \cite{Leung2000,Leung2002,Dodd2002,Nielsen2002,Votto23UniversalQuantumProcessors} or are classically hard to solve \cite{garcia2024}.
Moreover, all of these methods except Ref.~\cite{Votto23UniversalQuantumProcessors} are not robust against any error source.
For our method we observe that the total evolution time to implement the dynamics under the target Hamiltonian scales only sublinear with the number of qubits.

We also investigate the effects of dominant experimental error sources and introduce a general framework to cancel them.
To this end, we generalize a method by \citet{Votto23UniversalQuantumProcessors} based on \ac{AHT} to mitigate these errors.
As a result, we are able to make our methods robust against finite pulse time and single-qubit rotation angle errors.
We leverage the generality of our method to combine it with robust composite pulses, making it robust against many experimental error sources.
This is especially beneficial in light of recent efforts which have demonstrated a remarkable single-qubit pulse fidelity of $10^{-6}$ utilizing fast robust composite pulses \cite{Leu2023}.

\begin{figure*}[ht]
	\centering
	\includegraphics[width=0.84\linewidth]{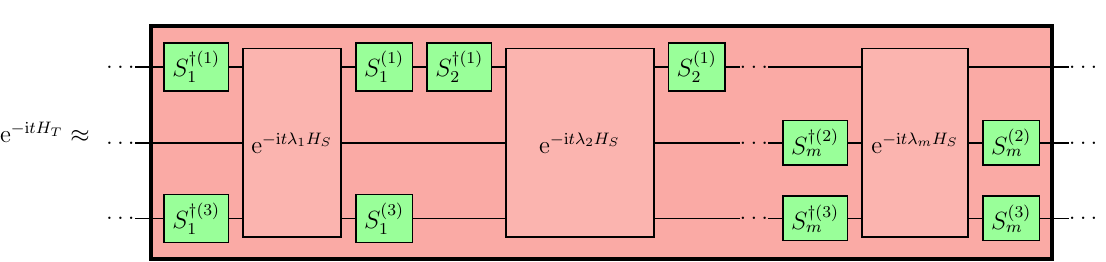}
	\caption{
			We engineer the target Hamiltonian $H_T$ by interleaving the natural dynamics of the quantum device, governed by the system Hamiltonian $H_S$, with layers of single-qubit $\pi$ or $\pi/2$ pulses, $\vec S_i = S_i^{(1)} \otimes \dots S_i^{(n)}$, as in \cref{eq:approximate_evolution}.
			The large red box highlights that $H_S$ is assumed to be always on.
			Our Hamiltonian engineering results are straightforward to implement:
			apply single-qubit pulse layers $\vec S_{i-1} \vec S_i^\dagger$, let the system evolve freely under $H_S$ for a duration $t \lambda_i$, and repeat.
			}
	\label{fig:circuit_1}
\end{figure*}

In summary, we propose a novel Hamiltonian engineering method, that is only limited by the locality of the native system interactions, the number of required pulses scales linear with the number of interactions, the total evolution time is minimized, and which is robust against common experimental imperfections.
Therefore, we believe that our method finds many applications in medium scale quantum simulation and gate based quantum computation.

In \cref{sec:main_results} we provide an overview of our general, efficient and robust methods.
In \cref{sec:num_sim} we show the generality, efficiency and robustness of our methods in numerical simulations of a realistic setting.
In \cref{sec:reshaping} we introduce our general framework for engineering Hamiltonians using a \ac{LP}.
In \cref{sec:pauli_conjugation,sec:clifford_conjugation} the efficient Hamiltonian engineering method by Pauli and Clifford conjugation are presented.
Finally, in \cref{sec:robustness} we provide methods to investigate and mitigate the effect of errors on both, the Pauli and Clifford conjugation methods.

\section{Overview of main results}
\label{sec:main_results}
Before describing the technical details of our Hamiltonian engineering method, we first provide an overview of our main contributions.
We begin by introducing the general framework, then highlight its generality, efficiency, and robustness using the example of an Ising-type system Hamiltonian.

In the following, we denote the $2 \times 2$ identity and Pauli matrices as $I, X, Y, Z$ and use $\1_n\equiv \1$ for the $n$-qubit identity matrix.
Tensor products of Pauli matrices are called \emph{multi-qubit Pauli operators} or \emph{Pauli strings}, the set of which is denoted as $\mathsf{P}^n = \{I,X,Y,Z\}^{\otimes n}$.
We note that such a Pauli operator, e.g.~$P = X \otimes Y \otimes X \otimes I$, can be both interpreted as a layer of simultaneous single-qubit $\pi$ pulses and as an interaction term in a local Hamiltonian.
This is because the Pauli operators forms a basis for the space of Hermitian operators, and we can thus write any Hamiltonian as
\begin{equation}
\label{eq:pauli_exp}
 H = \sum_{P\in \mathsf{P}^n} J_{P} P \, ,
\end{equation}
where $P$ represents the interaction and $J_{P} \in \RR$ is the corresponding interaction strength.
We collect all $J_{P}$ in a vector $\vec J \in \RR^{4^n}$.
We denote the \emph{locality} of an interaction $P$, i.e.\ the qubit indices on which $P$ acts non-trivially, as $\supp (P)$.
For any constant $k$, we call the Hamiltonian~\eqref{eq:pauli_exp} \emph{$k$-local} if $J_{P} = 0$ whenever $\abs{\supp(P)} > k$.

Our goal is to simulate the time evolution under a desired target Hamiltonian $H_T$ by using the native evolution of a system Hamiltonian $H_S$ with layers of (parallel) single-qubit pulses $\vec S$.
We express the system and target Hamiltonian as
\begin{equation}
 H_S = \sum_{P\in \mathsf{P}^n} J_{P} P
 \, , \qquad
 H_T = \sum_{P\in \mathsf{P}^n} A_{P} P \, .
\end{equation}
We consider two types of pulse layers: $\pi$ pulses (Pauli gates) and $\pi/2$ pulses (single-qubit Clifford gates).
By conjugating the system Hamiltonian with the pulses $\vec S$, we change its dynamics as
\begin{equation}
 \vec S^\dagger \e^{-\i t H_S} \vec S = \e^{-\i t \vec S^\dagger H_S \vec S} ,
\end{equation}
see \cref{fig:circuit_1}. 
This transformation is also known as the toggling-frame transformation.
Our objective is to find an exact decomposition of the target Hamiltonian such that
\begin{equation}
\label{eq:decomp_ex}
 H_T = \sum_{i=1}^{\dd} \lambda_i \vec S_i^\dagger H_S \vec S_i \,,
\end{equation}
where $\lambda_i > 0$ denotes the \emph{relative evolution time} under the system Hamiltonian, $\vec S_i$ denotes the corresponding single-qubit pulse layer and $\dd$ is the number of interaction terms that can be generated by the transformations.
Below \eqref{eq:LP}, we explain why $\dd$ coincides with the number of terms in the decomposition.
The transformation $\vec S_i^\dagger H_S \vec S_i$ linearly transforms the interaction strengths $\Set{J_{P}}$, and we capture this effect with a matrix $\WG \in \RR^{\dd \times \rr}$:
each entry $\WG_{P, i}$ specifies the contribution of the interaction term $P$ under the transformation with respect to the pulse layer $i$.
Specifically, we have
\begin{equation}
\label{eq:first_general_ideal_matrix}
\WG_{P, i} \coloneqq \frac{1}{2^n} \Tr \left(P \left(\vec S_i^\dagger H_S \vec S_i \right) \right) \,.
\end{equation}
Thus, $\rr$ is the number of considered single-qubit pulse layers.

Target Hamiltonians $H_T$ that allow for a decomposition as in \cref{eq:decomp_ex} need to have the same locality as the system Hamiltonian, i.e.~we cannot generate interaction terms of higher locality from lower locality ones (as we assume single-qubit pulses only).
Moreover, the set of single-qubit pulses also restricts the possible target Hamiltonians, see the discussion in the following subsection.

If it exists, an exact decomposition as in \cref{eq:decomp_ex} can be obtained by solving the following \ac{LP}, which minimizes the total relative evolution time
\begin{equation}
\tag{LP}
\label{eq:LP}
\begin{aligned}
 \mathrm{minimize} &&& \vec 1^T \vec \lambda  \\
 \mathrm{subject\ to}    &&& \WG \vec \lambda = \vec A \,, \; \vec \lambda \in \R^{\rr}_{\geq 0} \,,
\end{aligned}
\end{equation}
where $\vec A \in \RR^\dd$ represents the target interaction strengths and $\vec 1=(1,1,\dots,1)$ is the all-ones vector such that $\vec{1}^T \vec \lambda = \sum_{i}\lambda_{i}$.
The number of considered pulse layers $\rr$ equals the number of optimization variables $\lambda_i$, since each $\lambda_i$ is associated with a specific $\vec S_i$ in \cref{eq:decomp_ex}.
\eqref{eq:LP} is central for our approach to Hamiltonian engineering and its analysis plays an important role in this work.
For this purpose, we require some basic notions and properties from the theory of linear
programming \cite{boyd_convex_2004} which we will briefly introduce in the following.
A \emph{feasible solution} $\vec \lambda$ is a vector which satisfies all constraints of \eqref{eq:LP}.
An \emph{optimal solution} is a feasible solution that also minimizes the objective function and is indicated by an asterisk as $\vec \lambda^*$.
If \eqref{eq:LP} has a feasible solution, then there also exists a $\dd$-sparse optimal solution $\vec \lambda^*$ \cite{BARANY1982}, corresponding to a decomposition as in \cref{eq:decomp_ex} with only $\dd$ terms.
Such a $\dd$-sparse optimal solution can be found using the simplex algorithm which, in practice, has a runtime that scales polynomial in the problem size $\dd \times \rr$ \cite{spielman_smoothed_2004,Bach2025}.
In \cref{sec:reshaping}, we derive conditions for the existence of a solution to this \ac{LP}.
In \cref{sec:eff_pauli_heu,sec:eff_cliff_heu}, we further demonstrate that choosing pulse layers $\vec S_i$ at random is an effective strategy for obtaining feasible \acp{LP} with relatively few variables.

The resulting dynamics are then implemented as a sequence of evolutions under $H_S$ for time $t \lambda_i$ interleaved with single-qubit pulse layers $\vec S_i$
\begin{equation}
\label{eq:approximate_evolution}
 \e^{-\i t H_T} \approx \vec S_1^\dagger \e^{-\i t \lambda_1 H_S} \vec S_1 \dots \vec S_m^\dagger \e^{-\i t \lambda_m H_S} \vec S_m \, ,
\end{equation}
where $m$ is the total number of applied evolution blocks.
If the conjugated terms in \eqref{eq:decomp_ex} mutually commute, then the target evolution is implemented exactly, and we have $m = \dd$.
Otherwise, the evolution can be approximated using, potentially higher-order, product formulas such as the Trotter formula, then we have $m > \dd$.
We want to highlight, that the only approximation of the target dynamics arises from the product formula, since our method provides an exact decomposition.

In the following discussion of this section we illustrate our method on a system with native Ising interactions of arbitrary coupling strengths
\begin{equation}
\label{eq:ising_sys_ex}
 H_S = \sum_{i \neq j} J_{ij} Z_i Z_j \, ,
\end{equation}
with $J_{ij} \in \R$.
Note, that such a system Hamiltonian contains $n(n-1)/2$ distinct interaction terms.

\subsection{Generality}
A natural question to ask is which target Hamiltonians $H_T$ can be generated by \cref{eq:decomp_ex}. In the following, we consider two different approaches:
Pauli conjugation via $\pi$ pulses (also known as refocusing pulses), and Clifford conjugation via $\pi/2$ pulses.
We also assume that we have local control over individual qubits, thus we can apply arbitrary $\pi$ or $\pi/2$ pulses to any subset of qubits.
This level of control is necessary if we want to engineer arbitrary target coupling strengths.
A lower level of control will restrict the possible target couplings but might still be sufficient in certain situations.

\paragraph*{Pauli Conjugation.}
This method allows the modification of individual coupling strengths in the system Hamiltonian.
For an Ising system Hamiltonian (general Hamiltonians are discussed in \cref{sec:pauli_conjugation}),
we can engineer target Hamiltonians of the form
\begin{equation}
 H_T = \sum_{i \neq j} A_{ij} Z_i Z_j \, ,
\end{equation}
with arbitrary interaction strength $A_{ij} \in \R$.
This enables the individual control over all $\dd = n/2 (n-1)$ Ising interactions.
Moreover, it is possible to modify interactions of unknown strength or even cancel them, see \cref{sec:unknown_hermitian_op}.
We illustrate this capability in \cref{sec:2D_lattice_model}, where pulse sequences are constructed to engineer an Ising Hamiltonian in the presence of additional, unknown three-body interaction terms.
It is also possible to suppress certain long-range interactions in an Ising Hamiltonian to implement local multi-qubit gates.

\paragraph*{Clifford Conjugation.}
By using $\pi/2$ pulses, we can not only modify interaction strengths but also change interaction types, which we discuss in \cref{sec:clifford_conjugation} in detail. 
In particular, Ising system Hamiltonian can be engineered into arbitrary 2-local Hamiltonians of the form
\begin{equation}
 H_T = \sum_{\substack{P\in \mathsf{P}^n \\ |\supp (P)| = 2}} A_{P} P \, ,
\end{equation}
with arbitrary interaction strength $A_{P} \in \R$.
This yields individual control over all $\dd = 3^2 n (n-1)/2$ distinct 2-local interactions.
The prefactor $3^2$ stems from the enlarged family of reachable target Hamiltonians.
In general, for any $k$-body interaction in the system Hamiltonian, we can engineer $3^k$ distinct interaction terms with individually tunable strengths.

\subsection{Efficiency}
The \ac{LP} described above involves $\dd$ equality constraints, one per target interaction term, but the number of variables $\rr$ required for an optimal solution scales exponentially with the number of qubits, i.e.\ $\rr = 4^n$ for $\pi$ pulses or $\rr = 12^n$ for $\pi/2$ pulses.
To address this, we propose an efficient relaxation that reduces the number of variables needed.
In particular, it suffices to choose
\begin{equation}
 \rr > 2 \dd
\end{equation}
variables to obtain a solution of the \ac{LP}.
Due to the relaxation the solution does not minimize the total evolution time.
However, it provides a low evolution time solution that still yields an exact decomposition as in \cref{eq:decomp_ex}.
This reduced \ac{LP} allows a flexible trade-off between classical runtime and resulting quantum evolution time, increasing $\rr$ reduces evolution time at the cost of increased classical runtime.
For more details on the efficient \ac{LP}, see \cref{sec:eff_pauli_heu,sec:eff_cliff_heu}, and for a numerical investigation of the efficiency we refer to \cref{sec:efficiency_LP_approach}.

For Ising interactions, this leads to an exponential reduction in \ac{LP} size, e.g.\ $\rr > n(n-1)$ for $\pi$ pulses and $\rr > 3^2 n(n-1)$ for $\pi/2$ pulses.

\subsection{Robustness}
In practice, single-qubit pulses have some non-zero duration $t_p > 0$. 
Depending on the platform, the system Hamiltonian $H_S$ remains active during that time.
As a result, the implemented pulse and its inverse becomes
\begin{equation}
 \vec S_{\err} = \e^{- \i t_p (H_S + H_c)} \quad \text{and} \quad \vec S_{\err}^\prime = \e^{- \i t_p (H_S - H_c)}\, ,
\end{equation}
where $H_c$ denotes the control pulse Hamiltonian.
Using \acf{AHT}, we approximate the overall finite pulse time effect as
\begin{equation}
 \vec S_{\err}^\prime \e^{-\i t H_S} \vec S_{\err} \approx \e^{-\i t \vec S^\dagger H_S \vec S + H_{\err}} \, .
\end{equation}
A key insight is, that the error term $H_{\err}$ has the same locality as $H_S$ in first order Magnus expansion and can, hence, be incorporated into the efficient \ac{LP} framework.
Therefore, the error $H_{\err}$ can be cancelled exactly, see \cref{sec:robustness} for more details.

Moreover, our framework supports the use of robust composite pulses in place of $\pi/2$ pulses.
This enables mitigation of various control errors, including rotation angle error (Rabi frequency errors) and off-resonance errors \cite{Masamitsu2013,kukita2021}.
We provide numerical simulations validating the robustness of our approach under realistic error models in \cref{sec:num_sim}.
\begin{figure*}[ht]
	\centering
	\includegraphics{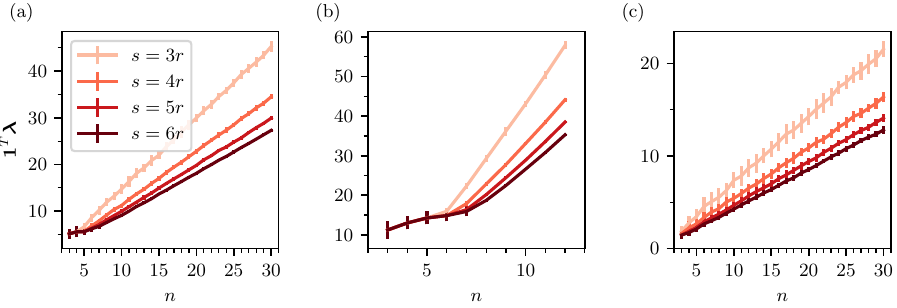}
	\includegraphics{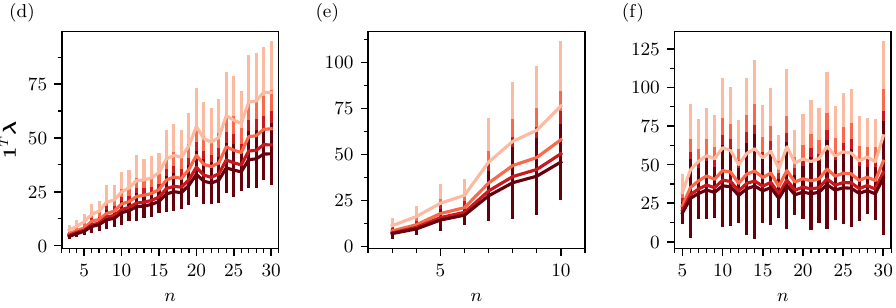}
	\caption{
			Scaling of total evolution time $\sum_i \lambda_i = \vec 1^T \vec \lambda$ with system size $n$ for randomly sampled target Hamiltonians.
			Each data point represents the average and error bars represent the sample standard deviation over $50$ uniform random samples of the target interaction strengths $\vec A$ from $[-1 , 1]^\dd$.
			The numbers of variables in the \ac{LP} is varied from $\rr = 3\dd$ to $6\dd$ to demonstrate the reduction of the total evolution time when increasing the search space.
			\textbf{Top (a-c)}: Results for the efficient Pauli conjugation \ac{LP} with $\pi$ pulses, using 2-local (a), 3-local (b), and random many-body (c) system Hamiltonians.
			\textbf{Bottom (d-f)}: Results for the Clifford conjugation \ac{LP} with $\pi/2$ pulses, using similar 2-local (d) and 3-local (e) system Hamiltonians, and a randomly sampled 5-local Hamiltonian in (f).
			Moreover, for (d-e) randomly selected interaction terms in the system Hamiltonian are set to zero.
	}
	\label{fig:eff_heu_pauli_LP}
\end{figure*}

\section{Numerical results}
\label{sec:num_res}
In this section, we assess the performance of our method by numerical simulations.
We divide the analysis into two parts.
First, we demonstrate the efficiency of our \ac{LP} approach in terms of the total evolution time and classical runtime.
This includes the scaling of the total evolution time for engineering general Hamiltonians (2-local, 3-local and random many-body Hamiltonians), and the scaling of the classical runtime for 2-local Hamiltonians on a 2D square lattice.
Second, we simulate the time evolution under a target Hamiltonian, evaluating the robustness of our approach.
These simulations include Ising interactions on a 2D lattice, unwanted 3-body interactions, and control imperfections, as well as general Heisenberg Hamiltonians with imperfect pulse control.

Our robust \ac{LP} methods can be replaced by a \acf{MILP}, which we define in \cref{sec:MILP}, to reduce the number of required pulses at the cost of additional integer valued constraints.
All \acp{LP} and \acp{MILP}, in our work are modelled with the Python package CVXPY \cite{agrawal2018rewriting,diamond2016cvxpy} and, if not otherwise stated, solved with the MOSEK solver \cite{mosek} on a workstation with $130\,$GB RAM and the AMD Ryzen Threadripper PRO 3975WX 32-Cores processor.
Furthermore, we made the code to reproduce all figures publicly available on GitHub \cite{GitHub}.

\begin{figure*}[ht]
	\centering
	\includegraphics{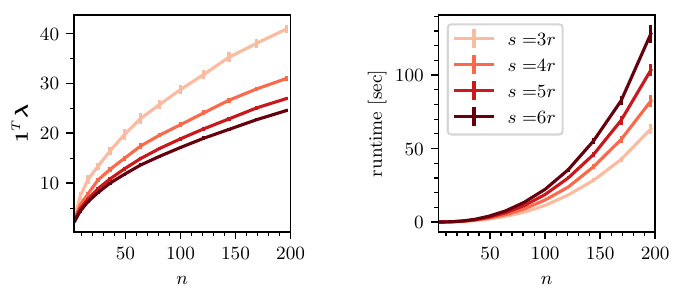}
	\caption{
			A 2-local system Hamiltonian on a 2D square lattice with $n = 4, \dots , 196$ qubits as in \cref{eq:2D_lattice} and a constant interaction strenght $J$.
			Each data point represents the average and error bars represent the sample standard deviation over $50$ uniform random samples of the target interaction strengths $\vec A$ from $[-1 , 1]^\dd$.
			\textbf{Left}: The evolution time over the number of qubits $n$ is shown.
			\textbf{Right}: The runtime for solving the Pauli conjugation \ac{LP} over the number of qubits $n$ is shown.
	}
	\label{fig:2D_lattice_ham}
\end{figure*}

\subsection{Generality and efficiency of the linear program approach}
\label{sec:efficiency_LP_approach}
\subsubsection{Scaling of the total evolution time}
We solve the efficient Pauli and Clifford conjugation \acp{LP}, which utilizes $\pi$ and $\pi/2$ pulses and are defined in \cref{sec:pauli_conjugation,sec:clifford_conjugation}, respectively.
Both of which aim to reduce the total evolution time $\vec 1^T \vec \lambda = \sum_i \lambda_i$.
Increasing the number of considered single-qubit pulse layers $\rr$ yields a larger search space for the \ac{LP}, which typically finds a lower optimum, thus decreases the total evolution time.
We investigate how the total evolution time scales with the number of variables $\rr$ and the system size $n$, considering various families of Hamiltonians with all-to-all connectivity,
which are well-suited to explore the generality of our approach.
Note, that a $k$-local Hamiltonian with all-to-all connectivity contains $\sum_{i=2}^k 3^i \binom{n}{i}$ interaction terms.
For instance, an Ising Hamiltonian includes $\binom{n}{2} = n(n - 1)/2$ terms, while a general 2-local Hamiltonian has $3^2 \binom{n}{2} = (9/2) n(n - 1)$ terms.

In the Pauli conjugation \ac{LP}, the number of constraints $\dd$ equals the number of non-zero terms in the system Hamiltonian.
We solve instances of this \ac{LP} with system Hamiltonians containing: all possible 2-local interactions ($\dd = 3^2 \binom{n}{2}$), all possible 3-local interactions ($\dd = 3^2 \binom{n}{2} + 3^3 \binom{n}{3}$) and randomly selected many-body interactions with $\dd = n^2$.
Results are shown in \cref{fig:eff_heu_pauli_LP}(a–c).
In all cases, the total evolution time scales linearly with system size, even with the polynomial scaling in the number of interaction terms.

For the Clifford conjugation \ac{LP}, the number of constraints $\dd$, i.e.\ the number of interactions that can be generated, depends on both the number and locality of the interactions in the system Hamiltonian.
We again consider 2-local and 3-local system Hamiltonians, as well as a random 5-local Hamiltonian.
The random 5-local Hamiltonian is constructed by choosing $10$ random localities $\operatorname{supp} (P)$ for each $i = 2, \dots , 5$ and including all $3^i$ interaction terms per locality, resulting in $\dd = 10 \sum_{i=2}^5 3^i$ total terms which is independent of the system size.
To add further variability, we randomly set some interaction strengths $J_{P}$ to zero, ensuring that at least one term per locality remains non-zero.
Non-zero interaction strengths are drawn independently from the uniform distribution, $J_{P} \overset{\text{i.i.d.}}{\sim} \operatorname{unif}([-1,+1])$.
So, the number, index, and values of the non-zero interactions in the system Hamiltonians are random.
\Cref{fig:eff_heu_pauli_LP}(d-f) shows the scaling results for these systems.
The evolution time again scales roughly linearly for 2-local systems, slightly super-linearly for 3-local systems, and appears nearly constant for the random 5-local case with constantly many interactions.
Due to the randomization, these simulations exhibit greater sample standard deviation compared to the Pauli conjugation approach.

\subsubsection{Evolution time vs.\ classical runtime}
\label{sec:2D_lattice}
To further assess practical efficiency, we evaluate the performance of our Pauli conjugation method on 2-local Hamiltonians defined on a 2D square lattice.
Such system Hamiltonians are in a broad sense motivated by experimental quantum platforms such as interacting superconducting qubits or cold atoms in optical lattices \cite{Arute2019QuantumSupremacy_short_auth,Schmid_2023,Henriet2020quantumcomputing}.

We consider the following system and target Hamiltonians
\begin{equation}
\label{eq:2D_lattice}
 H_S = J \sum_{\substack{P\in \mathsf{P}^n \\ \supp (P) \in E}} P
 \, , \quad
 H_T = \sum_{\substack{P\in \mathsf{P}^n \\ \supp (P) \in E}} A_{P} P \, ,
\end{equation}
where each $P$ represents a 2-body interaction term on a 2D square lattice with the set of edges $E$.
For a $d \times d$ lattice the number of interactions is $\dd = 3^2 \abs{E} = 3^2 2d(d-1)$.
The system interaction strength $J$ is uniform across all interactions.
This system Hamiltonian includes all possible interactions on the edges, i.e.\ interactions of the form $Z_i Y_j$, $X_i Z_j$ or $Y_i Y_j$ with $(i,j) \in E$.

To emphasize scalability even on commodity hardware, we solve the efficient Pauli conjugation \ac{LP} with the MOSEK solver \cite{mosek} on a laptop with Intel Core i7 Processor ($8\,$x$\,1.8\,$GHz) and $16\,$GB RAM.
As shown in \cref{fig:2D_lattice_ham}, our method can engineer a target Hamiltonian $H_T$ with arbitrary couplings $\vec A$ on $196$ qubits in about $60$ seconds of classical runtime using $\pi$ pulses.
Notably, the total evolution time scales sublinear with system size, allowing a fast implementation of the target Hamiltonian on platforms with restricted connectivity.

\subsection{Numerical Hamiltonian simulations}
\label{sec:num_sim}
To benchmark our methods, we perform numerical simulations that model the most relevant error sources occurring in practice.
We consider a device with 2D lattice interactions, motivated by superconducting qubit platforms, and one with all-to-all connectivity modelling an ion trap.
However, the presented methods are also applicable to other platforms such as cold atoms in the weak coupling regime.

For simplicity, we describe the transformation of the system Hamiltonian $H_S$ with respect to a single pulse layer.
However, our simulations employ transformations of $H_S$ containing multiple pulse layers, as detailed in the general discussion in \cref{sec:robust_general}.
Let the ideal control Hamiltonian for a single pulse layer on $n$ qubits be given by
\begin{equation}
 H_c = \frac{1}{t_p} \sum_{i=1}^n h_i \, ,
\end{equation}
where $h_i = (\pi/2) P_i$ for $\pi$ pulses and $h_i = (\pi/4) P_i$ for $\pi/2$ pulses, with $P_i = I^{\otimes (i-1)} \otimes P \otimes I^{\otimes (n-i)}$ and $P \in \{I,X,Y,Z\}$.
To model realistic errors, we introduce two sources of pulse imperfections by
\begin{equation}
\label{eq:rot_err_model}
 H_{c, (\varepsilon, f)} = \frac{1}{t_p} \sum_{i=1}^n \left( (1+\varepsilon_{i}) h_{i} + f_i Z_i \right) \, ,
\end{equation}
where $\varepsilon_{i}\overset{\text{i.i.d.}}{\sim} \operatorname{unif}([0, \varepsilon])$ is the \emph{relative angle error} and $f_{i}\overset{\text{i.i.d.}}{\sim} \operatorname{unif}([0, f])$ is the \emph{off-resonance error} on the $i$-th qubit.
Both errors are sampled once and represent faulty and inhomogeneous control of the single-qubit pulses.

\begin{figure*}[ht]
	\begin{minipage}{0.3\linewidth}
	 \includegraphics{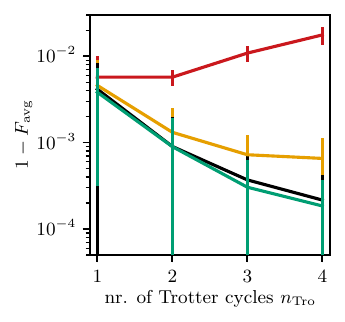}
	\end{minipage}
	\begin{minipage}{0.6\linewidth}
	\begin{tabular}{l|c|c|c}
	  \multicolumn{1}{r|}{Error: } & finite pulse time & rotation angle & three-body \\
	  \hline
	  \tikzLabel{green} \textsf{robust Pauli}    & \textsf{X} & \textsf{X} & \textsf{X} \\
	  \hline
	  \tikzLabel{myRed} \textsf{naive}           & \textsf{X} & \textsf{X} & \textsf{X} \\
	  \hline
	  \tikzLabel{myOrange} \textsf{angle error}     &   & \textsf{X} & \textsf{X} \\
	  \hline
	  \tikzLabel{black} \textsf{exact (Trotter)} & & & \textsf{X} \\
	\end{tabular}
	\includegraphics{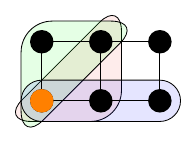}
	\end{minipage}
	\caption{\textbf{Left}: The sample mean and standard deviation of the average gate infidelities \cref{eq:avg_infidelity} for implementing the time evolution of $\e^{-\i t H_T}$ with $H_T$ from \cref{eq:ising_target} and $t=1\, \second$ over the number of Trotter cycles $\nTro$ is shown.
			The first-order Trotter formula from \cref{eq:first_order_trotter_sim} is used to approximate the time evolution.
			The sample mean and standard deviation are calculated over $50$ random samples of the target Hamiltonians $H_T$.
			Note, that each of the $\nTro$ Trotter cycles contains $\pp \dd$ evolution blocks $U(t\lambda_{\vec c})$, where $\pp = 8$ and $\rr = 17$, with different rotation directions for the $\pi$ pulses, see \cref{sec:robust_pauli} for more details.
			\textbf{Top Right}: Table indicating which error types are present for the different simulations.
			\textbf{Bottom right}: Example $2 \times 3$ 2D square lattice on $n=6$ qubits with the two-body interactions, solid black lines, and the three-body interactions for the lower left qubit, colored areas.
			}
	\label{fig:2D_lattice}
\end{figure*}

We also take into account the finite pulse time effects by modeling such single-qubit pulse layers
\begin{equation}
\vec S_{\err} = \e^{- \i t_p (H_S + H_{c, (\varepsilon, f)})} \quad \text{and} \quad \vec S_{\err}^\prime = \e^{- \i t_p (H_S - H_{c, (\varepsilon, f)})} \, ,
\end{equation}
which capture the joint evolution under the system Hamiltonian $H_S$ and the imperfect control pulse Hamiltonian $H_{c, (\varepsilon, f)}$.
In the absence of $H_S$, these two operators are exact inverses.
We define the resulting erroneous \emph{evolution block} as
\begin{equation}
\label{eq:evo_block_finite_pulse_time_sim}
 U(t \lambda) = \vec S_{\err}^\prime \e^{-\i t \lambda H_S} \vec S_{\err} \, ,
\end{equation}
which can be generalized to the expression \eqref{eq:evo_block_finite_pulse_time} below.

We numerically simulate the time evolution given by the (robust) Pauli and Clifford conjugation methods by explicitly computing products of matrix exponentials.
More precisely, the decomposition of the target Hamiltonian is implemented using the first- or second-order Trotter formula,
\begin{equation}
\label{eq:first_order_trotter_sim}
 U_\mathrm{sim} = \left(\prodr_{\vec c} U\left(\frac{t \lambda_{\vec c}}{\nTro}\right) \right)^{\nTro}
\end{equation}
or
\begin{equation}
\label{eq:sec_order_trotter}
U_\mathrm{sim} = \left(\prodl_{\vec c} U\left(\frac{t\lambda_{\vec c}}{2 \nTro}\right) \prodr_{\vec c} U\left(\frac{t\lambda_{\vec c}}{2 \nTro}\right) \right)^{\nTro} \, ,
\end{equation}
where a single evolution block $U(t\lambda_{\vec c})$ consists of the time evolution under the system Hamiltonian conjugated by single-qubit pulses as in \cref{eq:evo_block_finite_pulse_time}, and thus explicitly models finite pulse time errors.
Throughout our work we use the convention $\prodr_{i=1}^L A_i \coloneqq A_1\dots A_L$ and $\prodl_{i=1}^L A_i \coloneqq A_L\dots A_1$ to indicate the order of the products of non-commuting operators.

To capture the quality of the implementation, $U_\mathrm{sim}$ is compared to the target evolution $U_T = \e^{-\i t H_T}$.
As a measure of quality, we use the \emph{average gate infidelity}
\begin{equation}
\label{eq:avg_infidelity}
 1 - F_{\mathrm{avg}} (U_\mathrm{sim} , U_T) = 1 - \frac{\Tr (U_T^\dagger U_\mathrm{sim}) +1}{d+1} \, ,
\end{equation}
where $d$ is the Hilbert space dimension.

\subsubsection{Simulation of a 2D lattice model}
\label{sec:2D_lattice_model}
We consider a quantum platform with a native $2 \times 3$ lattice Hamiltonian with $n=6$ qubits as an abstract model for interacting superconducting qubits,
\begin{equation}
 H_S = J \sum_{ij} Z_i Z_j + \sum_{ijk} E_{ijk} X_i X_j X_k \, .
\end{equation}
This Hamiltonian has $\dd=17$ non-zero interaction terms.
We assume that there are unwanted three-body interactions of unknown strength, where we consider all possible nearest neighbor three-body interactions.
An example of the considered two- and three-body interactions are depicted in \cref{fig:2D_lattice} as solid black lines and colored areas respectively.
The two-body coupling coefficients are constant, $J = 10^3 \, \Hz$.
The three-body coupling coefficients are uniformly sampled $E_{ijk} \overset{\text{i.i.d.}}{\sim} \operatorname{unif}([-10^{2} , 10^{2}]\cdot \Hz)$ \emph{after} the design of the pulse sequences and therefore considered as unknown.
The finite $\pi$ pulse time is $t_p= 10^{-7} \, \second$.
We want to implement the Ising Hamiltonian
\begin{equation}
\label{eq:ising_target}
 H_T = \sum_{ij} A_{ij} Z_i Z_j \, ,
\end{equation}
with random coupling coefficients $A_{ij} \overset{\text{i.i.d.}}{\sim} \operatorname{unif}([10^{-1} , 1]\cdot \Hz)$.
The pulse errors are modeled as in \cref{eq:rot_err_model} with $\varepsilon = 10^{-1}$ and no off-resonance error.
\begin{figure*}[ht]
	\centering
    \includegraphics{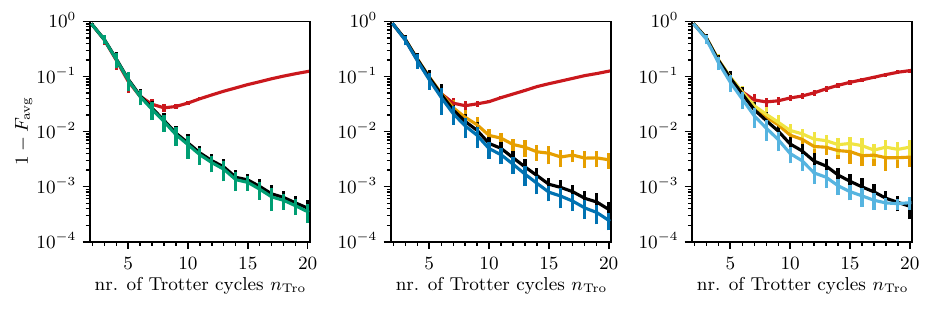}
\footnotesize
\textsf{
	\begin{tabular}{l|c|c|c}
	  \multicolumn{1}{r|}{\textnormal{Error:} }& \textnormal{finite pulse time} & \textnormal{rotation angle} & \textnormal{off-resonance} \\
	  \hline
	  \tikzLabel{green} robust Clifford          & X &   &  \\
	  \hline
	  \tikzLabel{darkBlue} CP$_\mathrm{SCROFULOUS}$ & X & X &  \\
	  \hline
	  \tikzLabel{lightBlue} CP$_\mathrm{SCROBUTUS}$  & X & X & X \\
	  \hline
	  \phantom{\tikzLabel{myRed}} naive (left)             & X &   &  \\
	  \tikzLabel{myRed} naive (middle)           & X & X &  \\
	  \phantom{\tikzLabel{myRed}} naive (right)            & X & X & X \\
	  \hline
	  \tikzLabel{myOrange} angle error              &   & X &  \\
	  \hline
	  \tikzLabel{myYellow} off-resonance error      &   &   & X \\
	  \hline
	  \tikzLabel{black} exact (Trotter)          & & &\\
	\end{tabular}}
	\caption{
			\textbf{Top}: The sample mean and standard deviation of the average gate infidelities \cref{eq:avg_infidelity} for implementing the time evolution of $\e^{-\i t H_T}$ with $H_T$ from \cref{eq:heisenberg} and $t=1\, \second$ over the number of Trotter cycles $\nTro$ is shown.
			The second-order Trotter formula from \cref{eq:sec_order_trotter} is used to approximate the time evolution.
			The sample mean and standard deviation are calculated over $50$ random samples of the Heisenberg Hamiltonians $H_T$ on $n=8$ qubits.
			\textbf{Top Left}: The Clifford conjugation robust against finite pulse time errors (dark green) is compared to the non-robust Clifford conjugation (red).
			\textbf{Top Middle}: The Clifford conjugation in combination with the SCROFULOUS pulse sequence \cite{Cummins2002} robust against finite pulse time errors and rotation angle errors (dark blue) is compared to the non-robust Clifford conjugation (red).
			\textbf{Top Right}: The Clifford conjugation in combination with the SCROBUTUS pulse sequence \cite{kukita2021} robust against finite pulse time errors, rotation angle errors and off-resonance errors (light blue) is compared to the non-robust Clifford conjugation (red).
			\textbf{Bottom}: Table indicating which error types are present for the different simulations.
	}
	\label{fig:noisy_hamiltonian}
\end{figure*}

In \cref{fig:2D_lattice}, we compare the \textsf{naive} Pauli conjugation from \cref{sec:pauli_conjugation} against the \textsf{robust Pauli} conjugation from \cref{sec:robust_pauli}.
For the sake of a clear presentation we moved details to \cref{app:details_2D_lattice_model}.
In the \textsf{angle error} and \textsf{exact (Trotter)} data we apply the same sequences as for the \textsf{naive} but with different errors, see table in \cref{fig:2D_lattice}.
To approximate the time evolution under $H_T$ we use the first-order Trotter formula from \cref{eq:first_order_trotter_sim} for all sequences.
Increasing the number of Trotter cycles $\nTro$ improves the accuracy of the Trotter approximation in absence of other errors.
However, an increased number of Trotter cycles $\nTro$ also yields an increased number of evolution blocks $U(t\lambda_{\vec c})$.
From \cref{fig:2D_lattice} it is clear that the finite pulse time errors and angle errors in the \textsf{naive} sequences quickly accumulate even for small number of Trotter cycles.
However, the \textsf{robust Pauli} sequences converge as quickly as the Trotter sequences without any errors.

\subsubsection{Simulation of Heisenberg Hamiltonians with an ion trap model}
\label{sec:ion_trap_model}
We consider an ion trap with Ytterbium ions in an external magnetic field gradient \cite{Piltz2016VersatileMicrowaveDriven}.
The effective system Hamiltonian is
\begin{equation}
 H_S = \sum_{i \neq j}^n J_{ij} Z_i Z_j \, ,
\end{equation}
where the coupling coefficients $J_{ij}$ are calculated for a harmonic trapping potential affecting the equilibrium positions of the ions in the magnetic field gradient (cf.~Ref.~\cite[{App.~A}]{basler_synthesis_2023}).
We consider all-to-all connectivity, thus the number of reachable interactions is $\dd = \frac{n}{2} (n-1)$ for the Pauli conjugation and $\dd = 3^2 \frac{n}{2} (n-1)$ for the Clifford conjugation.
The coupling coefficients $J_{ij}$ are proportional to $(B_1 / \omega)^2$, with a magnetic field gradient of $B_1 = 40\, T/m$ and a trap frequency of $\omega = 2\pi 500\, k\Hz$.
The finite $\pi$ pulse time is $t_p= 2 \, \mu \second$, which is proportional to $\pi / \Omega$ with the Rabi frequency $\Omega = \omega / 2$.
Moreover, we consider the rotation angle errors of strength $\varepsilon = 10^{-1}$ and off-resonance errors of strength $f = 10^{-1}$, which are modeled as in \cref{eq:rot_err_model}.
As a target Hamiltonian, we consider the general Heisenberg Hamiltonian
\begin{equation}
\label{eq:heisenberg}
 H_T = \sum_{i \neq j}^n \left(A^x_{ij} X_i X_j + A^y_{ij} Y_i Y_j + A^z_{ij} Z_i Z_j \right) \, ,
\end{equation}
with random coupling coefficients $A^x_{ij}, A^y_{ij}, A^z_{ij} \overset{\text{i.i.d.}}{\sim} \operatorname{unif}([10^{-1} , 1]\cdot \Hz)$ and all-to-all connectivity.

In \cref{fig:noisy_hamiltonian} we compare the \textsf{naive} Clifford conjugation from \cref{sec:clifford_conjugation} to the \textsf{robust Clifford}, \textsf{CP$_\mathrm{SCROFULOUS}$} and \textsf{CP$_\mathrm{SCROBUTUS}$} Clifford conjugations from \cref{sec:robust_cliff}.
Again, we moved details to \cref{app:details_ion_trap_model}.
The \textsf{CP$_\mathrm{SCROFULOUS}$} and \textsf{CP$_\mathrm{SCROBUTUS}$} conjugation implement the robust composite pulses SCROFULOUS \cite{Cummins2002} and SCROBUTUS \cite{kukita2021}, respectively.
These robust composite pulses are designed to compensate for angle errors (SCROFULOUS) or both angle errors and off-resonance errors (SCROBUTUS).
For the \textsf{angle error}, \textsf{off-resonance error} and \textsf{exact (Trotter)} data we apply the same sequences as for the \textsf{naive} data but with different errors, see table in \cref{fig:noisy_hamiltonian}.
To approximate the time evolution under $H_T$ we use the second-order Trotter formula from \cref{eq:sec_order_trotter} for all sequences.
As for the Pauli conjugation method in \cref{sec:2D_lattice_model}, we can observe that the finite pulse time errors and angle errors in the \textsf{naive} sequences quickly accumulate even for moderate number of Trotter cycles.
However, the \textsf{robust Clifford}, \textsf{CP$_\mathrm{SCROFULOUS}$} and \textsf{CP$_\mathrm{SCROBUTUS}$} sequences converge as quickly as the Trotter sequences without any errors, and seem to be even more accurate at moderate numbers of Trotter cycles.

Note, that the rotation angle error robustness in the \textsf{CP$_\mathrm{SCROFULOUS}$} Clifford conjugation is stronger than in the \textsf{robust Pauli} conjugation in \cref{sec:2D_lattice_model}. Since each $\pi/2$ pulse (replaced by robust composite pulses) is robust against different angle errors whereas in the \textsf{robust Pauli} conjugation the angle errors are assumed to be constant over several implementations of $U(t\lambda_{\vec c})$.

\subsection{Summary of numerical results}
\label{sec:num_sum}
This section presents a comprehensive numerical demonstration of the performance of our Hamiltonian engineering framework.
We show that our methods achieve linear or sublinear scaling in evolution time with system size, particularly for systems with constrained connectivity.
We construct pulse sequences for engineering a 2D lattice Hamiltonian with $196$ qubits in about $60$ seconds on a laptop, highlighting the efficiency of our approach.
Moreover, we show how our methods can be tailored for robustness against a variety of realistic control errors, including finite pulse durations, angle errors, and off-resonance effects.
In summary, these numerical results underpin our claim of a general, efficient and robust method to engineer Hamiltonians.

\section{Hamiltonian engineering by linear programming}
\label{sec:reshaping}
In \cref{sec:main_results} we already introduced our \ac{LP} approach for Hamiltonian engineering.
Now, we provide the technical details in full generality.

We begin, by introducing a useful representation of Pauli operators as binary strings.
$\FF_2$ is the finite field over two elements, i.e.\ the set $\{0, 1\}$, where multiplication and addition are performed modulo $2$.
We index each $n$-qubit Pauli operator $P_{\vec a} \in \mathsf{P}^n$ by a binary vector $\vec a=(\vec a_x, \vec a_z) \in \FF_2^{2n}$ such that
\begin{equation}
 P_{\vec a} = P_{(\vec a_x, \vec a_z)} = \i^{\vec a_x \cdot \vec a_z} X(\vec a_x) Z(\vec a_z) \, .
\end{equation}
Here, $X(\vec x) \coloneqq X^{x_1}\otimes\dots\otimes X^{x_n}$ and $Z(\vec z) \coloneqq Z^{z_1}\otimes\dots\otimes Z^{z_n}$ where $X$, $Z$ are the single-qubit Pauli matrices.
Sometimes, we also call $P_{\vec a}$ a \emph{Pauli string}.
For example, $P_{\vec a} = X \otimes Y \otimes X \otimes I$ has the binary representation $\vec a = ((1,1,1,0),(0,1,0,0)) \equiv (1,1,1,0,0,1,0,0)$.
The Pauli strings satisfy the following commutator relation
\begin{equation}
\label{eq:pauli_conjugation}
 P_{\vec a} P_{\vec b} = (-1)^{\langle \vec a,\vec b \rangle} P_{\vec b} P_{\vec a} \, ,
\end{equation}
for any $\vec a, \vec b \in \FF_2^{2n}$
with the binary symplectic form on $\FF_2^{2n}$ defined by $\langle \vec a,\vec b \rangle \coloneqq \vec a_x \cdot \vec b_z + \vec a_z \cdot \vec b_x$, where the sum is taken in $\FF_2$, i.e., modulo $2$.
With this notation, the locality of a Pauli operator $P_{\vec a}$ is defined as $\supp (P_{\vec a}) \coloneqq \supp (\vec a) \coloneqq \{ i \in [n] | a_{x_i} \neq 0 \text{ or } a_{z_i} \neq 0 \}$.

We write $H_S$ and $H_T$ in their Pauli decompositions, utilizing the binary notation,
\begin{equation}
\label{eq:H_S_H_T}
 H_S = \sum_{\vec a \in \FF_2^{2n}\setminus \{ \vec 0 \}} J_{\vec a} P_{\vec a}
 \, , \qquad
 H_T = \sum_{\vec a \in \FF_2^{2n}\setminus \{ \vec 0 \}} A_{\vec a} P_{\vec a} \, ,
\end{equation}
with $J_{\vec a}, A_{\vec a} \in \RR$.
We exclude the term $P_{\vec 0} = P_{(0, \dots, 0)} = I^{\otimes n}$, leading to a global and thus unobservable phase.
Our goal is to simulate the time evolution under $H_T$ by the one under $H_S$ interleaved with layers of single-qubit pulses.
Recall, that we utilize the linear transformation $\e^{-\i t U^\dagger H U} = U^\dagger \e^{-\i t H} U$, which holds for any Hamiltonian $H$ and unitary $U$.
Hence, we seek a decomposition of the form in \cref{eq:decomp_ex}, with terms $\lambda_i \vec S_i^\dagger H_S \vec S_i$, where $\lambda_i > 0$ and $\vec S_i = S_i^{(1)} \otimes \dots \otimes S_i^{(n)}$ representing a layer of single-qubit pulses, which we later set to be either $\pi$ or $\pi/2$ pulses (i.e.\ Pauli or Clifford gates).
To this end, we consider the effect of $H_S \mapsto \vec S_i^\dagger H_S \vec S_i$, similar as in \cref{eq:first_general_ideal_matrix}, on the Pauli coefficients $J_{\vec a}$ which can be captured by a column vector $\WG_{i} \in \RR^\dd$
with elements
\begin{equation}
\label{eq:general_ideal_matrix}
\WG_{\vec a, i} \coloneqq \frac{1}{2^n} \Tr \left(P_{\vec a} \left(\vec S_i^\dagger H_S \vec S_i \right) \right) \,.
\end{equation}
Here, $\dd$ is the number of interaction terms $\{P_{\vec a}\}$ we can generate from the ones in $H_S$ by any pulse layer $\vec S_i$.
Let $\rr$ be the number of possible pulse layers on $n$ qubits, and define the matrix $\WG \in \RR^{\dd \times \rr}$ by taking all vectors $\WG_i$ as columns.
Then, a decomposition \eqref{eq:decomp_ex} can be found by solving \eqref{eq:LP} minimizing the total \emph{relative evolution time}. 
Conjugation of $H_S$ with arbitrary single-qubit pulses yields a $\dd$ which is only limited by the locality of the interactions in $H_S$.
Below, we restrict ourselves to $\pi/2$ pulses achieving the same flexibility.

The first main contribution of our work is a hierarchy of relaxations of \eqref{eq:LP} with drastically reduced size that still allows for an exact decomposition as in \cref{eq:decomp_ex}.
For this, we need general sufficient conditions on the matrix $\WG$ such that \eqref{eq:LP} has a feasible solution 
for any $\vec A \in \R^\dd$.
\begin{definition}
\label{def:feasibleW}
 We say that a matrix $\WG \in \RR^{\dd \times \rr}$ is \emph{feasible} if for each $\vec A \in \RR^\dd$ there exists a $\vec \lambda \in \RR^\rr_{\geq 0}$ such that $\WG \vec \lambda = \vec A$.
\end{definition}
This definition captures the constraints in \eqref{eq:LP}. 
There is a simple sufficient condition for the feasibility of the matrix $\WG$:

\begin{proposition}
\label{thrm:feasibility}
 Let $\WG \in \RR^{\dd \times \rr}$ such that $\operatorname{ker} (\WG^T) = \{ \vec 0 \}$. If there exists $ \vec x \in \RR^\rr$ such that $\WG \vec x = \vec 0$ and $\vec x > \vec 0$, then for any $\vec A \in \RR^\dd$ there exists $ \vec x^\prime \in \RR^\rr$ such that $\WG \vec x^\prime = \vec A$ and $\vec x^\prime \geq \vec 0$.
\end{proposition}

This proposition follows directly from two well-known results in convex analysis.

\begin{lemma}[Farkas \cite{farkas1902}]\label{lem:farkas}
 Let $\WG \in \RR^{\dd \times \rr}$ and $\vec A \in \RR^\dd$. Then exactly one of the following assertions is true:
 \begin{enumerate}
  \item $\exists \vec x \in \RR^\rr$ such that $\WG \vec x = \vec A$ and $\vec x \geq \vec 0$.
  \item $\exists \vec y \in \RR^\dd$ such that $\WG^T \vec y \geq \vec 0$ and $\vec A^T \vec y < \vec 0$.
 \end{enumerate}
\end{lemma}
\begin{lemma}[Stiemke \cite{stiemke1915}]\label{lem:stiemke}
 Let $\WG\in \RR^{\dd \times \rr}$. Then exactly one of the following assertions is true:
 \begin{enumerate}
  \item $\exists \vec x \in \RR^\rr$ such that $\WG \vec x = \vec 0$ and $\vec x > \vec 0$.
  \item $\exists \vec y \in \RR^\dd$ such that $\WG^T \vec y \gneq \vec 0$.
 \end{enumerate}
\end{lemma}

\begin{proof}[Proof of Prop.~\ref{thrm:feasibility}]
We show that the second assertion of Farkas lemma is not possible.
By \cref{lem:stiemke}, there does not exist a $\vec y$ such that $\WG^T \vec y \gneq \vec 0$.
From $\operatorname{ker} (\WG^T) = \{ \vec 0 \}$ it follows that if $\WG^T \vec y = \vec 0$, then $\vec y = \vec 0$ which contradicts $\vec A^T \vec y < \vec 0$ in the second assertion of \cref{lem:farkas}.
Finally, the case $\WG^T \vec y < \vec 0$ directly contradicts the second assertion of \cref{lem:farkas}.
\end{proof}
To provide a geometric interpretation of \cref{thrm:feasibility} we first define the \emph{convex hull of the column vectors} of $\WG$
\begin{equation}
\begin{aligned}
 &\operatorname{conv} (\WG) \coloneqq \\
 &\Set*{ \vec u \in \R^\dd \given \WG \vec x = \vec u, \vec x \geq \vec 0, \sum_i x_i = 1} \, ,
\end{aligned}
\end{equation}
and the \emph{interior of a polytope} $P$
\begin{equation}
\begin{aligned}
 & \operatorname{int} (P) \coloneqq \\
 &\Set*{ \vec u \in P \given \exists \varepsilon > 0 \text{ s.t. } \| \vec u - \vec x \| < \varepsilon \,\, \forall \vec x \in \R^\dd \Rightarrow \vec x \in P } \, .
\end{aligned}
\end{equation}
\eqref{eq:LP} is feasible if the convex hull of the column vectors of $\WG$ has an interior that contains the origin, $\vec 0 \in \operatorname{int} (\operatorname{conv} (\WG))$.
The conditions of \cref{thrm:feasibility} can be efficiently verified:
If the \ac{LP}
\begin{mini}
{}{\vec 1^T \vec x}{}{}
\label{eq:LP1}
\addConstraint{\WG \vec x} {= \vec 0}{}{}
\addConstraint{\vec x}{\geq \vec 1}{}{}
\end{mini}
has a feasible solution and $\WG$ has full rank, then $\WG$ is feasible.
Note that \cref{thrm:feasibility} applies to general matrices and will be useful for the efficient relaxations.

\subsection{Pauli conjugation}
\label{sec:pauli_conjugation}
Conjugation of a system Hamiltonian $H_S$ with a Pauli string $P_{\vec b}$, i.e.\ a $\pi$ pulse layer, leads to
\begin{equation}
\label{eq:pauli_conjugation_system}
 P_{\vec b}^\dagger H_S P_{\vec b} = \sum_{\vec a \in \FF_2^{2n}\setminus \{ \vec 0 \}} (-1)^{\langle \vec a,\vec b \rangle} J_{\vec a} P_{\vec a} \, ,
\end{equation}
which follows from the commutation relations \eqref{eq:pauli_conjugation}.
We seek a decomposition of $H_T$ of the form
\begin{equation}\label{eq:H_TandW}
\begin{aligned}
	H_T &= \sum_{\vec a \in \FF_2^{2n}\setminus \{ \vec 0 \}} A_{\vec a} P_{\vec a} \\
	&= \sum_{\vec b \in \FF_2^{2n}} \lambda_{\vec b} P_{\vec b}^\dagger H_S P_{\vec b} \\
	&= \sum_{\vec a \neq \vec 0, \vec b } \lambda_{\vec b} (-1)^{\langle \vec a,\vec b \rangle} J_{\vec a} P_{\vec a}
 \, , \quad
 (\lambda_{\vec b} \geq 0) \,.
\end{aligned}
\end{equation}
\begin{figure*}[ht]
	\centering
	\includegraphics{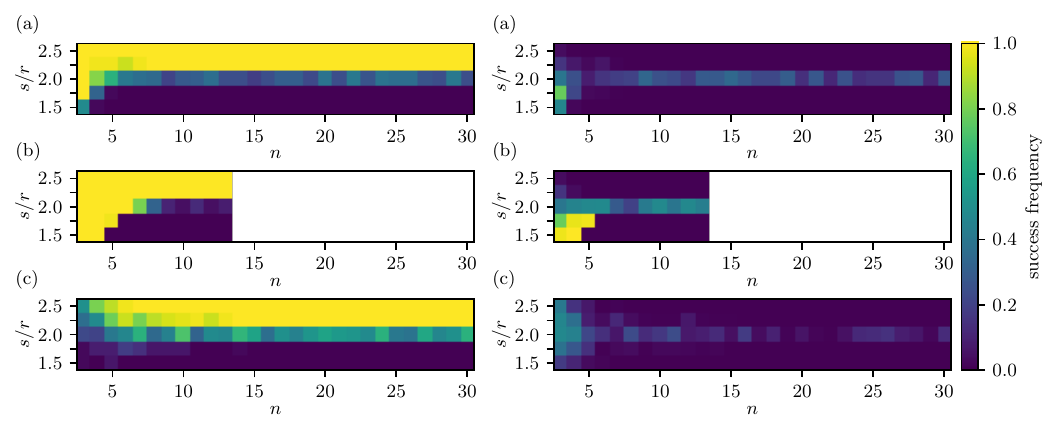}
	\caption{
 			We select $\dd$ rows of the Walsh-Hadamard matrix as determined by the considered interactions in the system Hamiltonian, to obtain the partial Walsh-Hadamard matrix $\WP{\dd}{4^n}$. 
 			Then we implement the subsampling of the possible pulse layers by randomly sampling $\rr$ columns to obtain $\WP{\dd}{\rr}$.
 			\textbf{Left:} The plots show the success frequency for ``$\WP{\dd}{\rr}$ is feasible'' (yellow) over $50$ samples for $n = 3, \dots , 30$ and $\rr/\dd = 1.5, \dots , 2.5$.
 			There is a sharp transition at $\rr/\dd=2$ in all cases.
 			\textbf{Right:} The difference of our numerical observation to Wendel's formula \eqref{eq:Wendel_p}.
            \textbf{(a)} and
 			\textbf{(b)}: Same 2- and 3-local system Hamiltonian as in \cref{fig:eff_heu_pauli_LP} (a) and (b).
 			\textbf{(b)}: We stopped the numerical experiments at $n=13$ due to the large amount of all-to-all interactions with locality $k \leq 3$.
 			\textbf{(c)}: Same random system Hamiltonian as in \cref{fig:eff_heu_pauli_LP} (c).
	}
	\label{fig:eff_heu_pauli}
\end{figure*}
The $4^n \times 4^n$ matrix with entries $\WPfull_{\vec a \vec b} = (-1)^{\langle \vec a,\vec b \rangle}$ is called the \emph{Walsh-Hadamard matrix}.
Here, we only need a submatrix $\WP{\dd}{\rr} \in \{-1,1\}^{\dd \times \rr}$ defined by choosing $\dd$ row-indices $\vec a\in \FF_2^{2n}$ and $\rr$ column-indices $\vec b\in \FF_2^{2n}$.
We call $\WP{\dd}{\rr}$ a \emph{partial Walsh-Hadamard matrix}.
By \cref{eq:H_TandW}, it is clear that only the interaction terms with $J_{\vec a} \neq 0$ contribute to the target Hamiltonian $H_T$.
Therefore, we define $\operatorname{nz}(\vec J) \coloneqq \Set*{\vec a \in \FF_2^{2n}\setminus \{ \vec 0 \} \given J_{\vec a} \neq 0}$, and require that the target interaction strengths satisfy
\begin{equation}
\label{eq:Pauli_nz_requirement}
 \operatorname{nz}(\vec A)\subseteq\operatorname{nz}(\vec J) \, .
\end{equation}
This requirement can be eliminated with single-qubit Clifford conjugation, which we will discuss in detail in \cref{sec:clifford_conjugation}.
In addition to the restriction $\vec a \in \operatorname{nz}(\vec J)$ given by the system Hamiltonian, we may also restrict the set of Pauli strings ($\pi$ pulse layers) $\vec b \in \mathcal{F} \subseteq \FF_2^{2n}$, as long as there is still a solution to \cref{eq:H_TandW}.
Comparing the target interaction strengths in \cref{eq:H_TandW}, we can write this as
\begin{equation}
 A_{\vec a} = J_{\vec a} \sum_{\vec b \in \mathcal{F}} \WP{\dd}{\rr}_{\vec a \vec b} \lambda_{\vec b} \quad \forall \vec a \in \operatorname{nz}(\vec J) \, ,
\end{equation}
with the number of non-zero system interaction strengths $\dd \coloneqq  \abs{\operatorname{nz}(\vec J)}$ and the number of considered Pauli strings $\rr = \abs{\mathcal{F}}$.
Hence, the constraint in \eqref{eq:LP} becomes $\vec A = \vec J \odot (\WP{\dd}{\rr} \vec \lambda)$, where $\odot$ denotes element-wise multiplication and $\vec A , \vec J \in \R^\dd$ are restricted to $\vec a \in \operatorname{nz}(\vec J)$.
For the following analysis, it is convenient to set $\vec M \coloneqq \vec A \oslash \vec J$ using the element-wise division $(\vec A \oslash \vec J)_{\vec a} \coloneqq A_{\vec a} / J_{\vec a}$ for all $\vec a \in \operatorname{nz}(\vec J)$.
This results in the following \ac{LP}:
\begin{equation}
\tag{PauliLP}
\label{eq:PauliLP}
\begin{aligned}
 \mathrm{minimize} &&& \vec 1^T \vec \lambda  \\
 \mathrm{subject\ to}    &&& \WP{\dd}{\rr} \vec \lambda = \vec M \,, \; \vec \lambda \in \R^{\rr}_{\geq 0} \, .
\end{aligned}
\end{equation}
The runtime of a \acs{LP} scales with its size, i.e.~the number of constraints $\dd$ and the number of variables $\rr$.
We discuss in detail the existence of a solution $\vec \lambda \in \R^{4^n}_{\geq 0}$ and bounds on the total relative evolution time $\vec 1^T \vec \lambda$ in \cref{app:solution_exist_bounds}.
Here, $\dd = \abs{\operatorname{nz}(\vec J)} \leq 4^n-1$ is fixed by the system Hamiltonian and a priori $\rr = 4^n$, thus solving \eqref{eq:PauliLP} becomes computationally intensive already at moderate system size.
To overcome this, we propose a simple and efficient relaxation of \eqref{eq:PauliLP} in \cref{sec:eff_pauli_heu}: 
For any $\dd = \abs{ \operatorname{nz}(\vec J)}$, we sample $\rr \geq 2\dd$ Pauli strings at random, leading to a feasible \ac{LP} with high probability.
As we have $\dd=\poly(n)$ for local Hamiltonians, this indeed yields an efficient relaxation in the system size $n$.
Note that this relaxation still leads to an exact decomposition; however, the relative evolution time $\vec 1^T \vec \lambda$ may not be minimal anymore.

\subsubsection{Efficient relaxation}
\label{sec:eff_pauli_heu}
We provide an efficient method to construct a feasible matrix $\WP{\dd}{\rr}$ with $\dd = \abs{ \operatorname{nz}(\vec J)}$ rows and $\rr\geq 2\dd$ columns.
The construction is rather simple and consists of sampling $\rr$ vectors $\vec b \in \FF_2^{2n}$ uniformly at random and taking the corresponding partial Walsh-Hadamard matrix with entries $\WP{\dd}{\rr}_{\vec a \vec b} = (-1)^{\langle \vec a,\vec b \rangle}$ where $\vec a \in \operatorname{nz}(\vec J)$.
Thus, engineering a Hamiltonian with $\dd$ non-zero interaction terms leads to a \eqref{eq:PauliLP}, whose number of constraints and variables both scale linearly with $\dd$.

To ensure feasibility of the sub-sampled matrix, we invoke \cref{thrm:feasibility}.
Thus, we have to check that $\conv (\WP{\dd}{\rr})$ has a non-empty interior and $\vec 0 \in \conv (\WP{\dd}{\rr})$.
First, we state general results for i.i.d.\ copies $\vec x_1, \dots, \vec x_\rr$ of a random vector $\vec x$.
Then, we relate the results to the partial Walsh-Hadamard matrix $\WP{\dd}{\rr}$.
Here, we assume that $\rr$ is large enough and that $\vec x_1, \dots, \vec x_\rr$ are in general position such that $\conv (\vec x_1, \dots, \vec x_\rr)$ always has non-empty interior.
Thus, we focus on the condition $\vec 0 \in \conv (\vec x_1, \dots, \vec x_\rr)$.

\begin{definition}\label{def:p_rx}
 Let $\vec x_1, \dots, \vec x_\rr$ be i.i.d.\ copies of an arbitrary random vector $\vec x$ in $\R^\dd$ and define
 \begin{equation}\label{eq:r_rx}
  p_{\rr,\vec x} \coloneqq \PP (\vec 0 \in \conv (\vec x_1, \dots, \vec x_\rr)) \, .
 \end{equation}
\end{definition}

Intuitively, one would expect that the probability should increase quickly with the number of samples $\rr$, and taking $\rr$ of the order of $\dd$ should make $p_{\rr,\vec x}$ reasonably large.

\begin{lemma}[\cite{Hayakawa2023}, Proposition 4] \label{prop:monotonicity}
 Let $\vec x \in \R^\dd$ be an arbitrary random vector with $\mathbb{E} [\vec x] = 0$ and $\PP (\vec x \neq 0) > 0$. Then we have
 \begin{equation}
  0 < p_{\dd+1,\vec x} < p_{\dd+2,\vec x} < \dots < p_{\dd+l,\vec x} \rightarrow 1 \quad \text{for } l\rightarrow \infty\, ,
 \end{equation}
 and $p_{\rr , \vec x} = 0$ if $\rr \leq \dd$.
\end{lemma}

Indeed, letting $\vec w$ be the $\dd$-dimensional random vector drawn uniformly from the columns of the 
partial Walsh-Hadamard matrix $\WP{\dd}{4^n}$, we can readily verify $\mathbb{E} [\vec w] = 0$ and thus \cref{prop:monotonicity} applies.
However, the question of how large $\rr$ shall be taken still remains.
In the convex geometry literature, we find an elegant solution, at least for spherically symmetric random vectors, in the form of \textbf{Wendel's theorem} \cite{wendel1962}:
If the random vector $\vec x$ has a spherically symmetric distribution around $\vec 0$,
then the probability \eqref{eq:r_rx} is given by
 \begin{equation}\label{eq:Wendel_p}
  p_{\rr,\vec x} = 1 - \frac{1}{2^{\rr-1}} \sum_{k=0}^{\dd-1} \binom{\rr-1}{k} \, .
 \end{equation}
This distribution shows a sharp transition from $\approx 0$ to $\approx 1$ at $\rr = 2 \dd$. 
Wendel's theorem does however not apply to the random columns $\vec w$ of $\WP{\dd}{4^n}$ because of the lack of spherical symmetry (see \cref{app:eff_pauli_heu} for details).
Instead, the result \eqref{eq:Wendel_p} of Wendel's theorem provides an upper bound on the probabilities $p_{\rr,\vec w}$ \cite{Wagner2001,Hayakawa2023}, and adequate lower bounds prove to be tricky to derive.

Nevertheless, we numerically observe that the random columns $\vec w$ of $\WP{\dd}{4^n}$ follow the behavior that we would expect from Wendel's theorem:
We observe a sharp transition of $p_{\rr,\vec w}$ at the same position as for spherical symmetric distributions, approximating the upper bound \eqref{eq:Wendel_p}.
In this sense, $p_{\rr,\vec w}$ is optimal, maximizing the success probability for finding a feasible $\WP{\dd}{\rr}$.

\begin{observation}[\cref{fig:eff_heu_pauli}]
\label{conj:feasibility}
 Let $\WP{\dd}{4^n}$ be a partial Walsh-Hadamard matrix with $\dd = \abs{ \operatorname{nz}(\vec J)}$, and let $\vec w$ be a $\dd$-dimensional random vector drawn uniformly from $\operatorname{col} (\WP{\dd}{4^n})$.
 Then, we numerically observe that Wendel's statement \eqref{eq:Wendel_p} approximately holds, i.e.~the random submatrix $\WP{\dd}{\rr}$ is feasible with high probability provided we take $\rr \geq 2\dd$ samples of $\vec w$.
\end{observation}

Let $\WP{\dd}{\rr}$ be the partial Walsh-Hadamard matrix obtained from $\WP{\dd}{4^n}$ by randomly sampling $\rr$ columns $\vec b \overset{\text{i.i.d.}}{\sim} \operatorname{unif}(\FF_2^{2n})$ with $\rr\geq 2\dd$.
As  $\WP{\dd}{\rr}$ is feasible with high probability, it can be used in \eqref{eq:PauliLP} to engineer any target Hamiltonian with a Pauli decomposition compatible with the system Hamiltonian, i.e.~with $\operatorname{nz}(\vec A)\subseteq\operatorname{nz}(\vec J)$.
One $\WP{\dd}{\rr}$ can be reused for different Hamiltonians with a Pauli decomposition with $\dd$ terms, since the construction of $\WP{\dd}{\rr}$ is independent of the choice of $\vec a \in \FF_2^{2n}\setminus \{ \vec 0 \}$.
As $\dd = \poly(n)$ for local Hamiltonians, the relaxed \eqref{eq:PauliLP} can be solved efficiently, and it provides an exact decomposition of the target Hamiltonian.
The evolution time $\vec 1^T \vec \lambda$ may however not be minimal anymore.
The quality of the relaxation can be improved by increasing $\rr$, thus expanding the search space.
This provides a flexible trade-off between the runtime of \eqref{eq:PauliLP} and the optimality of $\vec \lambda$, see \cref{fig:eff_heu_pauli_LP} (a-c).

\subsection{Clifford conjugation}
\label{sec:clifford_conjugation}
In \cref{sec:pauli_conjugation}, we introduced the Hamiltonian engineering method based on conjugation with $\pi$ pulses.
In this section, we extend this method to a certain set of single-qubit Clifford gates, consisting of $\pi/2$ pulses.
This extension allows us to change the interaction terms in the system Hamiltonian, such the only restriction in engineering a target Hamiltonian comes from the locality of the interactions in the system Hamiltonian $H_S$.
The same restriction also holds for conjugation with arbitrary single-qubit pulses, making Clifford conjugation a powerful Hamiltonian engineering method.
The simplest set of $\pi/2$ pulses with full expressivity are also called the square-root Pauli gates
\begin{equation}
\begin{aligned}
 \sqrt{X} &=\frac12 \begin{pmatrix} 1+\i & 1-\i \\ 1-\i & 1+\i\end{pmatrix}\, , \quad
 \sqrt{Y} =\frac12 \begin{pmatrix} 1+\i & -1-\i \\ 1+\i & 1+\i\end{pmatrix}\, , \\
 \sqrt{Z} &= \begin{pmatrix} 1 & 0 \\ 0 & \i \end{pmatrix} \, .
\end{aligned}
\end{equation}
\begin{table}[H]
	\centering
 \begin{tabular}{|c||c|c|c|c|c|c|}\hline
 $S^\dagger P S$ & $\sqrt{X}$ & $\sqrt{Y}$ & $\sqrt{Z}$ & $\sqrt{X}^\dagger$ & $\sqrt{Y}^\dagger$ & $\sqrt{Z}^\dagger$ \\\hline\hline
 $X$ &   $X$ &  $Z$ & $-Y$ &  $X$ & $-Z$ & $Y$  \\\hline
 $Y$ &  $-Z$ &  $Y$ &  $X$ &  $Z$ &  $Y$ & $-X$ \\\hline
 $Z$ &   $Y$ & $-X$ &  $Z$ & $-Y$ &  $X$ &  $Z$ \\\hline
\end{tabular}
	\caption{Conjugation of Pauli operators $P \in \mathsf{P}$ with square-root Pauli gates ($\pi/2$ pulses) $S \in \{ \sqrt{X}, \sqrt{Y}, \sqrt{Z}, \sqrt{X}^\dagger, \sqrt{Y}^\dagger, \sqrt{Z}^\dagger\}$.}
	\label{tab:sqrt_pauli_conjugation}
\end{table}
As displayed in \cref{tab:sqrt_pauli_conjugation}, the linear transformation  $S^\dagger P S$ of a Pauli operator with a $\pi/2$ pulse does not only flip the sign but also changes the interaction type (rotation axis).
The change of the interaction type increases the set of Hamiltonians reachable by Clifford conjugation compared to the Pauli conjugation.
In the following, we consider a gate set that behaves very similar to the conjugation with $\pi$ pulses:
\begin{equation}
 \mathcal{C}_{XY} \coloneqq \{Z\} \cup \Set*{QD \given Q,D \in \{\sqrt{X}, \sqrt{Y}, \sqrt{X}^\dagger, \sqrt{Y}^\dagger\}} \, .
\end{equation}
Likewise, we can define gate sets $\mathcal{C}_{ZY}$ and $\mathcal{C}_{XZ}$, for which conclusions analogous to the following ones may be drawn.
For example, the transformation given by the conjugation of interaction terms with the Clifford gates $\sqrt{X}\sqrt{Y}$ (two $\pi/2$ pulses) or $\sqrt{Y}^\dagger\sqrt{X}^\dagger$ changes the interaction type and leaves the sign of the Pauli coefficient unchanged.
The signs of the conjugated interaction terms depend on the rotation direction, i.e.\ the placement of ``$\dagger$'', of the square-root Pauli gates; see \cref{tab:cliff_conjugation} for examples.
Therefore, we label an element in $\mathcal{C}_{XY}$ by $c \coloneqq (p, \vec b) \in \FF_3 \times \FF_2^{2}$, where $p \in \FF_3$ represents the changes in the interaction type and $\vec b \in \FF_2^2$ captures the sign flips similar to the Pauli conjugation.
Denote $\mathcal{C}_{XY}^{\otimes n}$ a string of single-qubit gates on $n$ qubits from the gate set $\mathcal{C}_{XY}$.
We label each $S_{\vec c} \in \mathcal{C}_{XY}^{\otimes n}$ with $ \vec c = (c_1, \dots, c_n) = (\vec p, \vec b) \in \FF_{3}^n \times \FF_{2}^{2n}$, where $c_i \in \FF_{3} \times \FF_{2}^{2}$ represents the single-qubit Clifford gate from $\mathcal{C}_{XY}$ on the $i$-th qubit.
In \cref{tab:cliff_conjugation}, we show the effect of conjugating an interaction term $P_{\vec a} \in \mathsf{P}$ with $S_{c} \in \mathcal{C}_{XY}$.

\begin{figure*}[ht]
	\centering
	\includegraphics{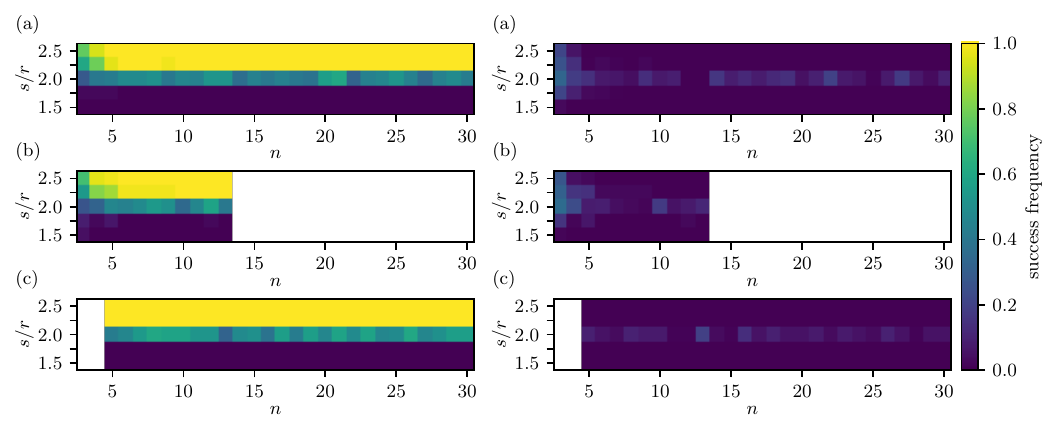}
	\caption{
			Select $\dd$ rows of $\WCfull$, determined by the considered interactions in the system Hamiltonian, to obtain the partial martrix $\WC{\dd}{12^n}$.
			Then we implement the subsampling of the possible pulse layers by randomly sampling $\rr$ columns to obtain $\WC{\dd}{\rr}$.
 			\textbf{Left:} The plots show the success frequency over $50$ samples for ``$\WC{\dd}{\rr}$ is feasible'' (yellow) for each $n = 3, \dots , 30$ and $\rr/\dd = 1.5, \dots , 2.5$.
 			There is a sharp transition at $\rr/\dd=2$ in all cases.
 			\textbf{Right:} The difference of our numerical observation from Wendel's statement \eqref{eq:Wendel_p}.
 			\textbf{(a)} and \textbf{(b)}:
              Same 2- and 3-local system Hamiltonian as in \cref{fig:eff_heu_pauli_LP} (d) and (e), respectively.
             \textbf{(b)}: We stopped the numerical experiments at $n=13$ due to the large amount of all-to-all interactions with locality $k \leq 3$.
 			\textbf{(c)}: Same random 5-local system Hamiltonian as in \cref{fig:eff_heu_pauli_LP} (f).
}
	\label{fig:eff_heu_cliff}
\end{figure*}
\begin{table}
	\centering
\setlength\tabcolsep{0pt}
\begin{tabular}{|@{\hspace{0.1cm}}c@{\hspace{0.1cm}}|c|c||c|c|c|c|}\hline
\multicolumn{2}{|@{\hspace{0.1cm}}c@{\hspace{0.1cm}}|}{\diagbox[width=2cm,height=1.2\line]{$\vec c$}{$\vec a$}} & $(a_x , a_z)$ & $(0,0)$ & $(1,0)$ & $(1,1)$ & $(0,1)$ \\\hline
$p$           & $\vec b$ & \diagbox[width=2cm,height=1.2\line]{$S_{\vec c}$}{$P_{\vec a}$} &  $I$ &  $X$ &  $Y$ &  $Z$ \\\hline\hline
\multirow{4}{*}{$0$} & $(0,0)$ &            $I$ &  \cellcolor{blue!15}$I$ & \cellcolor{blue!15} $X$ & \cellcolor{blue!15} $Y$ & \cellcolor{blue!15} $Z$ \\\hhline{|~*{6}{|-}|}
                   & $(1,0)$ &              $X$ &  \cellcolor{blue!15}$I$ & \cellcolor{blue!15} $X$ & \cellcolor{blue!15}$-Y$ & \cellcolor{blue!15}$-Z$ \\\hhline{|~*{6}{|-}|}
                   & $(1,1)$ &              $Y$ &  \cellcolor{blue!15}$I$ & \cellcolor{blue!15}$-X$ & \cellcolor{blue!15} $Y$ & \cellcolor{blue!15}$-Z$ \\\hhline{|~*{6}{|-}|}
                   & $(0,1)$ &              $Z$ &  \cellcolor{blue!15}$I$ & \cellcolor{blue!15}$-X$ & \cellcolor{blue!15}$-Y$ & \cellcolor{blue!15} $Z$ \\\hline\hline
\multirow{4}{*}{$1$}
 & $(0,0)$ & $\sqrt{X}\sqrt{Y}$                 &  \cellcolor{blue!15}$I$ & \cellcolor{blue!15} $Z$ & \cellcolor{blue!15} $X$ & \cellcolor{blue!15} $Y$ \\\hhline{|~*{6}{|-}|}
 & $(1,0)$ & $\sqrt{X}^\dagger\sqrt{Y}$         &  \cellcolor{blue!15}$I$ & \cellcolor{blue!15} $Z$ & \cellcolor{blue!15}$-X$ & \cellcolor{blue!15}$-Y$ \\\hhline{|~*{6}{|-}|}
 & $(1,1)$ & $\sqrt{X}^\dagger\sqrt{Y}^\dagger$ &  \cellcolor{blue!15}$I$ & \cellcolor{blue!15}$-Z$ & \cellcolor{blue!15} $X$ & \cellcolor{blue!15}$-Y$ \\\hhline{|~*{6}{|-}|}
 & $(0,1)$ & $\sqrt{X}\sqrt{Y}^\dagger$         &  \cellcolor{blue!15}$I$ & \cellcolor{blue!15}$-Z$ & \cellcolor{blue!15}$-X$ & \cellcolor{blue!15} $Y$ \\\hline\hline
\multirow{4}{*}{$2$}
 & $(0,0)$ & $\sqrt{Y}^\dagger\sqrt{X}^\dagger$ &  \cellcolor{blue!15}$I$ & \cellcolor{blue!15} $Y$ & \cellcolor{blue!15} $Z$ & \cellcolor{blue!15} $X$ \\\hhline{|~*{6}{|-}|}
 & $(1,0)$ & $\sqrt{Y}\sqrt{X}$                 &  \cellcolor{blue!15}$I$ & \cellcolor{blue!15} $Y$ & \cellcolor{blue!15}$-Z$ & \cellcolor{blue!15}$-X$ \\\hhline{|~*{6}{|-}|}
 & $(1,1)$ & $\sqrt{Y}\sqrt{X}^\dagger$         &  \cellcolor{blue!15}$I$ & \cellcolor{blue!15}$-Y$ & \cellcolor{blue!15} $Z$ & \cellcolor{blue!15}$-X$ \\\hhline{|~*{6}{|-}|}
 & $(0,1)$ & $\sqrt{Y}^\dagger\sqrt{X}$         &  \cellcolor{blue!15}$I$ & \cellcolor{blue!15}$-Y$ & \cellcolor{blue!15}$-Z$ & \cellcolor{blue!15} $X$ \\\hline
\end{tabular}
	\caption{The coloured cells show the conjugated interaction terms $S_{c}^\dagger P_{\vec a} S_{c}$ for an interaction $P_{\vec a} \in \mathsf{P}$ conjugated by the single-qubit Clifford gate $S_{c} \in \mathcal{C}_{XY}$.
	We label an element in $\mathcal{C}_{XY}$ by $c \coloneqq (p, \vec b) \in \FF_3 \times \FF_2^{2}$, where $p \in \FF_3$ captures the permutation of the conjugated interaction terms, and $\vec b$ captures the sign flips similar to the Pauli conjugation.
	This can be easily verified with \cref{tab:sqrt_pauli_conjugation}.
	}
	\label{tab:cliff_conjugation}
\end{table}
The transformation of $H_S$ with respect to $S_{\vec c} \in \mathcal{C}_{XY}^{\otimes n}$ leads to
\begin{equation}
\begin{aligned}
 S_{\vec c}^\dagger H_S S_{\vec c} &= \sum_{\vec a \in \FF_2^{2n}\setminus \{ \vec 0 \}} J_{\vec a} S_{\vec c}^\dagger P_{\vec a} S_{\vec c} \\
 &= \sum_{\vec a \in \FF_2^{2n}\setminus \{ \vec 0 \}} (-1)^{\langle \pi_{\vec p}(\vec a),\vec b \rangle} J_{\pi_{\vec p}(\vec a)} P_{\vec a} \, ,
\end{aligned}
\end{equation}
with $\vec c  = (\vec p, \vec b) \in \FF_{3}^n \times \FF_{2}^{2n}$ and the permutation $\pi_{\vec p} : \FF_2^{2n} \rightarrow \FF_2^{2n}$ with $\pi_{\vec p} (\vec a) \coloneqq \left( \pi_{p_1}(a_1), \dots , \pi_{p_n}(a_n) \right)$ given by the local permutations $\pi_{0}, \pi_{1}, \pi_{2}: \FF_2^2 \rightarrow \FF_2^2$ in the two-line notation
\begin{equation}
\begin{aligned}
 \pi_{0} &\coloneqq \begin{pmatrix} (0,0) & (1,0) & (1,1) & (0,1) \\
								   (0,0) & (1,0) & (1,1) & (0,1) \end{pmatrix}\, , \\
 \pi_{1} &\coloneqq \begin{pmatrix} (0,0) & (1,0) & (1,1) & (0,1) \\
								   (0,0) & (1,1) & (0,1) & (1,0) \end{pmatrix}\, , \\
 \pi_{2} &\coloneqq \begin{pmatrix} (0,0) & (1,0) & (1,1) & (0,1) \\
								   (0,0) & (0,1) & (1,0) & (1,1) \end{pmatrix}\, .
\end{aligned}
\end{equation}
We want to find a decomposition of $H_T$ with $\lambda_{\vec c} \geq 0$ such that
\begin{equation}
 H_T = \sum_{\vec a \in \FF_2^{2n}\setminus \{ \vec 0 \}} A_{\vec a} P_{\vec a} = \sum_{\vec c \in \FF_{3}^n \times \FF_{2}^{2n}} \lambda_{\vec c} S_{\vec c}^\dagger H_S S_{\vec c} \, .
\end{equation}
We define the matrix $\WCfull \in \R^{(4^n - 1) \times 12^n}$ with entries
\begin{equation}
\label{eq:cliff_W_entries}
 \WCfull_{\vec a \vec c} \coloneqq (-1)^{\langle \pi_{\vec p}(\vec a),\vec b \rangle} J_{\pi_{\vec p}(\vec a)} \, ,
\end{equation}
for $\vec a \in \FF_2^{2n}\setminus \{ \vec 0 \}$, excluding the identity term $P_{\vec a} = I^{\otimes n}$ with $\vec a = \vec 0$, and $\vec c \in \FF_{3}^n \times \FF_{2}^{2n}$.
Similar to the Pauli conjugation we define the submatrix $\WC{\dd}{\rr} \in \R^{\dd \times \rr}$ with entries given in \cref{eq:cliff_W_entries} for $\dd$ row-indices $\vec a \in \FF_2^{2n}\setminus \{ \vec 0 \}$ and $\rr$ column-indices $\vec c \in \FF_{3}^n \times \FF_{2}^{2n}$.
Up to the permutation of $\vec J$, the same Walsh-Hadamard matrix structure as in $\WP{\dd}{\rr}$ from \cref{sec:pauli_conjugation} is present in $\WC{\dd}{\rr}$.
In terms of the matrix $\WCfull$ it follows that
\begin{equation}
\label{eq:cliff_H_S_to_H_T}
 H_T = \sum_{\vec a \in \FF_2^{2n}\setminus \{ \vec 0 \}} \sum_{\vec c \in \FF_{3}^n \times \FF_{2}^{2n}} \lambda_{\vec c} \WCfull_{\vec a \vec c} P_{\vec a} \, .
\end{equation}
Due to the permutation in \cref{eq:cliff_W_entries} we are no longer restricted by the non-zero coefficients as in the Pauli conjugation.
From \cref{eq:cliff_H_S_to_H_T} it follows that for each $A_{\vec a} \neq 0$ there has to be at least one $J_{\hat{\vec a}} \neq 0$ such that $\operatorname{supp}(\hat{\vec a}) = \operatorname{supp}(\vec a)$.
Therefore, we define
\begin{equation}
\begin{aligned}
 \operatorname{suppnz} (\vec J) \coloneqq \left\{\vec a \in \FF_2^{2n}\setminus \{ \vec 0 \} \right| &\exists \hat{\vec a} \in \operatorname{nz}(\vec J), \\
 & \left. \operatorname{supp}(\hat{\vec a}) = \operatorname{supp}(\vec a)\right\} \, ,
\end{aligned}
\end{equation}
and require that the Pauli coefficients satisfy
\begin{equation}
\label{eq:cliff_nz_requirement}
\operatorname{nz}(\vec A) \subseteq \operatorname{suppnz} (\vec J) \, .
\end{equation}
In physical terms, this means that we are only restricted by the locality of the interactions in the system Hamiltonian $H_S$ and have full flexibility in the kind of interactions $P_{\vec a}$ and the interaction strength $A_{\vec a}$.
In addition to the restriction $\vec a \in \operatorname{suppnz} (\vec J)$, given by the system Hamiltonian, we can also consider a restricted set of conjugating Clifford strings $\vec c \in \mathcal{F} \subseteq \FF_{3}^n \times \FF_{2}^{2n}$, as long as there still exists a solution.
Then, the restricted \cref{eq:cliff_H_S_to_H_T} can be reformulated as
\begin{equation}
 A_{\vec a} = \sum_{\vec c \in \mathcal{F}} \WC{\dd}{\rr}_{\vec a \vec c} \lambda_{\vec c} \quad \forall \vec a \in \operatorname{suppnz} (\vec J) \, ,
\end{equation}
where $\dd \coloneqq \abs{\operatorname{suppnz} (\vec J)}$ denotes the number interactions that can be generated and $\rr = \abs{\mathcal{F}}$ denotes the number of considered Clifford strings.
The constraint in \eqref{eq:LP} then reads $\vec A = \WC{\dd}{\rr} \vec \lambda$, with $\vec A , \vec J \in \R^\dd$ restricted to $\vec a \in \operatorname{suppnz}(\vec J)$.
With that, we define the following \ac{LP}
\begin{equation}
\tag{CliffLP}
\label{eq:LPcliff}
\begin{aligned}
 \mathrm{minimize} &&& \vec 1^T \vec \lambda  \\
 \mathrm{subject\ to}    &&& \WC{\dd}{\rr} \vec \lambda = \vec A \,, \; \vec \lambda \in \R^{\rr}_{\geq 0} \, .
\end{aligned}
\end{equation}
There always exists a feasible solution of \eqref{eq:LPcliff} with $\WC{\dd}{12^n}$, which follows similarly to \cref{cor:full_W_feasibility} from the Walsh-Hadamard structure in $\WC{\dd}{12^n}$.
In general, we have $\rr \leq 12^n$ and $\dd \leq 4^n - 1$.
In the next section, we propose a simple and efficient relaxation of \eqref{eq:LPcliff}, where $\rr \geq 2 \dd$ can be chosen for arbitrary $\dd = \abs{ \operatorname{suppnz} (\vec J)}$.

\subsubsection{Efficient relaxation}
\label{sec:eff_cliff_heu}
The size of \eqref{eq:LPcliff} scales as $\LandauO (12^n)$ if all possible Clifford layers are considered.
Therefore, a relaxation is necessary to solve this \ac{LP} in practice.
Analogous to the Pauli conjugation method, \cref{sec:eff_pauli_heu}, sampling $\rr\geq 2\dd$ columns at random from $\WC{\dd}{12^n}$ results in a feasible matrix $\WC{\dd}{\rr}$ with high probability.
As in \cref{sec:eff_pauli_heu}, we test this statement numerically by checking the feasibility conditions of \cref{thrm:feasibility} for random submatrices.

\begin{observation}[\cref{fig:eff_heu_cliff}]
\label{conj:feasibility_cliff}
 Let $\WC{\dd}{12^n}$ be a matrix with $\dd = \abs{\operatorname{suppnz} (\vec J)}$ and entries as in \cref{eq:cliff_W_entries}, and let $\vec w$ be a $\dd$-dimensional random vector drawn uniformly from $\operatorname{col} (\WC{\dd}{12^n})$.
 Then, we numerically observe that Wendel's statement \eqref{eq:Wendel_p} approximately holds, i.e.~the random submatrix $\WC{\dd}{\rr}$ is feasible with high probability provided we take $\rr \geq 2\dd$ samples of $\vec w$.
\end{observation}

From the definition of feasibility, we know that \eqref{eq:LPcliff} always has a solution for arbitrary $\vec A$, i.e.\ arbitrary $\vec J, \vec A$ such that $\operatorname{nz}(\vec A) \subseteq \operatorname{suppnz} (\vec J)$.
Furthermore, the quality of the relaxation can be improved by increasing $\rr$, which leads to an expanded search space, see \cref{fig:eff_heu_pauli_LP} (d-f).
This provides again a flexible trade-off between the runtime of \eqref{eq:LPcliff} and the optimality of $\vec \lambda$.

\subsection{Hamiltonian engineering with unknown Hamiltonians}
\label{sec:unknown_hermitian_op}
In practice, sometimes not all coupling coefficients $J_{\vec a}$ in  the system Hamiltonian might be known.
For example, when an experiment aims to realize two-body interactions but also acquires unwanted three-body terms of unknown strength.
Such system Hamiltonians with unknown or only partially known coupling strengths can still be used for engineering based on the Pauli conjugation method from \cref{sec:pauli_conjugation}.
Solving \eqref{eq:PauliLP} for an $\vec M \in \R^\dd$ yields the target Hamiltonian
\begin{equation}
 \sum_{\vec a \in \FF_2^{2n}\setminus \{ \vec 0 \}} M_{\vec a} J_{\vec a} P_{\vec a} = H_T \, .
\end{equation}
The potentially unknown coefficients $J_{\vec a}$ are modified by an element-wise multiplication $\vec M \odot \vec J$.
Interesting choices for $M_{\vec a}$ are $-1$ and $0$, inverting the sign of or canceling the interaction term $P_{\vec a}$, respectively, without the knowledge of $J_{\vec a}$.
For each term in the system Hamiltonian, a different $M_{\vec a}$ can be chosen.
Therefore, engineering known terms $M_{\vec a} = A_{\vec a} / J_{\vec a}$ in the Hamiltonian while canceling or inverting other (potentially unknown) terms $M_{\vec a} = 0$ or $M_{\vec a} = -1$ is possible.
Using such a $\vec M$ in \eqref{eq:PauliLP} with dummy values for the unknown $J_{\vec a}$ inverts the signs or cancels the interactions.

In \cref{sec:2D_lattice_model} we applied the approach to cancel unknown 3-body terms in a 2D lattice Hamiltonian.

\section{Error robustness and mitigation techniques}
\label{sec:robustness}
To successfully apply the Pauli or Clifford conjugation in practice it is necessary to make the resulting pulse sequences robust against dominant errors.
In this section, we provide mitigation techniques for our efficient conjugation methods, which come with little overhead.
The simultaneous action of the single-qubit pulse and system Hamiltonian is called the finite pulse time error.
It has been shown that this error is detrimental to approaches similar to ours \cite{Vicente2023}.
Therefore, our focus lies on the error due to a finite single-qubit pulse duration.
Furthermore, we provide a modification to combine our Clifford method with robust composite pulses, making it robust against many different errors occurring in experiments.
We achieve robustness against finite pulse time effects and rotation angle errors using a similar approach as in \citet{Votto23UniversalQuantumProcessors}.
Their study focuses on specifically designed $\pi$ pulse sequences (i.e.~Pauli gates), so-called Walsh sequences, for engineering XY-Hamiltonians.
In our work, we generalize their approach to general local system Hamiltonians and arbitrary single-qubit pulses. 

\Acf{AHT} is a well-known approach in \ac{NMR} that allows to investigate the dynamics under a time-dependent Hamiltonian by approximating it with the one of a time-dependent Hamiltonian \cite{Haeberlen1968,Haeberlen1976}.
This approximation is done by a low order Magnus expansion \cite{Magnus1954}.
We utilize \ac{AHT} to investigate the error due to a finite single-qubit pulse time.
We will heavily use that the error term in the average Hamiltonian has the same locality as the system Hamiltonian.

In the following, we give an explicit form of the finite pulse time error in the average Hamiltonian when interleaving a system Hamiltonian with arbitrarily many layers of single-qubit pulses.
For the Clifford conjugation method, we find that the finite pulse time error can be exactly cancelled by a slight modification of \eqref{eq:LPcliff}.
This general investigation enables us to also mitigate finite pulse time error in combination with other pulse errors by replacing the single-qubit Clifford pulses in our Clifford conjugation method with robust composite pulses.
The latter are designed to compensate for experimental errors by implementing a gate with a specific pulse sequence.
In this way, dominant error sources can be cancelled or suppressed, which includes rotation angle error (Rabi frequency errors), off-resonance errors \cite{Masamitsu2013,kukita2021}, phase errors \cite{Torosov2019}, pulse shape errors \cite{Genov2014,Genov2020,Wu2023}, or non-stationary, non-Markovian noise \cite{Kabytayev2014} or crosstalk \cite{Torosov2020}.
If needed, remaining errors might then be mitigated using software-based methods \cite{Garcia-Molina24MitigatingNoiseIn}. 

For the Pauli conjugation method, the finite pulse time error can only be partially mitigated by modifying \eqref{eq:PauliLP}.
However, the freedom in choosing the rotation direction in the $\pi$ pulse can be leveraged in addition to the modified \eqref{eq:PauliLP} to completely cancel the finite pulse time error term in the average Hamiltonian.
Moreover, this simultaneously cancels first order effects of rotation angle errors in the single-qubit pulses.

We introduce basic concepts required for the following chapters in \cref{sec:robust_preliminaries}.
Then, in \cref{sec:robust_general}, we start with the investigation of the finite pulse time error when conjugating a system Hamiltonian with arbitrary many general single-qubit pulses.
In \cref{sec:robust_cliff}, we apply the general results to present our robust Clifford conjugation methods.
Finally, in \cref{sec:robust_pauli}, we present our robust Pauli conjugation method.

\subsection{Preliminaries}
\label{sec:robust_preliminaries}
For our robust methods, we require two well-known approaches to approximate the time evolution under non-commuting Hamiltonians.
First, we present general Trotter product formulae to approximate the time evolution under a linear combination of time-independent Hamiltonians.
Recently, the performance of the Suzuki-Trotter approximation has been greatly improved \cite{Morales2022}.
Second, we present the Magnus expansion to approximate the time evolution under a time-dependent Hamiltonian by a time-independent Hamiltonian.

\subsubsection{Product formula}
The time evolution under a Hamiltonian $H = \sum_{i=1}^L H_i$ can generally be approximated by a \emph{product formula}
\begin{equation}
 \e^{-\i t H} \approx \e^{-\i \alpha_q H_{i_q}} \dots \e^{-\i \alpha_1 H_{i_1}} \, ,
\end{equation}
with the number of evolution steps $q$ and $i_1,\dots , i_q \in [L]$.
We call a product formula \emph{deterministic} if $i_1,\dots , i_q$ can be found by a deterministic algorithm.
For our methods we do not require, that $\alpha_1,\dots, \alpha_q \in \R$ are chosen deterministically.
The best known deterministic product formulae are the Suzuki-Trotter formulae \cite{Trotter1959,Suzuki1992}.
The first- and second order approximations are given by
\begin{equation}
\label{eq:first_order_trotter}
 \e^{-\i t H} \approx \left(\prodr_{i=1}^L \e^{-\i \frac{t}{\nTro} H_{i}} \right)^{\nTro} \eqqcolon S_1 (t/\nTro)^{\nTro}
\end{equation}
and
\begin{equation}
\begin{aligned}
\e^{-\i t H} &\approx \left(\prodl_{i=1}^L \e^{-\i \frac{t}{2 \nTro} H_{i}} \prodr_{i=1}^L \e^{-\i \frac{t}{2 \nTro} H_{i}} \right)^{\nTro} \\
&\eqqcolon S_2 (t/\nTro)^{\nTro} \, ,
\end{aligned}
\end{equation}
respectively, with the \emph{number of Trotter cycles} $\nTro$.
The number of Trotter cycles can be increased to improve the accuracy.
The $2k$-th order Suzuki-Trotter formula for $k>1$ is defined recursively by
\begin{equation}
 S_{2k}(t) \coloneqq S_{2k-2} (u_{k} t)^2 S_{2k-2} ((1-4u_{k}) t) S_{2k-2} (u_{k} t)^2 \, ,
\end{equation}
with $u_k \coloneqq (4-4^{(2k-1)^{-1}})^{-1}$.
For a $2k$-th order Suzuki-Trotter formula the approximation error in spectral norm is bounded by
\begin{equation}
\begin{aligned}
\label{eq:trotter_error_bound}
 \| S_{2k}(t/\nTro)^{\nTro} &- \e^{-\i t H} \| \\
 &\leq \LandauO \left(\left(t \sum_{i=1}^L \|H_i\| \right)^{2k+1} {\nTro}^{-2k} \right) \, ,
\end{aligned}
\end{equation}
with the spectral norm $\| H_i \|$ \cite{Childs2019}.

\subsubsection{Average Hamiltonian theory and the Magnus expansion}
Given a time-dependent Hamiltonian $H(t)$, it is sometimes useful to consider a time-independent effective or \emph{average Hamiltonian} $H_{\av}$ satisfying
\begin{equation}
 \mathcal{U}(T) \approx \e^{-\i T H_{\av}} \, ,
\end{equation}
where $\mathcal{U}(T)$ is the time evolution operator defined by the differential equation
\begin{equation}
  \frac{d \mathcal{U}(t)}{dt} =  - \i H(t) \mathcal{U}(t) \, , \quad \text{and} \quad \mathcal{U}(0) = \1 \, .
\end{equation}
This average Hamiltonian can be expressed by the Magnus expansion as follows
\begin{equation}
 H_{\AV} = H_{\av}^{(1)} + H_{\av}^{(2)} + \dots \, ,
\end{equation}
where the first and second order terms are explicitly given by
\begin{equation}
 H_{\av}^{(1)} \coloneqq \frac{1}{T} \int_0^T H(\tau) \rmd \tau
\end{equation}
and
\begin{equation}
H_{\av}^{(2)} \coloneqq \frac{1}{2 \i T} \int_0^T \int_0^\tau [H(\tau), H(\tau^\prime)] \rmd \tau^\prime \rmd \tau\, .
\end{equation}
The Magnus expansion converges if $\int_0^T \norm{ H(\tau) } \rmd \tau < \pi$ \cite{Moan2008}.
As a rule of thumb, the Magnus expansion converges rapidly if
\begin{equation}
\label{eq:magnus_approx_error}
 \max_{\tau \in [0,T]} \norm{ H(\tau) } T \ll 1 \, ,
\end{equation}
where the spectral norm is used \cite{Brinkmann2016}.
Throughout this work we only consider the first order approximation and write $H_{\av} = H_{\av}^{(1)}$.

\begin{figure*}[ht]
	\centering
	\includegraphics{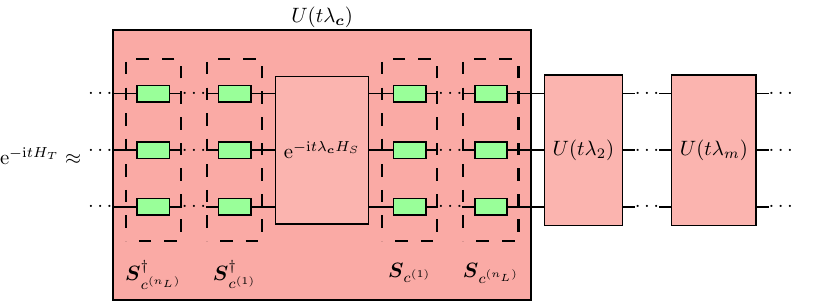}
	\caption{
			Exemplary quantum circuit for approximating the target evolution.
			We only implement simple single-qubit pulses in the presence of an always-on system Hamiltonian $H_S$.
	}
	\label{fig:circuit}
\end{figure*}

\subsection{General robust conjugation method}
\label{sec:robust_general}
We start with a general discussion of the finite pulse time error when conjugating a system Hamiltonian $H_S$ with multiple layers of general single-qubit pulses.
To this end, we first define the Hamiltonian of a general single-qubit pulse layer and the operators relevant for the following discussion.
In the following we define different compositions of single-qubit pulse layers, which allow a clear and compact presentation of the main results.
In particular, we define a single pulse layer, multiple consecutive pulse layers, which we call a pulse block, and a partial pulse block.

\begin{restatable}{restatableDef}{DefSingleGateLayer}
\label{def:operators_general_robust}
Let $\mc G \subset \Herm(\CC^2)$ be a set of \emph{single-qubit pulse generators}.
A \emph{single-qubit pulse layer} is given by a \emph{local generator} $h_i$ chosen from $\mc G$ for each qubit $i \in [n]$, the \emph{(single-qubit) rotation directions} $\vec s\in \FF_2^n$ and the \emph{finite pulse time} $t_p>0$.
The Hamiltonian generating the single-qubit pulse layer on $n$ qubits is given by
\begin{equation}
 H(t_p, \vec s, \vec h)
 \coloneqq
 \frac{1}{t_p} \sum_{i=1}^n (-1)^{s_i} h_i\, ,
\end{equation}
with $\vec h \coloneqq (h_1 , \dots , h_n)$.
The single-qubit pulse layer Hamiltonian is completely specified by the tuple $c \coloneqq (t_p, \vec s, \vec h)$, and we write $H_c \coloneqq H(t_p, \vec s, \vec h)$.
This generates the evolution operator for one single-qubit pulse layer
 \begin{equation}
  \vec S_c(t) \coloneqq \e^{- \i t H_c} \, , \quad \text{ and } \quad \vec S_c \coloneqq \vec S_c(t_p) \, .
 \end{equation}
 More generally, we consider a sequence of $\ns$ single-qubit pulse layers and introduce the layer index $\kk \in [\ns]$.
 The evolution for the $\kk$-th single-qubit pulse layer is specified by the tuple $c^{(\kk)} \coloneqq (t_p^{(\kk)}, \vec s^{(\kk)}, \vec h^{(\kk)})$, and we define the tuple $\vec c \coloneqq (c^{(1)}, \dots , c^{(\ns)})$.
 Moreover, we define the \emph{single-qubit pulse block} and the \emph{partial single-qubit pulse block} as
\begin{equation}
  \vec S_{\vec c} \coloneqq \prodr_{\kk=1}^{\ns} \vec S_{c^{(\kk)}}  \, ,
\end{equation}
and
\begin{equation}
\vec S_{\vec c^{\geq D}} \coloneqq \begin{cases}
		\prodr_{\kk=D}^{\ns} \vec S_{c^{(\kk)}} \, , & 1 \leq D \leq \ns \\
		\1 \, , & \text{otherwise} \, ,
	\end{cases}
\end{equation}
respectively.
Similarly, we define the pulse time for the single-qubit pulse block $T_p \coloneqq \sum_{\kk=1}^{\ns} t_p^{(\kk)}$ and for the partial single-qubit pulse block $T_p^{\leq D} \coloneqq \sum_{\kk=1}^{D} t_p^{(\kk)}$.
\end{restatable}

To illustrate \cref{def:operators_general_robust} we provide the Hamiltonian for the single-qubit Pauli pulses from \cref{sec:pauli_conjugation}.
The Hamiltonian for a single-qubit Pauli pulse layer is given by the $\pi$ pulse time $t_p$, an arbitrary rotation direction $\vec s \in \FF_2^n$ and the generators $h_i = \frac{\pi}{2} P_i$, with the Pauli operator $P_i$ given by $P_{\vec b_i}$ from \cref{eq:pauli_conjugation_system}.

Recall from \cref{sec:reshaping} that ideally, we would like to implement the conjugation $\vec S_{\vec c}^{\dagger} \e^{-\i t \lambda_{\vec c} H_S} \vec S_{\vec c} = \e^{-\i t \lambda_{\vec c} \vec S_{\vec c}^{\dagger} H_S \vec S_{\vec c}}$, with the system Hamiltonian $H_S$ and the \emph{relative evolution time} $\lambda_{\vec c}$ associated with a single-qubit pulse block $\vec S_{\vec c}$.
However, due to the finite duration of the single-qubit pulse and the always-on system Hamiltonian, we get the time evolution of the single-qubit pulse blocks
\begin{equation}
\vec S_{\err, \vec c} = \prodl_{\kk=1}^{\ns} \e^{- \i t_p^{(\kk)} (H_S - H_{c^{(\kk)}})}
\end{equation}
and
\begin{equation}
 \vec S_{\err, \vec c}^\prime = \prodr_{\kk=1}^{\ns} \e^{- \i t_p^{(\kk)} (H_S + H_{c^{(\kk)}})} \, ,
\end{equation}
respectively, with the finite duration of the $\kk$-th single-qubit pulse layer $t_p^{(\kk)}>0$.
In the absence of $H_S$, these two operators are exact inverses.
Then, the resulting evolution block with the finite pulse time error has the form
\begin{equation}
\label{eq:evo_block_finite_pulse_time}
 U(t \lambda_{\vec c}) \coloneqq \vec S_{\err, \vec c}^\prime \e^{-\i t \lambda_{\vec c} H_S} \vec S_{\err, \vec c} \, ,
\end{equation}
similar as in \cref{eq:evo_block_finite_pulse_time_sim}, and is depicted in \cref{fig:circuit}.

First, we provide the average Hamiltonian for one conjugation with a single-qubit pulse block to investigate the effect of the finite pulse duration.
\begin{restatable}{restatableLem}{LemHav}
 \label{lem:robust_finite_pulse_time}
 The approximation of $U(t \lambda_{\vec c})$ in first order Magnus expansion is given by $\e^{-\i H_{\av, \vec c} (t)}$ with
 \begin{equation}
 H_{\av, \vec c} (t) = t \lambda_{\vec c} \vec S_{\vec c}^{\dagger} H_S \vec S_{\vec c} + H_{\err, \vec c} ,
\end{equation}
where
\begin{equation}
\begin{aligned}
 &H_{\err, \vec c} \coloneqq \\
 &2 \sum_{\kk=1}^{\ns} \vec S_{\vec c^{\geq (\kk + 1)}}^{\dagger} \int_0^{t_p^{(\kk)}} \vec S_{c^{(\kk)}}^{\dagger} (t) H_S \vec S_{c^{(\kk)}} (t) \rmd t \, \vec S_{\vec c^{\geq (\kk + 1)}} \, .
\end{aligned}
\end{equation}
The average Hamiltonian $H_{\av, \vec c} (t)$ has the same locality as $H_S$.
Moreover, the error for truncating the Magnus expansion after the first order can be bounded in spectral norm by
\begin{equation}
\label{eq:magnus_trunc_error}
 \| U(t \lambda_{\vec c}) - \e^{-\i H_{\av, \vec c} (t)} \| \leq \LandauO((2T_p + t \lambda_{\vec c})^2 \| H_S \|^2) \, .
\end{equation}
\end{restatable}
The proof of \cref{lem:robust_finite_pulse_time} can be found in \cref{app:robust_general}.

Our goal is to formulate a \ac{LP} that cancels the finite pulse time error in the average Hamiltonian (in first order).
To this end, we define the matrices capturing the effect of the ideal dynamics $\vec S_{\vec c}^{\dagger} H_S \vec S_{\vec c}$ and the finite pulse time effect $H_{\err, \vec c}$.
The ideal dynamics is captured by the matrix $\WGG{\dd}{\rr} \in \R^{\dd \times \rr}$ with the same elements as in \cref{eq:general_ideal_matrix}.
The effect of the finite pulse time error is captured by the matrix $\EC{\dd}{\rr} \in \R^{\dd \times \rr}$ with the elements
\begin{equation}
\EC{\dd}{\rr}_{\vec a \vec c} \coloneqq \frac{1}{2^n} \Tr \left(P_{\vec a} H_{\err, \vec c} \right) \, ,
\end{equation}
and can be efficiently calculated for any local system Hamiltonian $H_S$, see \cref{lem:robust_LP} in the appendix.
Let $\WGG{\dd}{\rr}$ be feasible, i.e.\ let $\sum_{\vec c} \lambda_{\vec c} \vec S_{\vec c}^{\dagger} H_S \vec S_{\vec c}$ be able to modify all interaction terms with the same locality as the interactions in $H_S$.
 Then, the \ac{LP}
 \begin{equation}
\tag{robustLP}
\label{eq:robustLP}
\begin{aligned}
 \mathrm{minimize} &&& \vec 1^T \vec \lambda  \\
 \mathrm{subject\ to}    &&& \WGG{\dd}{\rr} \vec \lambda  + \EC{\dd}{\rr} \vec 1 = \vec A \,, \; \\
 &&& \vec \lambda \in \R^{\rr}_{\geq 0}
\end{aligned}
\end{equation}
always has a solution for any target interaction strength $\vec A \in \R^{\dd}$.
This follows directly from \cref{def:feasibleW}.
Now, we are ready to state the main result of this section.
\begin{restatable}{restatableThm}{MainThm}
\label{thrm:robust_finite_pulse_time}
The target time evolution $\e^{- \i t H_T}$ can be efficiently approximated by a deterministic product formula implementing a product of $U(t \lambda_{\vec c})$, with $\lambda_{\vec c}$ and the corresponding pulse block $\vec S_{\vec c}$ from \eqref{eq:robustLP}.
Moreover, this approximation is robust against the finite pulse time effect.
The only approximation errors are given by the ones from the Magnus approximation (\cref{lem:robust_finite_pulse_time}) and the approximation error from the deterministic product formula.
\end{restatable}
The proof of \cref{thrm:robust_finite_pulse_time} can be found in \cref{app:robust_general}.

\begin{remark}
The Trotter approximation error can be bounded as in \cref{eq:trotter_error_bound} and depends on the evolution time $t \sum_{\vec c}\lambda_{\vec c}$.
The approximation error for truncating the Magnus expansion increases with the finite pulse duration, see \cref{eq:magnus_trunc_error}.
Note, that the truncation error bound in \cref{eq:magnus_trunc_error} is not tight and might be improved \cite{Sharma2024}.
\end{remark}
\begin{remark}
The \eqref{eq:robustLP} is efficiently solvable with the relaxation from \cref{sec:eff_cliff_heu} if $\rr \geq 2 \dd$ with $\dd$ the number of interaction terms that can be modified.
To approximate the evolution under the target Hamiltonian $H_T$ with $\nTro$ Trotter cycles we have to implement at most $\nTro \rr$ evolution blocks $U(t\lambda_{\vec c})$, even if $\lambda_{\vec c} = 0$, since the finite pulse time errors for all single-qubit pulse layers are taken into account in \eqref{eq:robustLP}.
Each evolution block contains $2 \ns$ single-qubit pulse layers (possibly being part of a robust composite pulse sequence).
Then, the total number of single-qubit pulse layers is at most $2 \ns \nTro \rr$.
The number of single-qubit pulse layers can be reduced from $2 \ns \nTro \rr$ to $\approx 2 \ns \nTro \dd$ by formulating a \acf{MILP}, which we explain in \cref{sec:MILP}.

To summarize, our efficient relaxation plays a central role for solving the \ac{LP} formulation Hamiltonian engineering problems.
\end{remark}

\subsection{Robust Clifford conjugation method}
\label{sec:robust_cliff}
In this section, we leverage the general robust conjugation method to robustify the Clifford conjugation method.
We have two single-qubit pulse layers $\ns = 2$ of $\pi/2$ pulses, splitting each Pauli $\pi$ pulse in the gate set $\mathcal{C}_{XY}$ into two $\pi/2$ pulses.
Therefore, the ideal evolution in \eqref{eq:robustLP} is captured by the matrix $\WGG{\dd}{\rr} = \WC{\dd}{\rr}$.
\begin{definition}[Single-qubit Clifford pulse block]
\label{def:cliff_pulse_block}
 The \emph{single-qubit Clifford pulse block} $\vec S_{\vec c}$ is specified by the tuple $\vec c = ((t_p^{(1)}, \vec s^{(1)}, \vec h^{(1)}), (t_p^{(2)}, \vec s^{(2)}, \vec h^{(2)}))$, with the finite pulse times of one $\pi/2$ pulse, and we have $t_p^{(1)} = t_p^{(2)}$.
The rotation direction $\vec s^{(\kk)}$ is fixed, and we set
$s^{(\kk)}_i = 1$ if on the $i$-th qubit and the $\kk$ layer we have $\sqrt{(\argdot)}^\dagger$ and $s^{(\kk)}_i = 0$ if we have $\sqrt{(\argdot)}$.
The generators are $h_i = \frac{\pi}{4} P^{(\kk)}_i$, with the Pauli operators from the square-root Pauli gates $P^{(\kk)}_i$.
\end{definition}

The Clifford method can be made robust against the finite pulse time error by direct application of \cref{thrm:robust_finite_pulse_time}.
Our results for the efficient relaxation of the Clifford conjugation method in \cref{sec:eff_cliff_heu} ensures that \eqref{eq:robustLP} is feasible and that all interaction terms with the same locality as the interactions in $H_S$ can be modified.
\begin{corollary}[Robustness against finite pulse time errors]
\label{cor:finite_pulse_time_cliff}
Let there be two layers of $\pi/2$ pulses, representing the single-qubit pulses for the Clifford conjugation as in \cref{def:cliff_pulse_block}.
Then, the finite pulse time effect can be suppressed using \cref{thrm:robust_finite_pulse_time}.
\end{corollary}

The Clifford conjugation method can also be made robust by taking advantage of the rich field of robust composite pulses.
The $\pi/2$ pulses in each of the two layers in the Clifford conjugation method can be made robust by replacing each pulse with robust composite pulses of length $\ns/2$.
Then, the finite pulse time error is different but can still be corrected using \cref{thrm:robust_finite_pulse_time}.
\begin{corollary}[Robustness against pulse errors]
\label{cor:finite_pulse_time_gate_cliff}
Let there be a robust composite pulse sequence for $\pi/2$ pulses of length $\ns/2$. 
Replacing each $\pi/2$ pulse in the Clifford conjugation method by such robust composite pulses yields $\ns$ single-qubit pulse layers.
Then, the finite pulse time effect is different from \cref{cor:finite_pulse_time_cliff} but still can be suppressed by \cref{thrm:robust_finite_pulse_time}.
\end{corollary}
The ability to modify any interaction with the same locality as $H_S$ and the combination with robust composite pulses makes our Clifford conjugation robust against a wide range of different errors in experiments.
Note, that composite pulses with a short overall duration are more beneficial due to the faster convergence of the Magnus expansion \cref{eq:magnus_approx_error}.
In \cref{sec:ion_trap_model} we combined the SCROFULOUS pulses \cite{Cummins2002} and the SCROBUTUS pulses \cite{kukita2021} with the robust Clifford conjugation to implement arbitrary Heisenberg Hamiltonians.

\subsection{Robust Pauli conjugation method}
\label{sec:robust_pauli}
The Pauli conjugation can only change the non-zero interaction strengths in the system Hamiltonian but not modify the type of interactions.
Consequently, the general robust conjugation method is not directly applicable.
Therefore, we generalize the robustness conditions of \citet{Votto23UniversalQuantumProcessors} to arbitrary local Hamiltonians.
\begin{definition}[Single-qubit Pauli pulse layer]
\label{def:pauli_pulse_block}
 The \emph{single-qubit Pauli pulse layer} $\vec S_{\vec c}$ is specified by the tuple $\vec c = (t_p , \vec s , \vec h)$, with the finite pulse time $t_p$ of one $\pi$ pulse.
The rotation direction is $\vec s \in \FF_2^n$ and can be chosen freely for $\pi$ pulses.
The generators are $h_i = \frac{\pi}{2} P_i$, with the Pauli operator $P_i$ given by $P_{\vec b_i}$ from \cref{eq:pauli_conjugation_system}.
\end{definition}
For the Pauli conjugation $U(t\lambda_{\vec c})$ simplifies to \cref{eq:evo_block_finite_pulse_time_sim}.
Next, we provide the average Hamiltonian for the conjugation of the system Hamiltonian with a single-qubit Pauli pulse layer to investigate the effect of the finite pulse time.
\begin{restatable}{restatableLem}{LemHerr}
\label{lem:pauli_finite_pulse_error}
 Consider all labels $\vec a \in \FF_2^{2n}$ for the non-zero interactions $J_{\vec a} \neq 0$ in the system Hamiltonian.
 Then, the approximation of $U(t\lambda_{\vec c})$ for the Pauli conjugation in first order Magnus expansion is given by $\e^{-\i H_{\av, \vec c} (t)}$ with
 \begin{equation}
 H_{\av, \vec c} (t) = t \lambda_{\vec c} \vec S_{\vec c}^{\dagger} H_S \vec S_{\vec c} + H_{\err, \vec c} \, ,
\end{equation}
where
 \begin{equation}
 \label{eq:pauli_finite_pulse}
  H_{\err, \vec c} = \sum_{\vec a \in \FF_2^{2n}\setminus \{ \vec 0 \}} 
  	\left(J_{\vec a} \EP{\dd}{\rr}_{\vec a, \vec c} P_{\vec a} + R_{\vec a, \vec c} \right) \, .
 \end{equation}
 We call the first term in \cref{eq:pauli_finite_pulse} the \emph{interaction term}, with
 \begin{equation}
 \begin{aligned}
 &\EP{\dd}{\rr}_{\vec a, \vec c} \coloneqq \\
 &\frac{4 t_p}{\pi} \int_0^{\frac{\pi}{2}} \left( \prod_{i \in \supp (\vec a)} (\cos^2 (\theta) + (-1)^{\langle \vec a_i , \vec b_i \rangle} \sin^2 (\theta)) \right) \rmd \theta \, ,
 \end{aligned}
\end{equation}
 which we collect as entries of the matrix  $\EP{\dd}{\rr} \in \R^{\dd \times \rr}$.
 We call the second term $R_{\vec a, \vec c}$ the \emph{rest term}, and it is proportional to $(-1)^{\vec e_{\vec a} \cdot \vec s}$, with $\vec s \in \FF_2^n$ representing the chosen rotation direction of the $\pi$ pulses and $\vec e_{\vec a} \in \FF_2^{n}$ encodes the sign flips due to the finite pulse time error such that $e_{\vec a, i} = 0$ for all $i \notin \supp (\vec a)$.
\end{restatable}

The proof of \cref{lem:pauli_finite_pulse_error} can be found in \cref{app:robust_pauli}.
We define the \ac{LP} similar to \eqref{eq:robustLP}
\begin{equation}
\tag{robustPauliLP}
\label{eq:robustPauliLP}
\begin{aligned}
 \mathrm{minimize} &&& \vec 1^T \vec \lambda  \\
 \mathrm{subject\ to}    &&& \WP{\dd}{\rr} \vec \lambda + \EP{\dd}{\rr} \vec 1 = \vec M \,, \; \\
 &&& \vec \lambda \in \R^{\rr}_{\geq 0} \, ,
\end{aligned}
\end{equation}
with $\WP{\dd}{\rr} \in \R^{\dd \times \rr}$ and $\vec M \in \R^{\dd}$ from \cref{sec:pauli_conjugation} and $\EP{\dd}{\rr} \in \R^{\dd \times \rr}$ from \cref{lem:pauli_finite_pulse_error}.
Similar to \cref{thrm:robust_finite_pulse_time} we can cancel the effect of the interaction terms in \cref{eq:pauli_finite_pulse} by implementing $U(t\lambda_{\vec c})$, with $\lambda_{\vec c}$ from \eqref{eq:robustPauliLP}.
However, there still remains the rest terms $R_{\vec a, \vec c}$ in \cref{eq:pauli_finite_pulse} which can be cancelled by selecting the $\pi$ pulse direction appropriately.
To this end, we define the set of \emph{robust rotation directions} of $\pi$ pulses
\begin{equation}
\begin{aligned}
 \mathcal{S}_{\vec J} \coloneqq \left\{\vec s \in \FF_2^n \right. |& \sum_{\vec s} (-1)^{\vec e_{\vec a} \cdot \vec s} = 0, \, \forall \vec e_{\vec a} \in \FF_2^{n} \text{ with } e_{\vec a, i} = 0 \, \\
 & \left. \forall i \notin \supp (\vec a) \, \forall \vec a \in \FF_2^{2n} \text{ with }J_{\vec a} \neq 0 \right\} \, .
\end{aligned}
\end{equation}
This definition requires that the sum of all potential sign flips caused by the finite pulse time error of the non-zero interactions cancels over all rotation directions in $\mathcal{S}_{\vec J}$.

\begin{restatable}{restatableProp}{MainProp}
\label{prop:robust_pauli}
The target time evolution $\e^{- \i t H_T}$ can be approximated by a deterministic product formula implementing $U(t\lambda_{\vec c})$, with $\lambda_{\vec c}$ from \eqref{eq:robustPauliLP}, and choosing the robust rotation directions $\vec s \in \mathcal{S}_{\vec J}$ of the $\pi$ pulses.
The only approximation errors are given by the ones from the Magnus approximation (\cref{lem:robust_finite_pulse_time}) and the approximation error from the deterministic product formula.
\end{restatable}
The proof of \cref{prop:robust_pauli} can be found in \cref{app:robust_pauli}.

Note, that even if $\lambda_{\vec c} = 0$ the evolution block $U(t\lambda_{\vec c})$ still has to be implemented with zero free evolution time.
The number of evolution blocks can be reduced by formulating a \ac{MILP}, which we explain in \cref{sec:MILP}. 

The robustness against finite pulse time errors simultaneously implies robustness against the first order effects of rotation angle errors.
This means, that for a set of robust rotation directions $\vec s \in \mathcal{S}_{\vec J}$ implies the cancelation rotation angle errors in the first order Taylor approximation, see \cref{prop:rotation_err_robust} in the appendix.
Given a system Hamiltonian with only two-body interactions, i.e.\ only interactions $P_{\vec a}$ with $|\supp (\vec a)|=2$, a good choice for the robust rotation directions of the $\pi$ pulses $\vec s_{(j)} \in \FF_2^n$ for $j=1, \dots , \pp$ can be found by utilizing the orthogonality property of Walsh-Hadamard matrices.
\begin{restatable}{restatableProp}{PropPiRotationsChoice}
\label{prop:rot_dir_choice}
   Let  $\pp = 2^{\lceil \log_2 (n+1) \rceil} \leq 2n$ and let $\WP{\pp}{\pp}$ be the $\pp \times \pp$ dimensional Walsh-Hadamard matrix.
 Choose $n$ distinct columns from $\WP{\pp}{\pp}$ without the first column and define the resulting partial Walsh-Hadamard matrix as $\WP{\pp}{n}$.
 Let $(-1)^{\vec s_{(j)}}$ be the $j$-th row of $\WP{\pp}{n}$.
 Then, for any non-zero two-body interaction $J_{\vec a} \neq 0$ with $\abs{\supp (\vec a)}=2$ we have $\vec s_{(j)} \in \mathcal{S}_{\vec J}$ for all $j=1, \dots , \pp$.
\end{restatable}
The proof of \cref{prop:rot_dir_choice} can be found in \cref{app:robust_pauli}.

\subsection{Reducing the number of single-qubit pulses}
\label{sec:MILP}
Let the matrices $W, E \in \R^{\dd \times \rr}$ be either $\WGG{\dd}{\rr}, \EC{\dd}{\rr}$ or $\WP{\dd}{\rr}, \EP{\dd}{\rr}$ as in \eqref{eq:robustLP} or \eqref{eq:robustPauliLP} respectively.
In these \acp{LP} the sum of the finite pulse time effects for all considered evolution blocks $U(t\lambda_{\vec c})$ is given by the vector $E \cdot \vec 1 \in \R^{\dd}$. 
Moreover, the solution $\vec \lambda$ in the \eqref{eq:robustLP} and \eqref{eq:robustPauliLP} is sparse \cite{BARANY1982}.
However, when mitigating the finite pulse time errors as in \cref{thrm:robust_finite_pulse_time,prop:robust_pauli} then all the single-qubit pulse conjugations have to be implemented (with zero free evolution time if $\lambda_{\vec c}=0$).
The \acp{LP} can be extended with additional binary variables $\vec z \in \{0,1 \}^s$ and an additional constraint to only consider the finite pulse time errors for the single-qubit pulses with non-zero free evolution time.
This yields the \acf{MILP} \cite{karahanoglu_mixed_2013}
\begin{equation}
\tag{MILP}
\label{eq:MILP}
\begin{aligned}
 \mathrm{minimize} &&& \alpha \vec 1^T \vec \lambda + (1-\alpha) \vec 1^T \vec z \\
 \mathrm{subject\ to}    &&& W \vec \lambda + E \vec z = \vec M \,, \\
 &&& c_l \vec z \leq \vec \lambda \leq c_u \vec z \, , \\
 &&& \vec \lambda \in \R^{\rr}_{\geq 0}\, , \quad \vec z \in \{0,1\}^{\rr}  \, .
\end{aligned}
\end{equation}
The free parameter $\alpha \in [0,1]$ assigns weights to the minimization of the free evolution times $\alpha=1$ or the minimization of the number of single-qubit pulse layers $\alpha=0$.
$0 \leq c_l < c_u$ are lower and upper bounds on the entries of $ \vec \lambda$.
The interval $[c_l, c_u]$ has to be large enough such that \eqref{eq:MILP} has a solution.

For our robust method \eqref{eq:robustLP}, this \eqref{eq:MILP} reduces the number of evolution blocks in one Trotter cycle from $\rr \geq 2 \dd$ to $\rr \approx \dd$.
Although \eqref{eq:MILP} is in general hard to solve, there are many powerful heuristics and software packages to solve such optimizations \cite{mosek}.
Moreover, the size of \eqref{eq:MILP} can be drastically reduced with to our efficient relaxation, and solving it for small instances is still feasible.
The \eqref{eq:MILP} is feasible only for small instances whereas the \eqref{eq:robustLP} is efficiently solvable at the cost of more single-qubit pulse layers.

If not otherwise stated, we used the parameters $c_l = 10^{-6}$, $c_u = 10^{3}$ and $\alpha = 10^{-2}$.
We used MOSEK to solve \eqref{eq:MILP} with parameter {\tt MSK\_DPAR\_MIO\_TOL\_REL\_GAP} set to $1.0$ \cite{mosek}.

\section{Discussion and outlook}
\label{sec:outlook}
We consider a quantum computing or quantum simulation platform that has one entangling Hamiltonian as \emph{system Hamiltonian} and provide efficient, general and robust methods to engineer arbitrary Hamiltonians with the same locality from it.
Our methods only rely on the use of $\pi$ or $\pi/2$ pulses, i.e.\ Pauli or single-qubit Clifford gates, and explicitly allows for robust composite pulses.
They can be directly used in experiments by applying the pulse sequences generated by the provided Python code \cite{GitHub}.

The pulse sequences are obtained by solving a suitable \acf{LP}, the classical runtime of which depends polynomially only on the number of interaction terms that can be generated from the system Hamiltonian and can thus be efficiently solved for local Hamiltonians.
The Pauli conjugation method is even applicable if only partial knowledge of the system Hamiltonian is available, and can be used to cancel unwanted, but unknown interaction terms.
Moreover, the quantum simulation runtime can be reduced at the cost of a higher classical runtime, which provides a flexible trade-off.

Another major advantage of our methods is the robustness against experimental imperfections.
Simulation errors introduced by finite pulse times can be explicitly compensated in the computation of the pulse sequence.
The Clifford conjugation method can be combined with robust composite pulses, making it robust against major experimental errors.
We discuss in detail the effect of finite pulse time errors and rotation angle errors and show that these can be fully mitigated by modified pulse sequences in combination with robust composite pulses.
Due to their generality and efficiency, we expect that our methods will find many applications in quantum computation and quantum simulation, such as the fast synthesis of multi-qubit gates, or analogue quantum simulation for problems arising in many-body physics.
Furthermore, some recent Hamiltonian learning schemes rely on ``reshaping'' an unknown Hamiltonian to a diagonal Hamiltonian which can be done efficiently and in a robust way with our Pauli conjugation method \cite{huang2023learning,ma2024learningkbodyhamiltonianscompressed,hu2025}.

In the future, we would like to extend the efficient Hamiltonian engineering method to fermionic systems for digital and analogue quantum simulations.
Another promising future research direction could be the investigation of finite pulse time effects for continuous robust pulses similar to \cref{sec:robust_general}.
Moreover, the design of robust pulses tailored to the conjugation methods might be another interesting direction.

\section*{Acknowledgements}
We are grateful to Ivan Boldin, Patrick Huber, Markus Nünnerich and Rodolfo Muñoz Rodriguez for a productive dialogue on the ion trap platform, and to Matthias Zipper and Christopher Cedzich for fruitful discussions on gate designs for ion trap platforms.
Moreover, we thank Matteo Votto for a constructive exchange about his impressive work and the robust sequences therein.
We also thank Gaurav Bhole for making us aware of his results.
Furthermore, we want to thank Özgün Kum for invaluable comments on our manuscript.

This work has been funded by the German Federal Ministry of Research, Technology and Space (BMFTR) within the funding program ``Quantum Technologies -- from Basic Research to Market'' via the joint project MIQRO (grant number 13N15522);
the QuantERA II Programme that has received funding from the EU’s H2020 research and innovation programme under Grant Agreement No.\ 101017733, and with the Deutsche Forschungsgemeinschaft (DFG, German Research Foundation) under the grant number 532779266; 
the Hamburg Quantum Computing project that is co-financed by the ERDF of the European Union and the Fonds of the Hamburg Ministry of Science, Research, Equalities and Districts (BWFGB); 
and the Fujitsu Services GmbH as part of the endowed professorship ``Quantum Inspired and Quantum Optimization''.
Publishing fees supported by Funding Programme Open Access Publishing of Hamburg University of Technology (TUHH).

\newpage
\onecolumngrid
\section*{Appendix}
\appendix

\section{Details for the numerical simuations}
\label{app:details_num_sim}
In this section, we provide the details for the numerical simulations in \cref{sec:2D_lattice_model,sec:ion_trap_model}.

\subsection{Simulation of a 2D lattice model}
\label{app:details_2D_lattice_model}
In \cref{sec:2D_lattice_model} we engineered a 2D lattice Hamiltonian with Ising interactions and unwanted but unknown three-body interactions.
In \cref{fig:2D_lattice}, we compare the \textsf{naive} Pauli conjugation from \cref{sec:pauli_conjugation} against the \textsf{robust Pauli} conjugation from \cref{sec:robust_pauli}.
To reduce the number of required pulses we solved the \eqref{eq:MILP} for the \eqref{eq:robustPauliLP} to obtain the relative evolution times $\vec \lambda$.
We use multiple $\pi$ pulse rotation directions as in \cref{prop:rot_dir_choice}, to cancel the rest term of the finite pulse time error as explained in \cref{lem:pauli_finite_pulse_error}.
Therefore, for $n=6$ qubits $\pp = 2^{\lceil \log_2 (n+1) \rceil} = 8$ different rotation directions and thus $8$ evolution blocks are required within each Trotter cycle.
To cancel the unwanted and unknown three-body interaction terms we apply the results from \cref{sec:unknown_hermitian_op} to both \textsf{naive} and \textsf{robust Pauli} approaches.

\subsection{Simulation of Heisenberg Hamiltonians with an ion trap model}
\label{app:details_ion_trap_model}
In \cref{sec:ion_trap_model} we engineered a random Heisenberg Hamiltonian from a system Hamiltonian with Ising interactions.
In \cref{fig:noisy_hamiltonian} we compare the \textsf{naive} Clifford conjugation from \cref{sec:clifford_conjugation} to the \textsf{robust Clifford}, \textsf{CP$_\mathrm{SCROFULOUS}$} and \textsf{CP$_\mathrm{SCROBUTUS}$} Clifford conjugations from \cref{sec:robust_cliff}.
The \textsf{robust Clifford} method is only robust against the finite pulse time error.
The \textsf{CP$_\mathrm{SCROFULOUS}$} method is additionally robust against rotation angle errors.
It is given as a combination of our robust Clifford conjugation with the SCROFULOUS composite pulses \cite{Cummins2002}.
Finally, the \textsf{CP$_\mathrm{SCROBUTUS}$} method is robust against the finite pulse time error, rotation angle errors and off-resonance errors.
It is given as a combination of our robust Clifford conjugation with the SCROBUTUS composite pulses \cite{kukita2021}.
We again solved the \eqref{eq:MILP} for the \eqref{eq:robustLP} to obtain the relative evolution times $\vec \lambda$ and reduce the number of required pulses for the \textsf{robust Clifford}, \textsf{CP$_\mathrm{SCROFULOUS}$} and \textsf{CP$_\mathrm{SCROBUTUS}$} Clifford conjugation.

\section{Properties of solutions of \texorpdfstring{\eqref{eq:PauliLP}}{(PauliLP)}}
\label{app:solution_exist_bounds}
In this section, we show the existence of solutions, provide lower and upper bounds on the relative evolution time, and discuss the tightness of these bounds.
First, a direct consequence of \cref{thrm:feasibility} is that \eqref{eq:PauliLP} is feasible or arbitrary Hamiltonians $H_S$ and $H_T$ satisfying $\operatorname{nz}(\vec A)\subseteq\operatorname{nz}(\vec J)$.

\begin{corollary}[existence of solutions]
\label{cor:full_W_feasibility}
 Let
 \begin{equation}
  H_S = \sum_{\vec a \in \FF_2^{2n}\setminus \{ \vec 0 \}} J_{\vec a} P_{\vec a} \quad \text{ and } \quad H_T = \sum_{\vec a \in \FF_2^{2n}\setminus \{ \vec 0 \}} A_{\vec a} P_{\vec a} \, ,
 \end{equation}
 with $\operatorname{nz}(\vec A)\subseteq\operatorname{nz}(\vec J)$.
 Then the partial Walsh-Hadamard matrix $\WP{\dd}{4^n}$ with $\dd= \abs{\operatorname{nz}(\vec J)}$ leads to a feasible \eqref{eq:PauliLP}.
\end{corollary}
\begin{proof}
We know that all rows of the Walsh-Hadamard matrix are linearly independent.
Furthermore, each row of the Walsh-Hadamard matrix sums to zero, thus $\vec x = \vec 1 > \vec 0$ is a solution to $\WP{\dd}{4^n}\vec x = \vec 0$.
By \cref{thrm:feasibility} we know that \eqref{eq:PauliLP} always has a feasible solution even if we consider an arbitrary subset of rows.
\end{proof}
To bound the optimal solutions $\vec 1^T \vec \lambda^*$ of \eqref{eq:PauliLP} we define the dual \ac{LP}
\begin{maxi}
{}{\vec M^T \vec y}{}{}
\label{eq:dualLP}
\addConstraint{(\WP{\dd}{4^n})^T \vec y } {\leq \vec 1 \, ,}{\;}{\vec y \in \R^{\dd}} \, .
\end{maxi}
As in the previously considered case of two-body Ising interactions \cite{basler_time-optimal_2024}, we have the following bounds on the value of \eqref{eq:PauliLP}:
\begin{theorem}[bounds on solutions]\label{thm:gate-time}
 The optimal objective function value of \eqref{eq:PauliLP} with a partial Walsh-Hadamard matrix $\WP{\dd}{4^n}$ is bounded by
 \begin{equation}
  \linfnorm{ \vec M } \leq \vec 1^T \vec \lambda^* \leq \lonenorm{ \vec M} \, .
 \end{equation}
\end{theorem}
\begin{proof}
The lower bound can be verified by the fact that $\WP{\dd}{4^n}$ in \eqref{eq:PauliLP} only has entries $\pm 1$ and that $\vec \lambda^*$ is non-negative.
Thus, it holds that $\linfnorm{ \vec M } = \linfnorm{ \WP{\dd}{4^n} \vec \lambda^* } \leq \vec 1^T \vec \lambda^*$.

To show the upper bound, we first define the set of feasible solutions for the dual \ac{LP}~\eqref{eq:dualLP}
\begin{equation}
 \mathcal{F} \coloneqq \Set*{ \vec y \in \R^\dd \given (\WP{\dd}{4^n})^T \vec y \leq \vec 1} \, .
\end{equation}
Next, consider the partial Walsh-Hadamard matrix $\WP{(4^n-1)}{4^n}$ without the first row, corresponding to the identity Pauli term $P_{\vec a} = I^{\otimes n}$ with $\vec a = \vec 0$.
We show that
\begin{equation}
 \mathcal{S} \coloneqq \Set*{ \vec y \in \R^{4^n-1} \given (\WP{(4^n-1)}{4^n})^T \vec y \leq \vec 1}
\end{equation}
is a simplex.
A simplex is formed by affine independent vectors.
Vectors of the form $(1, \vec v_i)^T$ are linearly dependent if and only if the vectors $\vec v_i$ are affine dependent.
From linear dependence it follows that there is a vector $\vec t \neq \vec 0$ such that $\sum_i t_i (1, \vec v_i)^T = \vec 0$.
Thus $\sum_i t_i = 0$ and $\sum_i \vec v_i = \vec 0$, and the vectors $\vec v_i$ are affine dependent.
For any $n$, the Walsh-Hadamard matrix $W=(\vec 1 , (\WP{(4^n-1)}{4^n})^T)$ has linearly independent rows and columns.
Therefore, $\WP{(4^n-1)}{4^n}$ has affine independent columns, and $\mathcal{S}$ forms a simplex.
With $\mathcal{H} = \{\vec y \in \R^{4^{n}-1} | \linfnorm{ \vec y } \leq 1 \}$ we denote the $4^{n}-1$ dimensional hypercube.
Note that the rows of $\WP{\dd}{4^n}$ are rows of the $4^{n} \times 4^{n}$ dimensional Walsh-Hadamard matrix.
We define the embedding $g: \R^\dd \rightarrow \R^{(4^{n})-1}$ by appending zeros, such that $(\WP{\dd}{4^n})^T \vec y = (\WP{(4^n-1)}{4^n})^T g(\vec y)$.
Therefore the objective value of the dual \ac{LP}~\eqref{eq:dualLP} can be upper bounded as follows
\begin{equation}
\begin{aligned}
 \max_{\vec y \in \mathcal{F}} \innerp{\vec M}{\vec y} &= \max_{g(\vec y) \in \mathcal{S}} \innerp{g(\vec M)}{g(\vec y)} \\
 &\leq \max_{\vec x \in \mathcal{S}} \innerp{g(\vec M)}{\vec x} \\
 &\leq \max_{\vec x \in \mathcal{H}} \innerp{g(\vec M)}{\vec x} = \lonenorm{ \vec M } \, .
\end{aligned}
\end{equation}
The upper bound for $\vec 1^T \vec \lambda^*$ follows by strong duality.
\end{proof}
Next, we show a (not complete) set of instances of \eqref{eq:PauliLP}, yielding solutions $\vec \lambda^*$ which satisfy the upper bound of \cref{thm:gate-time}.
Such $\vec M$ constitute the worst cases.
\begin{proposition}
\label{prop:worst_case}
  If $\vec M \in \Set*{ - \vec w_i \given \vec w_i \text{ is the i-th column of }\WP{(4^n-1)}{4^n}}$, then $\vec 1^T \vec \lambda^* = \lonenorm{ \vec M }$.
\end{proposition}
\begin{proof}
Let $\vec y = - \vec w_i$ be a negative column of $\WP{(4^n-1)}{4^n}$.
Then, by orthogonality of the rows/columns of the Walsh-Hadamard matrix, we obtain
 \begin{equation}
  \left((\WP{(4^n-1)}{4^n})^T (- \vec w_i) \right)_j = \begin{cases}
		-(4^n-1) \, , & i=j \\
		1 \, , & \text{otherwise}\, .
	\end{cases}
 \end{equation}
 Therefore, $\vec y = - \vec w_i$ is a feasible solution $(\WP{(4^n-1)}{4^n})^T \vec y \leq \vec 1$ to the dual \ac{LP}~\eqref{eq:dualLP}.
The dual objective function value is $\vec M^T \vec y = \lonenorm{\vec w_i} = 4^n-1$.
Next, we show a feasible solution of \eqref{eq:PauliLP}.
For a $\vec M = - \vec w_i$ define
\begin{equation}
   \lambda_j = \begin{cases}
		0 \, , & i=j \\
		1 \, , & \text{otherwise}\, .
	\end{cases}
\end{equation}
Clearly, this satisfies $\WP{(4^n-1)}{4^n} \vec \lambda = - \vec w_i$.
Furthermore, the primal objective function value $\vec 1^T \vec \lambda = 4^n-1$ is the same as for the dual objective function value, which shows optimality.
It is easy to see that $\lonenorm{ \vec M } = 4^n-1$.
\end{proof}

\section{Comments on the efficient relaxation for the Pauli conjugation}
\label{app:eff_pauli_heu}
As mentioned in the main text, Wendel's theorem is applicable to spherical symmetric distributions.
This would imply that sampling a certain column from $\WP{\dd}{4^n}$ has the same probability as sampling the same column with a flipped sign.
Recall that for a partial Walsh-Hadamard matrix $\WP{\dd}{4^n}_{\vec a \vec b} = (-1)^{\langle \vec a,\vec b \rangle}$, with $\langle \vec a,\vec b \rangle = \vec a_x \cdot \vec b_z + \vec a_z \cdot \vec b_x$.
It holds that
\begin{equation}
 \WP{\dd}{4^n}_{\vec a \vec b} = - \WP{\dd}{4^n}_{\vec a \bar{\vec b}} \quad \forall \vec a, \vec b \in \FF_2^{2n} ,
\end{equation}
for $\abs{a} = 1 \operatorname{mod}2$ and with the binary complement $\bar{\vec b}$ is given element-wise given by $\bar x\coloneqq 1-x$ for any $x\in \Set{0,1}$.
If the decomposition of $H_S$ has only terms $J_{\vec a} P_{\vec a}$ with odd $\abs{a}$, then
we can apply Wendel's theorem directly.
In this case, we have a success probability of $1/2$ to find a feasible $\WP{\dd}{2 \dd}$ (with $\dd$ interactions) if we sample $2\dd$ many $\vec b \in \FF_2^{2n}$ uniformly random.
However, an example of such a Hamiltonian with two-body interactions is
\begin{equation}
 H = \sum_{i=1}^n \left(J_i^X X_i + J_i^Y Y_i\right) + \sum_{i \neq j}^n \left(J_{ij}^{XZ} X_i Z_j + J_{ij}^{YZ} Y_i Z_j \right) ,
\end{equation}
and with commuting interactions is
\begin{equation}
 H = \sum_{i=1}^n J_i^X X_i + \sum_{ijk}^n J_{ijk}^{XXX} X_i X_j X_k \, ,
\end{equation}
with arbitrary coupling constants.
Unfortunately, for a general Hamiltonian, we cannot use Wendel's theorem to construct a relaxation for \eqref{eq:PauliLP}.

Recently, lower bounds on $p_{\rr,\vec x}$ have been proposed for arbitrary distributions \cite{Hayakawa2023}.
However, their results rely on the halfspace depth (or Tukey depth), which is hard to compute.
It is a measure of how extreme a point is with respect to a distribution of random points.
\begin{definition}[halfspace depth]
 Let $\vec x$ be an arbitrary $\dd$-dimensional random vector.
 Then, the halfspace depth at the origin is defined as
 \begin{equation}
  \alpha_{\vec x} \coloneqq \inf_{\| \vec c \|=1} \PP (\vec c^T \vec x \leq \vec 0) \, .
 \end{equation}
\end{definition}
The halfspace depth $\alpha_{\vec x}$ is the minimum (fraction) number of points in a halfspace with the origin on the boundary.
\begin{theorem}[{\cite[Theorem~14]{Hayakawa2023}}]
 Let $\vec x$ be an arbitrary $\dd$-dimensional random vector.
 Then, for each positive integer $\rr \geq 3\dd/\alpha_{\vec x}$, we have
 \begin{equation}
  p_{\rr,\vec x} > 1-\frac{1}{2^\dd} \, .
 \end{equation}
\end{theorem}
Let $\vec w$ be a $\dd$-dimensional random vector drawn uniformly from $\operatorname{col} (\WP{\dd}{4^n})$.
From \cref{cor:full_W_feasibility}, we find the trivial lower bound $1/4^n \leq \alpha_{\vec w}$, since at least one point is in an arbitrary halfspace with the origin on its boundary.
Finding a constant lower bound $1/\beta \leq \alpha_{\vec w}$ would imply that for $\rr \geq 3\dd \beta$ we find a feasible $\WP{\dd}{\rr}$ with high probability.
One needs to show that at least $4^{\LandauO (n)}$ points are contained in an arbitrary halfspace with the origin on its boundary.

\section{Proofs for the general robust conjugation method}
\label{app:robust_general}
We provide the proofs for \cref{sec:robust_general}.
Some proofs provide more technical details which aims for an easy implementation into a programming language.
For the sake of easy readability we repeat the definitions of the main text.

\DefSingleGateLayer*

We consider the time evolution
\begin{equation}
\label{eq_app:evo_block_finite_pulse_time}
 U(t\lambda_{\vec c}) =  \left( \prodl_{\kk=1}^{\ns} \e^{- \i t_p^{(\kk)} (H_S - H_{c^{(\kk)}})}  \right) \e^{-\i t \lambda_{\vec c} H_S} \left( \prodr_{\kk=1}^{\ns} \e^{- \i t_p^{(\kk)} (H_S + H_{c^{(\kk)}})}  \right) \, ,
\end{equation}
with the finite duration of the $\kk$-th pulse layer $0<t_p^{(\kk)}$.

\LemHav*

\begin{proof}
 We consider the time evolution during the $\kk$-th layer of single-qubit pulses in the interaction frame with respect to the single-qubit pulse Hamiltonian $H_{c^{(\kk)}}$ \cite{Haeberlen1968,Brinkmann2016}
\begin{equation}
\label{eq_app:int_frame_prop}
 \e^{-\i t_p^{(\kk)} (H_S \pm H_{c^{(\kk)}})} = \e^{\mp \i t_p^{(\kk)} H_{c^{(\kk)}}} \mathcal{U}^{(\kk)}_{\pm} (t_p^{(\kk)})
 = \vec S_{c^{(\kk)}}^{\pm 1} \mathcal{U}^{(\kk)}_{\pm} (t_p^{(\kk)}) \, ,
\end{equation}
with the interaction frame propagator $\mathcal{U}^{(\kk)}_{\pm} (t_p^{(\kk)})$.
The interaction frame propagator has to fulfill
\begin{equation}
 \frac{d \mathcal{U}^{(\kk)}_{\pm} (t)}{dt} =  - \i \left(\vec S_{c^{(\kk)}}^{\mp 1}(t) H_S \vec S_{c^{(\kk)}}^{\pm 1}(t) \right) \mathcal{U}^{(\kk)}_{\pm} (t) \, .
\end{equation}
Inserting \cref{eq_app:int_frame_prop} into \cref{eq_app:evo_block_finite_pulse_time} yields
\begin{equation}
 U(t\lambda_{\vec c}) = \left( \prodl_{\kk=1}^{\ns} \vec S_{c^{(\kk)}}^{- 1} \mathcal{U}^{(\kk)}_{-} (t_p^{(\kk)}) \right) \e^{-\i t \lambda_{\vec c} H_S} \left( \prodr_{\kk=1}^{\ns} \vec S_{c^{(\kk)}} \mathcal{U}^{(\kk)}_{+} (t_p^{(\kk)}) \right) \, .
\end{equation}
We define $\tilde{\mathcal{U}}^{(\kk)}_{-} (t)$ as
\begin{equation}
 \frac{d \tilde{\mathcal{U}}^{(\kk)}_{-} (t)}{dt} \coloneqq \vec S_{\vec c^{\geq \kk}}^{-1} \frac{d \mathcal{U}{(\kk)}_{-} (t)}{dt} \vec S_{\vec c^{\geq \kk}}
 =  - \i \left( \vec S_{\vec c^{\geq \kk}}^{-1} \vec S_{c^{(\kk)}}(t) H_S \vec S_{c^{(\kk)}}^{-1} (t) \vec S_{\vec c^{\geq \kk}} \right) \tilde{\mathcal{U}}^{(\kk)}_{-} (t) \, ,
\end{equation}
and $\tilde{\mathcal{U}}^{(\kk)}_{+} (t)$ as
\begin{equation}
 \frac{d \tilde{\mathcal{U}}^{(\kk)}_{+} (t)}{dt} \coloneqq \vec S_{\vec c^{\geq (\kk + 1)}}^{-1} \frac{d \mathcal{U}{(\kk)}_{+} (t)}{dt} \vec S_{\vec c^{\geq (\kk + 1)}}
 =  - \i \left( \vec S_{\vec c^{\geq (\kk + 1)}}^{-1} \vec S_{c^{(\kk)}}^{-1}(t) H_S \vec S_{c^{(\kk)}}(t) \vec S_{\vec c^{\geq (\kk + 1)}} \right) \tilde{\mathcal{U}}^{(\kk)}_{+} (t) \, ;
\end{equation}
where
$\tilde{\mathcal{U}}^{(\kk)}_{-} (t)$ and $\tilde{\mathcal{U}}^{(\kk)}_{+} (t)$ are defined such that we can commute the exact evolution of the single-qubit pulse layers to the free evolution under $H_S$.
It directly follows that
\begin{equation}
 U(t\lambda_{\vec c}) = \left( \prodl_{\kk=1}^{\ns} \tilde{\mathcal{U}}^{(\kk)}_{-} (t) \right) \e^{-\i t \lambda_{\vec c} \vec S_{\vec c}^{-1} H_S \vec S_{\vec c}} \left( \prodr_{\kk=1}^{\ns} \tilde{\mathcal{U}}^{(\kk)}_{+} (t) \right) \, .
\end{equation}
The Hamiltonian governing the evolution of $U(t\lambda_{\vec c})$ is defined piecewise as
\begin{equation}
H_{\vec c}(\tilde{t}) \coloneqq \begin{cases}
		\vec S_{\vec c^{\geq (\kk + 1)}}^{-1} \vec S_{c^{(\kk)}}^{-1} (\tilde{t}) H_S \vec S_{c^{(\kk)}}(\tilde{t}) \vec S_{\vec c^{\geq (\kk + 1)}} \, , &  T_p^{\leq (\kk-1)} \leq \tilde{t} < T_p^{\leq \kk} \\
		\vec S_{\vec c}^{-1} H_S \vec S_{\vec c} \, , & T_p \leq \tilde{t} < T_p + t \lambda_{\vec c} \\
		\vec S_{\vec c^{\geq \kk}}^{-1} \vec S_{c^{(\kk)}}(\tilde{t}) H_S \vec S_{c^{(\kk)}}^{-1} (\tilde{t}) \vec S_{\vec c^{\geq \kk}} \, , & T_p + T_p^{\leq (\kk-1)} + t \lambda_{\vec c} \leq \tilde{t} < T_p + T_p^{\leq \kk} + t \lambda_{\vec c} \, .
	\end{cases}
\end{equation}
The effective Hamiltonian is approximated by the first-order term in the Magnus expansion up to time $T = 2T_p + t \lambda_{\vec c}$, yielding the time independent average Hamiltonian
\begin{equation}
\begin{aligned}
 H_{\av,\vec c}(t) &= \int_0^{T} H_{\vec c}(\tilde{t}) \rmd \tilde{t} \\
 &=  t \lambda_{\vec c} \vec S_{\vec c}^{-1} H_S \vec S_{\vec c} + \sum_{\kk=1}^{\ns} \vec S_{\vec c^{\geq (\kk + 1)}}^{-1} \int_0^{t_p^{(\kk)}} \vec S_{c^{(\kk)}}^{-1} (t) H_S \vec S_{c^{(\kk)}}(t) \rmd t \vec S_{\vec c^{\geq (\kk + 1)}} + \vec S_{\vec c^{\geq \kk}}^{-1} \int_0^{t_p^{(\kk)}} \vec S_{c^{(\kk)}}(t) H_S \vec S_{c^{(\kk)}}^{-1} (t) \rmd t \vec S_{\vec c^{\geq \kk}}\\
 &=t \lambda_{\vec c} \vec S_{\vec c}^{-1} H_S \vec S_{\vec c} + 2 \sum_{\kk=1}^{\ns} \vec S_{\vec c^{\geq (\kk + 1)}}^{-1} \int_0^{t_p^{(\kk)}} \vec S_{c^{(\kk)}}^{-1} (t) H_S \vec S_{c^{(\kk)}}(t) \rmd t \vec S_{\vec c^{\geq (\kk + 1)}} \, .
\end{aligned}
\end{equation}
We denote the term corresponding to the finite pulse time error by
\begin{equation}
 H_{\err, \vec c} \coloneqq 2 \sum_{\kk=1}^{\ns} \vec S_{\vec c^{\geq (\kk + 1)}}^{-1} \int_0^{t_p^{(\kk)}} \vec S_{c^{(\kk)}}^{-1} (t) H_S \vec S_{c^{(\kk)}} (t) \rmd t \, \vec S_{\vec c^{\geq (\kk + 1)}} \, .
\end{equation}
Then, the average Hamiltonian is
\begin{equation}
 H_{\av, \vec c} (t) = t \lambda_{\vec c} \vec S_{\vec c}^{-1} H_S \vec S_{\vec c} + H_{\err, \vec c} \, .
\end{equation}
Note, that conjugation of the system Hamiltonian $H_S$ with single-qubit operations as in $\vec S_{\vec c}^{-1} H_S \vec S_{\vec c}$ and $H_{\err, \vec c}$ always preserves the locality of $H_S$.

Finally, we prove the error bounds on the truncation of the Magnus expansion after the first order.
The $k$-th order term of the Magnus expansion can be written as
\begin{equation}
 H_{\av}^{(k)} = \sum_{\sigma \in S_k} (-1)^{d_b} \frac{d_a! d_b!}{k!} \int_0^{T} \rmd \tau_1 \int_0^{\tau_1} \rmd \tau_2 \dots \int_0^{\tau_{k-1}} \rmd \tau_k H_{\vec c}(\tau_{\sigma(1)}) H_{\vec c}(\tau_{\sigma(2)}) \dots H_{\vec c}(\tau_{\sigma(k)}) \, ,
\end{equation}
where $T = 2T_p + t \lambda_{\vec c}$, $S_k$ denotes the group of permutations $\sigma$ of the set $[k]$ the number of ascents $d_a$ and descents $d_b$ are defined as
\begin{equation}
 d_a \coloneqq |\Set*{i \in [k-1] \given \sigma(i)<\sigma(i+1)}|\quad \text{ and } \quad
 d_b \coloneqq |\Set*{i \in [k-1] \given \sigma(i)>\sigma(i+1)}| \, ,
\end{equation}
and it holds that $d_a+d_b = k$ \cite{Arnal2017}.
Finally, we argue that
\begin{equation}
 \| H_{\av}^{(k)} \| \leq \LandauO \left( (2T_p + t \lambda_{\vec c})^k \| \max_{\tau \in [0,2T_p + t \lambda_{\vec c}]} H_{\vec c}(\tau) \|^k \right) = \LandauO ( (2T_p + t \lambda_{\vec c})^k \| H_S \|^k) \, ,
\end{equation}
where the equality follows from the definition of $H_{\vec c}(\tau) = U^{-1} H_S U$ for some unitaries $U$ and the invariance of the spectral norm under unitary transformation.
The error bound in spectral norm follows from the Duhamel principle in Ref.~\cite[{App.~A}]{Sharma2024}
\begin{equation}
 \| U(t \lambda_{\vec c}) - \e^{-\i H_{\av, \vec c} (t)} \| \leq \sum_{k=2}^\infty \| H_{\av}^{(k)} \| = \LandauO((2T_p + t \lambda_{\vec c})^2 \| H_S \|^2) \, ,
\end{equation}
assuming that $\| H_{\av}^{(2)} \|> \sum_{k=3}^\infty\| H_{\av}^{(k)} \|$, i.e.\ the Magnus expansion converges.
\end{proof}

\begin{restatable}{restatableLem}{LemWmat}
\label{lem:robust_LP}
 Let $\WGG{\dd}{\rr}, \EC{\dd}{\rr} \in \R^{\dd \times \rr}$ be the matrices representing the ideal conjugation and the finite pulse time error in the Pauli basis respectively.
 The entries are given by
\begin{equation}
 \WGG{\dd}{\rr}_{\vec a \vec c} \coloneqq \frac{1}{2^n} \Tr \left(P_{\vec a} \left(\vec S_{\vec c}^{\dagger} H_S \vec S_{\vec c} \right) \right) \quad \text{and} \quad \EC{\dd}{\rr}_{\vec a \vec c} \coloneqq \frac{1}{2^n} \Tr \left(P_{\vec a} H_{\err, \vec c} \right) \, ,
\end{equation}
and can be calculated in polynomial time in the number of qubits for a local system Hamiltonian $H_S$.
\end{restatable}
\begin{proof}
Recall, that the system Hamiltonian has the form
\begin{equation}
H_S = \sum_{\vec a \in \FF_2^{2n}\setminus \{ \vec 0 \}} J_{\vec a} P_{\vec a} \, .
\end{equation}
Let $H_{c^{(\kk)}} = \frac{\theta^{(\kk)}}{t_p^{(\kk)}} \sum_{i=1}^n (-1)^{s_i^{(\kk)}} h_i^{(\kk)}$ be an arbitrary layer of single-qubit rotations, with the rotation angle $\theta^{(\kk)}$, the pulse duration $t_p^{(\kk)}$ the generators $h_i^{(\kk)} \coloneqq p_{x,i}^{(\kk)} X + p_{y,i}^{(\kk)} Y + p_{z,i}^{(\kk)} Z$ and the rotation direction $s_i^{(\kk)} \in \FF_2$.
For the sake of a clear notation we omit the layer index $(\kk)$ for now.
The single-qubit pulse on the $i$-th qubit can be written as
\begin{equation}
 S_{c_i}(t) \coloneqq \e^{-\i t \frac{\theta}{t_p} h_i} = \cos(t \frac{\theta}{t_p}) I - \i (-1)^{s_i} \sin(t \frac{\theta}{t_p}) h_i \, .
\end{equation}
Then, $\vec S_{c} (t) = \e^{- \i t H_{c}} = \bigotimes_{i=1}^n S_{c_i}(t)$.
The conjugation with a single-qubit pulse layer changes the interaction term as
\begin{equation}
\label{eq:L_conjugation}
 \vec S_{c}^{-1} (t) J_{\vec a} P_{\vec a} \vec S_{c} (t) = J_{\vec a} \bigotimes_{i=1}^n S_{c_i}^{-1} (t) P_{\vec a_i} S_{c_i}(t) \, .
\end{equation}
The effect of such a conjugation on the $i$-th qubit is given by
\begin{equation}
\label{eq:L_conjugation_qubit}
 S_{c_i}^{-1} (t) P_{\vec a_i} S_{c_i}(t) = \cos^2(\tilde{\theta}) P_{\vec a_i} + \sin^2(\tilde{\theta}) h_i P_{\vec a_i} h_i + \i [h_i , P_{\vec a_i}] (-1)^{s_i} \sin (\tilde{\theta}) \cos (\tilde{\theta}) \, ,
\end{equation}
with $\tilde{\theta} \coloneqq t \frac{\theta}{t_p}$.
\Cref{eq:L_conjugation_qubit} can be further decomposed into Pauli terms $S_{c_i}^{-1} (t) P_{\vec a_i} S_{c_i}(t) = g_{x,i} (\tilde{\theta}) X + g_{y,i} (\tilde{\theta}) Y + g_{z,i} (\tilde{\theta}) Z$, with
\begin{equation}
\begin{aligned}
 g_{x,i} (\tilde{\theta}) &\coloneqq \begin{cases}
		\cos^2(\tilde{\theta}) + (p_{x,i}^2 - p_{y,i}^2 -p_{z,i}^2) \sin^2(\tilde{\theta})  \, , &\text{if } P_{\vec a_i} = X \\
		2 p_{x,i} p_{y,i} \sin^2(\tilde{\theta}) + 2 p_{z,i} (-1)^{s_i} \sin (\tilde{\theta}) \cos (\tilde{\theta}) \, , & \text{if }P_{\vec a_i} = Y \\
		2 p_{x,i} p_{z,i} \sin^2(\tilde{\theta}) - 2 p_{y,i} (-1)^{s_i} \sin (\tilde{\theta}) \cos (\tilde{\theta}) \, , & \text{if }P_{\vec a_i} = Z\, ,
	\end{cases} \\
 g_{y,i} (\tilde{\theta}) &\coloneqq \begin{cases}
		2 p_{x,i} p_{y,i} \sin^2(\tilde{\theta}) + 2 p_{z,i} (-1)^{s_i} \sin (\tilde{\theta}) \cos (\tilde{\theta}) \, , & \text{if }P_{\vec a_i} = X \\
		\cos^2(\tilde{\theta}) + (p_{y,i}^2 - p_{x,i}^2 -p_{z,i}^2) \sin^2(\tilde{\theta}) \, , & \text{if }P_{\vec a_i} = Y \\
		2 p_{y,i} p_{z,i} \sin^2(\tilde{\theta}) - 2 p_{x,i} (-1)^{s_i} \sin (\tilde{\theta}) \cos (\tilde{\theta}) \, , & \text{if }P_{\vec a_i} = Z\, ,
	\end{cases} \\
 g_{z,i} (\tilde{\theta}) &\coloneqq \begin{cases}
		2 p_{x,i} p_{z,i} \sin^2(\tilde{\theta}) - 2 p_{y,i} (-1)^{s_i} \sin (\tilde{\theta}) \cos (\tilde{\theta}) \, , & \text{if }P_{\vec a_i} = X \\
		2 p_{y,i} p_{z,i} \sin^2(\tilde{\theta}) + 2 p_{x,i} (-1)^{s_i} \sin (\tilde{\theta}) \cos (\tilde{\theta}) \, , & \text{if }P_{\vec a_i} = Y \\
		\cos^2(\tilde{\theta}) + (p_{z,i}^2 - p_{x,i}^2 -p_{y,i}^2) \sin^2(\tilde{\theta}) \, , & \text{if }P_{\vec a_i} = Z\, ,
	\end{cases}
\end{aligned}
\end{equation}
and if $P_{\vec a_i} = I$, then $S_{c_i}^{-1} (t) P_{\vec a_i} S_{c_i}(t) = I$.
Assume, that the interaction $J_{\vec a} P_{\vec a}$ is $k$-local, i.e.\ the interaction acts on the qubits $i \in \supp (\vec a)$ and $|\supp (\vec a)| = k$.
Applying $A \otimes (B + C) = A \otimes B + A \otimes C$ the conjugation in \cref{eq:L_conjugation} yields
\begin{equation}
\label{eq:extended_L_conjugation}
 \bigotimes_{i=1}^n S_{c_i}^{-1} (t) P_{\vec a_i} S_{c_i}(t)
 = \sum_{\substack{\tilde{\vec a} \in \FF_2^n \\ \supp (\tilde{\vec a}) = \supp (\vec a)}} \left(\prod_{i \in \tilde{\vec a}} g_{\tilde{\vec a}_i,i} (\tilde{\theta}) \right) P_{\tilde{\vec a}} \, ,
\end{equation}
where we identify $g_{(1,0),i} = g_{x,i}$, $g_{(1,1),i} = g_{y,i}$ and $g_{(0,1),i} = g_{z,i}$.
This sum has at most $3^k$ terms, which is constant for a system Hamiltonian $H_S$ with a fixed locality.

With that we are ready to compute $\EC{\dd}{\rr}_{\vec a \vec c}$ by calculating the Pauli coefficients of
\begin{equation}
\label{eq:Herr_general_robust_app}
 H_{\err, \vec c} = 2 \sum_{\kk=1}^S \vec S_{\vec c^{\geq (\kk + 1)}}^{-1} \int_0^{t_p^{(\kk)}} \vec S_{c^{(\kk)}}^{-1} (t) H_S \vec S_{c^{(\kk)}} (t) \rmd t \, \vec S_{\vec c^{\geq (\kk + 1)}} \, .
\end{equation}
We start with the integral
\begin{equation}
\begin{aligned}
\int_0^{t_p^{(\kk)}} \vec S_{c^{(\kk)}}^{-1} (t) H_S \vec S_{c^{(\kk)}} (t) \rmd t &= \sum_{\vec a \in \FF_2^{2n}\setminus \{ \vec 0 \}} J_{\vec a} \int_0^{t_p^{(\kk)}} \bigotimes_{i=1}^n S_{c_i^{(\kk)}}^{-1} (t) P_{\vec a_i} S_{c_i^{(\kk)}}(t) \rmd t \\
&= \sum_{\vec a \in \FF_2^{2n}\setminus \{ \vec 0 \}} \sum_{\substack{\tilde{\vec a} \in \FF_2^n \\ \supp (\tilde{\vec a}) = \supp (\vec a)}} J_{\vec a} \frac{t_p^{(\kk)}}{\theta^{(\kk)}} \int_0^{\theta^{(\kk)}} \left(\prod_{i \in \tilde{\vec a}} g_{\tilde{\vec a}_i,i} (\tilde{\theta}) \right) \rmd \tilde{\theta} P_{\tilde{\vec a}} \\
&\eqqcolon \sum_{\vec a \in \FF_2^{2n}\setminus \{ \vec 0 \}} E_{\vec a}^{(\kk)} P_{\vec a} \, ,
\end{aligned}
\end{equation}
where $E_{\vec a}^{(\kk)}$ can be efficiently calculated by integrating trigonometric polynoms and summing all contributions from the terms with the same locality $\supp (\tilde{\vec a}) = \supp (\vec a)$.
Next, the conjugation $\vec S_{\vec c^{\geq (\kk + 1)}}^{-1} (\cdot) \, \vec S_{\vec c^{\geq (\kk + 1)}}$ in \cref{eq:Herr_general_robust_app} corresponds to applying \cref{eq:extended_L_conjugation} for each $\tilde{\kk} = 1, \dots ,\kk+1$ to $E_{\vec a}^{(\kk)} P_{\vec a}$ with $t=t_p^{(\tilde{\kk})}$ or $\tilde{\theta} = \theta^{(\tilde{\kk})}$.
Then, we obtain
\begin{equation}
 H_{\err, \vec c} = \sum_{\vec a \in \FF_2^{2n}\setminus \{ \vec 0 \}} \EC{\dd}{\rr}_{\vec a \vec c} P_{\vec a} \, .
\end{equation}
The entries of $\WGG{\dd}{\rr}$ can be calculated similarly by applying \cref{eq:extended_L_conjugation} for each $\kk = 1, \dots ,\ns$ to $J_{\vec a} P_{\vec a}$ with $t=t_p^{(\kk)}$ or $\tilde{\theta} = \theta^{(\kk)}$.
Together, we obtain
\begin{equation}
 \vec S_{\vec c}^{-1} H_S \vec S_{\vec c} = \sum_{\vec a \in \FF_2^{2n}\setminus \{ \vec 0 \}} \WGG{\dd}{\rr}_{\vec a \vec c} P_{\vec a} \, .
\end{equation}
\end{proof}

\MainThm*

\begin{proof}
The chosen product formula determines the number of implemented single-qubit conjugations $U(t\lambda_{\vec c})$ for each $\vec c$.
Let this number be $n_{\vec c}$.
The finite pulse time error term in \eqref{eq:robustLP} has to be rescaled to account for $n_{\vec c}$ since each implementation of $U(t\lambda_{\vec c})$ causes the associated finite pulse time error.
This can be done by multiplying the $\vec c$-th column of $\EC{\dd}{\rr}$ by $n_{\vec c}$.

For simplicity, we assume the first order Trotter approximation in \cref{eq:first_order_trotter_sim}.
To account for the number of implemented single-qubit conjugations we have to rescale $\EC{\dd}{\rr}$ by $n_{\vec c} = \nTro$ for all $\vec c$.
Let $t H_{\av} \coloneqq \sum_{\vec c}^{\rr} H_{\av, \vec c} (t)$, with $\lambda_{\vec c}$ from \eqref{eq:robustLP}.
Let the target Hamiltonian be
\begin{equation}
 H_T = \sum_{\vec a \in \FF_2^{2n}\setminus \{ \vec 0 \}} A_{\vec a} P_{\vec a} \, .
\end{equation}
By the constraint of \eqref{eq:robustLP}, we have $H_{\av} = H_T$.
The time evolution governed by the target Hamiltonian $H_T$ can be approximated with
\begin{equation}
 \e^{-\i t H_T} = \e^{-\i t H_{\av}} \approx \left( \prodr_{\vec c} \e^{- \i H_{\av, \vec c} \left(\frac{t }{\nTro}\right) } \right)^{\nTro} \approx \left( \prodr_{\vec c} U(t \lambda_{\vec c}/\nTro) \right)^{\nTro} \, ,
\end{equation}
where the first approximation is given by the Trotter scheme and the second approximation is given by the first order Magnus expansion from \cref{lem:robust_finite_pulse_time}.
\end{proof}

\section{Proofs for the robust Pauli conjugation method}
\label{app:robust_pauli}

\LemHerr*
	
\begin{proof}
From \cref{lem:robust_finite_pulse_time} with $\ns=1$ and $H_{\vec c} = H_{c^{(1)}}$ we get the Hamiltonian corresponding to the finite pulse time effect
\begin{equation}
\label{eq:pauli_finite_pulse_error}
 H_{\err, \vec c} \coloneqq 2 \int_0^{t_p} \e^{ \i t H_{\vec c}} H_S \e^{- \i t H_{\vec c}} \rmd t \, .
\end{equation}
Recall, that the system Hamiltonian has the form
\begin{equation}
H_S = \sum_{\vec a \in \FF_2^{2n}\setminus \{ \vec 0 \}} J_{\vec a} P_{\vec a} \, .
\end{equation}
Before we compute $H_{\err, \vec c}$ we investigate the conjugation $\e^{ \i t H_{\vec c}} H_S \e^{- \i t H_{\vec c}}$ for a single-qubit.
A $\pi$ pulse on the $i$-th qubit can be written as
\begin{equation}
 S_{c_i} (t) \coloneqq \e^{- \i t \frac{\pi}{2 t_p} (-1)^{s_i} P_{\vec b_i}}
 = \cos \left(\frac{\pi}{2} \frac{t}{t_p} \right) I - \i (-1)^{s_i} P_{\vec b_i} \sin \left(\frac{\pi}{2} \frac{t}{t_p} \right)
 = \cos (\theta) I - \i (-1)^{s_i} P_{\vec b_i} \sin (\theta) \, ,
\end{equation}
with the Pauli generator $P_{\vec b}$ from \cref{eq:pauli_conjugation_system} and $\theta (t) \coloneqq \frac{\pi}{2} \frac{t}{t_p}$.
Then the effect of the conjugation on the $i$-th qubit is given by
\begin{equation}
\begin{aligned}
 S_{c_i}^{-1} (\theta) P_{\vec a_i} S_{c_i} (\theta) &= \underbrace{\left( \cos^2 (\theta) + (-1)^{\langle \vec a_i , \vec b_i \rangle} \sin^2 (\theta) \right)}_{\eqqcolon \alpha_{\vec a_i, \vec b_i} (\theta)} P_{\vec a_i} + (-1)^{s_i} \underbrace{\cos (\theta) \sin (\theta)}_{\beta (\theta)} \i [P_{\vec b_i}, P_{\vec a_i}] \\
 &= \alpha_{\vec a_i, \vec b_i} (\theta) P_{\vec a_i} + (-1)^{s_i} \beta (\theta) \i [P_{\vec b_i}, P_{\vec a_i}] \, ,
\end{aligned}
\end{equation}
with the binary symplectic form $\langle \argdot , \argdot \rangle$ from \cref{eq:pauli_conjugation} and the commutator $[\argdot , \argdot]$.
Let
\begin{equation}
 F(e) \coloneqq \begin{cases}
		P_{\vec a} \alpha_{\vec a_i, \vec b_i} (\theta) \, , & e = 0 \\
		\i [P_{\vec b_i}, P_{\vec a_i}] \beta (\theta) \, , & e = 1 \, .
	\end{cases}
\end{equation}
Inserting $H_S$ in \cref{eq:pauli_finite_pulse_error} and applying $A \otimes (B + C) = A \otimes B + A \otimes C$ yields
\begin{equation}
\label{eq:robust_pauli_Herr_calc}
\begin{aligned}
 H_{\err, \vec c} &= \frac{4 t_p}{\pi} \sum_{\vec a \in \FF_2^{2n}\setminus \{ \vec 0 \}} J_{\vec a} \int_0^{\frac{\pi}{2}} \left( \bigotimes_{i \in \supp (\vec a)} \left( P_{\vec a_i} \alpha_{\vec a_i, \vec b_i} (\theta) + \i [P_{\vec b_i}, P_{\vec a_i}] \beta (\theta) (-1)^{s_i} \right) \right) \rmd \theta \\
 &= \frac{4 t_p}{\pi} \sum_{\vec a \in \FF_2^{2n}\setminus \{ \vec 0 \}} \left( J_{\vec a} \int_0^{\frac{\pi}{2}}  \prod_{i \in \supp (\vec a)} \alpha_{\vec a_i, \vec b_i} (\theta)  \rmd \theta P_{\vec a} + \sum_{\substack{\vec e \in \FF_2^{n}\setminus \{ \vec 0 \} \\ e_i = 0 \, \forall i \notin \supp (\vec a)}} \left( \prod_{i \in \supp (\vec a)} (-1)^{e_i+s_i} \right) \left( \bigotimes_{i \in \supp (\vec a)} F(e_i) \right) \right) \\
 &= \sum_{\vec a \in \FF_2^{2n}\setminus \{ \vec 0 \}} J_{\vec a} \EP{\dd}{\rr}_{\vec a, \vec c} P_{\vec a} + \sum_{\substack{\vec e \in \FF_2^{n}\setminus \{ \vec 0 \} \\ e_i = 0 \, \forall i \notin \supp (\vec a)}} (-1)^{\vec e \cdot \vec s} \bigotimes_{i \in \supp (\vec a)} F(e_i) \, .
\end{aligned}
\end{equation}
Moreover, we define $\EP{\dd}{\rr}_{\vec a, \vec c} \coloneqq \frac{4 t_p}{\pi} \int_0^{\frac{\pi}{2}} \left( \prod_{i \in \supp (\vec a)} \alpha_{\vec a_i, \vec b_i} (\theta) \right) \rmd \theta$.
To conclude, we have the average Hamiltonian
\begin{equation}
  H_{\av, \vec c} (t) = t \lambda_{\vec c} \vec S_{\vec c}^{\dagger} H_S \vec S_{\vec c} + H_{\err, \vec c} \, ,
\end{equation}
from \cref{lem:robust_finite_pulse_time}, and we calculated the finite pulse time error term
\begin{equation}
 H_{\err, \vec c} = \sum_{\vec a \in \FF_2^{2n}\setminus \{ \vec 0 \}} \left( J_{\vec a} \EP{\dd}{\rr}_{\vec a, \vec c} P_{\vec a} + R_{\vec a , \vec c}\right) \, ,
\end{equation}
with
 \begin{equation}
 \EP{\dd}{\rr}_{\vec a, \vec c} \coloneqq \frac{4 t_p}{\pi} \int_0^{\frac{\pi}{2}} \left( \prod_{i \in \supp (\vec a)} (\cos^2 (\theta) + (-1)^{\langle \vec a_i , \vec b_i \rangle} \sin^2 (\theta)) \right) \rmd \theta
\end{equation}
and
\begin{equation}
 R_{\vec a , \vec c} \coloneqq \sum_{\substack{\vec e \in \FF_2^{n}\setminus \{ \vec 0 \} \\ e_i = 0 \, \forall i \notin \supp (\vec a)}} (-1)^{\vec e \cdot \vec s} \bigotimes_{i \in \supp (\vec a)} F(e_i) \, ,
\end{equation}
which proofs the lemma.
\end{proof}

\MainProp*

\begin{proof}
Similar as in the proof of \cref{thrm:robust_finite_pulse_time} the columns of $\EP{\dd}{\rr}$ in \eqref{eq:robustPauliLP} has to be rescaled by the number of implementations $U(t\lambda_{\vec c})$ which we denote by $n_{\vec c}$.
For simplicity we assume the first order Trotter approximation in \cref{eq:first_order_trotter_sim}.

We start with the rest term in the average Hamiltonian in \cref{eq:pauli_finite_pulse} which is proportional to $(-1)^{\vec e_{\vec a} \cdot \vec s}$.
Let's assume we require $\pp$ different rotation directions $\vec s_{(1)}, \dots , \vec s_{(\pp)} \in \FF_2^n$
such that $\sum_{j=1}^{\pp} (-1)^{\vec e_{\vec a} \cdot \vec s_{(j)}} = 0$ for all $\vec e_{\vec a} \in \FF_2^{n}$ with $e_{\vec a, i} = 0$ for all $i \notin \supp (\vec a)$ for any $\vec a$ with $J_{\vec a} \neq 0$.
A non-optimal choice of $\vec s_{(j)}$ are all possible binary combinations without the all-zero vector, then $\pp = 2^n - 1$.
An efficient choice is provided in \cref{prop:rot_dir_choice} below.
To indicate the used rotation direction, we specify the single-qubit Pauli pulse layer by the tuple $\vec c_{(j)} = (t_p, \vec s_{(j)}, \vec h)$.
Note, that the generators do not change.
Let the partial average Hamiltonian be the average Hamiltonian without the rest error term
\begin{equation}
\begin{aligned}
  \tilde{H}_{\av, \vec c} (t) &\coloneqq \sum_{j=1}^{\pp} H_{\av, \vec c_{(j)}} \left(\frac{t}{\pp} \right) \\
  &= t \lambda_{\vec c} \vec S_{\vec c}^\dagger H_S \vec S_{\vec c} + \pp \sum_{\vec a \in \FF_2^{2n}\setminus \{ \vec 0 \}} \EP{\dd}{\rr}_{\vec a, \vec c} J_{\vec a} P_{\vec a} \, .
\end{aligned}
\end{equation}
The evolution under the partial average Hamiltonian can be approximated by
\begin{equation}
 \e^{-\i t \tilde{H}_{\av, \vec c}(t)} \approx \prodr_{j=1}^{\pp} \e^{-\i H_{\av, \vec c_{(j)}} \left(\frac{t}{\pp} \right)} \approx \prodr_{j=1}^{\pp} U(t\lambda_{\vec c_{(j)}}/\pp)  \, ,
\end{equation}
where the first approximation is given by a first order Trotter scheme where we set $\nTro = 1$ for simplicity, and the second approximation is given by the first order Magnus expansion from \cref{lem:robust_finite_pulse_time}.
Let $t H_{\av} \coloneqq \sum_{\vec c}^{\rr} \tilde{H}_{\av, \vec c} (t)$, with $\lambda_{\vec c}$ from \eqref{eq:robustPauliLP}.
Let the target Hamiltonian be
\begin{equation}
 H_T = \sum_{\vec a \in \FF_2^{2n}\setminus \{ \vec 0 \}} A_{\vec a} P_{\vec a} \, .
\end{equation}
By the constraint of \eqref{eq:robustPauliLP}, we have $H_{\av} = H_T$.
The time evolution governed by the target Hamiltonian $H_T$ can be approximated with
\begin{equation}
 \e^{-\i t H_T} = \e^{-\i t H_{\av}} \approx \left( \prodr_{\vec c} \e^{- \i \tilde{H}_{\av, \vec c_{(j)}} \left(\frac{t}{\nTro} \right)} \right)^{\nTro} \approx \left( \prodr_{\vec c} \prodr_{j=1}^{\pp} U\left(\frac{t\lambda_{\vec c_{(j)}}}{\pp \nTro}\right) \right)^{\nTro} \, ,
\end{equation}
where the first approximation is given by the Trotter scheme.
To account for the number of implemented Pauli conjugations $U(t\lambda_{\vec c})$ we have to rescale $\EP{\dd}{\rr}$ by $n_{\vec c} = \pp \nTro$ for all $\vec c$.
\end{proof}

\begin{restatable}{restatableProp}{PropPiRotations}
\label{prop:rotation_err_robust}
Let the rotation directions $\vec s \in \FF_2^n$ of the $\pi$ pulses such that $\sum_{\vec s} (-1)^{\vec e_{\vec a} \cdot \vec s} = 0$ for all $\vec e_{\vec a} \in \FF_2^{n}$ with $e_{\vec a, i} = 0$ for all $i \notin \supp (\vec a)$ for any $\vec a$ with $J_{\vec a} \neq 0$, as in \cref{lem:pauli_finite_pulse_error}.
Then, the rotation angle errors in the first order Taylor approximation cancels.
\end{restatable}
\begin{proof}
We model the rotation angle error of a perfect $\pi$ pulse $\e^{- \i (-1)^{s_i} \frac{\pi}{2} P_{\vec b_i}}$ with the rotation direction $(-1)^{s_i}$ and $s_i \in \FF_2$ by
\begin{equation}
\begin{aligned}
 \tilde{S}_{c_i} =
  \e^{- \i (-1)^{s_i} \frac{\pi + \varepsilon_i}{2} P_{\vec b_i}} = -\frac{\varepsilon_i}{2} \Id - \i (-1)^{s_i} P_{\vec b_i} + \LandauO (\varepsilon_i^2) \, ,
\end{aligned}
\end{equation}
where we used the first-order Taylor expansion of sine and cosine.
Conjugating a Pauli operator with an imperfect Pauli pulse results in
\begin{equation}
 \tilde{S}_{c_i}^{\dagger} P_{\vec a_i} \tilde{S}_{c_i} = P_{\vec b_i} P_{\vec a_i} P_{\vec b_i} + \i (-1)^{s_i} \frac{\varepsilon_i}{2} [P_{\vec a_i}, P_{\vec b_i}] + \LandauO (\varepsilon_i^2) \, .
\end{equation}
We use the identity $U^\dagger \e^{-\i t H} U = \e^{-\i t U^\dagger H U}$ to compute the effective Hamiltonian of a system Hamiltonian $H_S$ conjugated with a layer of imperfect Pauli pulses $\tilde{\vec S}_{\vec c} = \bigotimes_{i=1}^n \tilde{S}_{c_i}$,
\begin{equation}
\label{eq:app_angle_error_cancel}
\begin{aligned}
 \tilde{\vec S}_{\vec c}^\dagger H_S \tilde{\vec S}_{\vec c} &= \sum_{\vec a \in \FF_2^{2n}\setminus \{ \vec 0 \}} J_{\vec a}  \bigotimes_{i=1}^n \left( \tilde{S}_{c_i}^{\dagger} P_{\vec a_i} \tilde{S}_{c_i} \right) \\
 &= \sum_{\vec a \in \FF_2^{2n}\setminus \{ \vec 0 \}} J_{\vec a}  \bigotimes_{i=1}^n \left( P_{\vec b_i} P_{\vec a_i} P_{\vec b_i} + \i (-1)^{s_i} \frac{\varepsilon_i}{2} [P_{\vec a_i}, P_{\vec b_i}] + \LandauO (\varepsilon_i^2) \right) \\
 &= \sum_{\vec a \in \FF_2^{2n}\setminus \{ \vec 0 \}} \left( J_{\vec a} P_{\vec b} P_{\vec a} P_{\vec b} + \sum_{\substack{\vec e \in \FF_2^{n} \\ e_i = 0 \, \forall i \notin \supp (\vec a)}} \left( \prod_{i \in \supp (\vec a)} (-1)^{e_i+s_i} \varepsilon_i^{e_i} \right) \left( \bigotimes_{i \in \supp (\vec a)} F(e_i) + \LandauO (\varepsilon_i^2) \right) \right) \, ,
\end{aligned}
\end{equation}
where we used $A \otimes (B + C) = A \otimes B + A \otimes C$ and
\begin{equation}
 F(e) \coloneqq \begin{cases}
		P_{\vec b_i} P_{\vec a_i} P_{\vec b_i} \, , & e = 0 \\
		\i \frac{1}{2} [P_{\vec a_i}, P_{\vec b_i}] \, , & e = 1 \, .
	\end{cases}
\end{equation}
The second term in the sum is the first-order angle error contribution, and has the same sign structure $(-1)^{e_i+s_i}$ as the rest term $R_{\vec a, \vec c}$ in \cref{eq:robust_pauli_Herr_calc}.
Therefore, choosing the signs $\vec s \in \FF_2^n$, such that the rest term $R_{\vec a, \vec c}$ cancels, simultaneously cancels the first order angle error contribution in \cref{eq:app_angle_error_cancel}.
\end{proof}

The proof of the following result is similar as the proof in \cite[Lemma 8]{basler_time-optimal_2024}.

\PropPiRotationsChoice*

\begin{proof}
From the choice of $(-1)^{\vec s_{(j)}}$ it directly follows that $\sum_{j=1}^{\pp} (-1)^{\vec e \cdot \vec s_{(j)}} = 0$ for all $\vec e \in \FF_2^{n}$ with $| \vec e | = 1$, since the sum over all rows of a Walsh-Hadamard matrix especially $\WP{\pp}{n}$ is zero.
We now have to show that $\sum_{j=1}^{\pp} (-1)^{\vec e \cdot \vec s_{(j)}} = 0$ for all $\vec e \in \FF_2^{n}$ with $| \vec e | = 2$.
For any $\vec e \in \FF_2^{n}$ with $| \vec e | = 2$ we can write $(-1)^{\vec e \cdot \vec s} = (-1)^{s_i} (-1)^{s_k} = \left(\left((-1)^{\vec s}\right) \left( (-1)^{\vec s}\right)^T\right)_{ik}$ with $i,k \in [n]$ and $i \neq k$.
Then, we obtain
\begin{equation}
\label{eq:sum_e_s}
 \sum_{j=1}^{\pp} (-1)^{\vec e \cdot \vec s_{(j)}} = \sum_{j=1}^{\pp} \left( \left((-1)^{\vec s_{(j)}}\right) \left( (-1)^{\vec s_{(j)}}\right)^T \right)_{ik} \, ,
\end{equation}
with $i,k \in [n]$ and $i \neq k$.
The orthogonality property of the Walsh-Hadamard matrix yields
\begin{equation}
 \sum_{j=1}^{\pp} \left((-1)^{\vec s_{(j)}}\right) \left( (-1)^{\vec s_{(j)}}\right)^T =  (\WP{\pp}{n})^T \WP{\pp}{n} = \pp I \, ,
\end{equation}
where the non-diagonal entries $i \neq k$ correspond to the sum \cref{eq:sum_e_s} and are zero.
\end{proof}


\begin{acronym}[MAGIC]\itemsep.5\baselineskip
\acro{AGF}{average gate fidelity}
\acro{AHT}{average Hamiltonian theory}
\acro{AQFT}{approximate quantum Fourier transform}

\acro{BOG}{binned outcome generation}

\acro{CP}{completely positive}
\acro{CPT}{completely positive and trace preserving}
\acro{CS}{compressed sensing} 

\acro{DAQC}{digital-analog quantum computing}
\acro{DFE}{direct fidelity estimation} 
\acro{DM}{dark matter}
\acro{DD}{dynamical decoupling}

\acro{EF}{extended formulation}

\acro{GST}{gate set tomography}
\acro{GUE}{Gaussian unitary ensemble}

\acro{HOG}{heavy outcome generation}

\acro{KKT}{Karush–Kuhn–Tucker}

\acro{LDP}{Lamb-Dicke parameter}
\acro{LP}{linear program}

\acro{MAGIC}{magnetic gradient induced coupling}
\acro{MBL}{many-body localization}
\acro{MIP}{mixed integer program}
\acro{MILP}{mixed integer linear program}
\acro{ML}{machine learning}
\acro{MLE}{maximum likelihood estimation}
\acro{MPO}{matrix product operator}
\acro{MPS}{matrix product state}
\acro{MS}{M{\o}lmer-S{\o}rensen}
\acro{MUBs}{mutually unbiased bases} 
\acro{mw}{micro wave}

\acro{NISQ}{noisy and intermediate scale quantum}
\acro{NMR}{nuclear magnetic resonance}

\acro{POVM}{positive operator valued measure}
\acro{PVM}{projector-valued measure}

\acro{QAOA}{quantum approximate optimization algorithm}
\acro{QFT}{quantum Fourier transform}
\acro{QML}{quantum machine learning}
\acro{QMT}{measurement tomography}
\acro{QPT}{quantum process tomography}
\acro{QPU}{quantum processing unit}

\acro{RDM}{reduced density matrix}

\acro{SFE}{shadow fidelity estimation}
\acro{SIC}{symmetric, informationally complete}
\acro{SPAM}{state preparation and measurement}

\acro{RB}{randomized benchmarking}
\acro{rf}{radio frequency}
\acro{RIC}{restricted isometry constant}
\acro{RIP}{restricted isometry property}
\acro{RSD}{relative standard deviation}

\acro{TT}{tensor train}
\acro{TV}{total variation}

\acro{VQA}{variational quantum algorithm}

\acro{VQE}{variational quantum eigensolver}

\acro{XEB}{cross-entropy benchmarking}

\end{acronym}


\twocolumngrid
\bibliographystyle{./myapsrev4-2}
\bibliography{./mk}

\end{document}